\documentclass[10pt,a4paper]{article}
\usepackage[utf8]{inputenc}
\usepackage[normalem]{ulem}
\usepackage{amsmath, amsthm,mathrsfs}
\usepackage{amsfonts}
\usepackage{amssymb}
\usepackage{bbm}
\usepackage{cite}
\newcommand{\DD}{\mathcal{D}}
\usepackage{hyperref, todonotes, mathtools}	
\allowdisplaybreaks

\mathtoolsset{showonlyrefs}

\newcommand{\tuI}{\tilde{u}_I}
\newcommand{\Par}{\lambda}
\newcommand{\ce}{\Upsilon_{\ep}}

\newcommand{\p}{\partial}
\newcommand{\ep}{\epsilon}

\newcommand{\rs}{r_{*}}
\newcommand{\duh}{\delta \check{u}_H}

\newcommand{\uhs}{ u_H^{*}}
\newcommand{\rd}{\partial}
\newcommand{\Ai}{\text{Ai}}
\newcommand{\Bi}{\text{Bi}}

\newtheorem{theo}{Theorem}

\usepackage{slashed}
\usepackage[left=2cm,right=2cm,top=2cm,bottom=1.3cm]{geometry}

\usepackage{authblk}
\usepackage{enumerate}

\usepackage{enumerate}

\theoremstyle{plain}
\newtheorem{thm}{Theorem}[section]

\newtheorem{lemma}[thm]{Lemma}

\newtheorem{prop}[thm]{Proposition}

\newtheorem{cor}[thm]{Corollary}

\theoremstyle{remark}
\newtheorem{rmk}{Remark}[section]

\theoremstyle{definition}

\newcommand{\RR}{\mathbb{R}}

\newcommand{\ls}{\lesssim}

\newcommand{\nofa}{d}

\theoremstyle{plain}

\theoremstyle{remark}

\theoremstyle{definition}

\usepackage{bm}

\usepackage{color}

\usepackage{float}

\addtocounter{tocdepth}{-2}
\usepackage{graphicx}
\usepackage{authblk}
 
\numberwithin{equation}{section}

\begin{document}
	\date{}
	\title{The asymptotics of massive fields on stationary\\ spherically symmetric black holes  for all angular momenta} 
	
	\author[1]{Federico Pasqualotto}
	\author[2]{Yakov Shlapentokh-Rothman}
	\author[3]{Maxime Van de Moortel}

	\affil[1]{\small Department of Mathematics, University of California, Berkeley,  970 Evans Hall, Berkeley, CA, USA}
	\affil[2]{\small  Department of Mathematics, University of Toronto,  40 St.~George~Street, Toronto, ON, Canada} 
	\affil[2]{\small Department of Mathematical and Computational Sciences, 
		University of Toronto Mississauga, 3359 Mississauga Road, Mississauga, ON, Canada}
	\affil[3]{\small Department of Mathematics, Rutgers University, Hill Center-Busch Campus, 110 Frelinghuysen Road, Piscataway, NJ, USA }
	
	\maketitle

	\abstract{
		We study the  massive scalar field equation $\Box_g \phi = m^2 \phi$ on a  stationary and spherically symmetric black hole $g$ (including in particular the Schwarzschild and Reissner--Nordstr\"om   black holes in the full sub-extremal range) for solutions $\phi$ projected on a  fixed spherical harmonic. Our problem involves the scattering of an attractive  long-range potential (Coulomb-like)  and thus cannot be treated perturbatively. 
		
		We prove precise (point-wise) asymptotic tails of the form $t^{-\frac{5}{6}} f(t)+ O(t^{-1+\delta})$, where $f(t)$ is an explicit oscillating profile. Our asymptotics appear to be the first rigorous decay result  for a massive scalar field on a black hole. Establishing these asymptotics  is also  an important step in retrieving the assumptions used in work of the third author regarding the interior of dynamical black holes and Strong Cosmic Censorship.

	}

	\section{Introduction}
	We study the asymptotic properties for large times of the massive Klein--Gordon equation  	\begin{equation}\label{KG.intro} \begin{split}
			& \Box_g \phi= m^2 \phi,\\ &
			g= -\Big(1-\frac{2M}{r}+\frac{e^2}{r^2}+O(r^{-3})\Big) dt^2 +  \Big(1-\frac{2M}{r}+\frac{e^2}{r^2}+O(r^{-3})\Big)^{-1} r^2 + r^2 \left( d\theta^2+ \sin^2(\theta) d\varphi^2\right).
		\end{split}
	\end{equation} where $g$ is the general expression of a spherically symmetric and stationary black hole, which includes in particular the  Schwarzschild  ($e=0$) and the subextremal ($0<|e|<M)$ Reissner--Nordstr\"om black hole.

	The study of the asymptotic properties of matter fields\footnote{This subject is also naturally related to the problem of black hole stability, with matter fields or in vacuum. We refer to \cite{Andras,MihalisStabExt,Wang,KS.polarized,SchwarzschildStab,Blue,HHV,Giorgi.RN,SRTdC2020boundedness,SRTdC2023boundedness2,millet,klainerman2021kerr,GKS} for the most up-to-date mathematical works.} on black holes is a classical subject, which can be traced back to the work of Price \cite{Pricepaper} in the 1970s on the decay of massless scalar fields, i.e.\ \eqref{KG.intro} with $m^2=0$  (see also \cite{KonoplyaZhidenko} for a more extended review of the physics literature). 
	In the massless case $m^2=0$, the numerics of Price \cite{Pricepaper} predicted a late-time tail with inverse-power-law decay of the form $t^{-3}$ at large times $t$ (Price's law). Price's law was later proven by Dafermos--Rodnianski \cite{PriceLaw} (for spherically symmetric Einstein-scalar-field equations), Donninger--Schlag--Soffer \cite{Schlag2,Schlag1} (on a Schwarzschild black hole), Tataru \cite{Tataru} (on a Kerr black hole) and later with Metcalfe and Tohaneanu \cite{Tataru2} (on certain non-stationary black holes), Hintz \cite{Hintz} (precise late-time tails on a Schwarzschild/Reissner--Nordstr\"om black hole for the solution projected to $L\geq L_0$ angular modes  for any non-negative integer $L_0$, on a Kerr black hole for the whole solution ), and Angelopoulos--Aretakis--Gajic \cite{AAG1,AAG2,AAG3} (precise late-time tails on a Schwarzschild/Reissner--Nordstr\"om black hole  for the solution projected to  $L\geq L_0$ angular modes  for any non-negative integer $L_0$,  on a Kerr black hole for the solution projected to $L=0,\ L=1, or\ L\geq 2$ angular modes).  We also mention \cite{Ma1,Ma3,Ma4,Ma5} for the analogue of Price's law for wave equations with non-zero spins. 
	
	In the massive case (i.e.\ \eqref{KG.intro} with $m^2>0$, arguably the simplest model for massive matter), no decay result	 have been rigorously established in any black hole setting, despite an abundance of works \cite{Burko,BurkoKhanna,HodPiran.mass,KonoplyaZhidenko,Konoplya.Zhidenko.num,KoyamaTomimatsu,KoyamaTomimatsu3,KoyamaTomimatsu2,Dilaton1,Dilaton2} in the physics literature. For $m^2>0$, a $t^{-\frac{5}{6}} \sin(m t+O(t^{\frac{1}{3}}))$ asymptotic tail was predicted \cite{KoyamaTomimatsu,KoyamaTomimatsu3,KoyamaTomimatsu2} 
	using heuristic arguments (see also the numerics of \cite{BurkoKhanna,Konoplya.Zhidenko.num}). Compared to the Price's law tail $t^{-3}$ prevailing in the massless case, the massive tail appears -- in drastic contrast -- to be  slower, and oscillates. 
	
	We prove that for each fixed angular frequency $L$ (including the case $L=0$ of spherical symmetry), solutions of \eqref{KG.intro} with (sufficiently regular and localized) initial data indeed decay at $t^{-\frac{5}{6}}$ and oscillate, although we do not exactly obtain the $t^{-\frac{5}{6}} \sin(m t+O(t^{\frac{1}{3}}))$ profile predicted in \cite{Konoplya.Zhidenko.num,KoyamaTomimatsu,KoyamaTomimatsu3,KoyamaTomimatsu2} due to the influence of infinitely many sublinear oscillating corrections. Our result includes precise late-time tails, and applies to the Schwarzschild  ($e=0$) and Reissner--Nordstr\"om  ($0<|e|<M)$   black hole \emph{in the full sub-extremal range}.
	
	\begin{theo} \label{thm.intro}  There exists $r_+(M,e)>0$ (the radius of the event horizon),  $\tilde{C}(M,e,m^2)>0$,  $\gamma_m(M,e,m^2) \neq 0$ with  $|\gamma_m|<1$,  a linear form  $\mathcal{L}[\phi_0,\phi_1]=|\mathcal{L}|[\phi_0,\phi_1]e^{i \theta[\phi_0,\phi_1]}$, $\Psi_q(t^{*})$  and $\check{\varphi}_H(r)$ such that for any  $\delta \in (0,1)$, $R_0 > r_+$, the solution $\phi=\sum_{L=0}^{L_{max}}  r^{-1}\psi_{L}(t,r) Y_{L}(\theta,\varphi)$ of \eqref{KG.intro}  with smooth compactly supported initial data $(\phi_0,\phi_1)$ at $\{t^{*}=0\}$, where $t^{*}$ is a time-coordinate, satisfies for all $r_+\leq r \leq R_0$, $t^{*}\geq 1$:
		\begin{align}\label{decay.intro}\begin{split}
				&\psi_{L}(t^{*},r)= \tilde{C}(M,e,m^2)   \bigl| \mathcal{L}[\phi_0,\phi_1] \bigr|  |\uhs|(r)\  (t^{*})^{-\frac{5}{6}} \left[\sum_{q=1}^{+\infty} (-\gamma_m)^{q-1}q^{\frac{1}{3}}\cos(\Psi_q(t^{*})+  \check{\varphi}_H(r)+\theta[\phi_0,\phi_1]) \right]\\ & +O_{R_0,L_{max}}(  (t^{*})^{-1+\delta}),\\ & \Psi_q(t^{*})= mt^{*}  -\frac{3}{2} [2\pi M]^{\frac{2}{3}} m q^{\frac{2}{3}} (t^{*})^{\frac{1}{3}}  + \phi_-(M,e,m^2) -\frac{\pi}{4} +O_{m}((t^{*})^{-\frac{1}{3}}),\\ & u_H^{*}(r)=|u_H^{*}|(r) e^{ i \check{\varphi}_H(r)} \sim 1 \text{ as } r \rightarrow r_+,\    |u_H^{*}|(r) \ls r^{\frac{1}{4}},
				\text{ for } r \geq r_+. \end{split}
		\end{align} 
	\end{theo}\begin{rmk} The coordinate  $t^{*}$, defined as  $t^{*}(t,r) = t+p(r)$, with $p(r) \sim \log(r-r_+)$ as $r\rightarrow r_+$ and $p(r)=O(r^{-1})$ as $r \rightarrow+\infty$, induces an asymptotically flat spacelike foliation $\Sigma_{t^{*}}$ (see Section~\ref{setup.section}).	Theorem~\ref{thm.intro}  expresses decay in $(t^{*},r)$ coordinates to encompass both decay on any  $\{r=r_0\}$ hypersurface for $r_0>r_+$ and decay on the event horizon $\mathcal{H}=\{r=r_+\}$. Note indeed that, if $r=r_0>r_+$ is fixed, then $t^{*} \sim t $ are comparable, hence \eqref{decay.intro} indeed shows that the decay is of order $t^{-\frac{5}{6}}$. On $\mathcal{H}=\{r=r_+\}$, however, $t=+\infty$, but $t^{*}<+\infty$  is an adequate advanced-time variable which   parametrizes $\mathcal{H}$.
	\end{rmk}
	\begin{rmk}  \label{bad.decay.rmk}
		While we prove Theorem~\ref{thm.intro} assuming smooth and compactly supported initial data (note that we allow the compact support to contain the event horizon $\mathcal{H}=\{r=r_+\}\times \mathbb{S}^2$), it is not difficult using our methods to generalize to data in some Sobolev space 	with polynomial weights in $r$. One can also easily show that the  $O_{R_0}((t^{*})^{-1+\delta})$ error in \eqref{decay.intro} grows at worse polynomially in $r$ (e.g.\ one can obtain $O(r^{3.25} (t^{*})^{-1+\delta})$ for large $r$). We leave to future works the precise rates in  $(r,t^{*})$ when $r$, $t^{*}$ are both large. Note, however,  that $|\phi|(t^{*},r) \ls r^{-1}$ (with no decay in $t^{*}$) by energy estimate, thus one can  interpolate  between $|\phi|(t^{*},r) \ls r^{-1}$ and \eqref{decay.intro} and get a polynomial decay estimate  $|\phi|(t^{*},r) \ls~ (t^{*})^{-s}$ uniformly in $r$ for some $0<s<\frac{5}{6}$.
	\end{rmk}\begin{rmk}
		As we will see, the function $u_H^{*}=|u_H^{*}| e^{ i \check{\varphi}_H(r)}$ -- which plays an important role in \eqref{decay.intro} -- is explicit: it is the unique solution which is regular at the event horizon $\{r=r_+\}$ at the frequency $\omega=\pm m$.
	\end{rmk}\begin{rmk} Note that the main $(t^{*})^{-\frac{5}{6}}$ term in \eqref{decay.intro} is independent of  $L$; however the error $O((t^{*})^{-1+\delta})$ depends on  $L$ and is not summable over infinitely many angular modes $L$. Thus, the restriction to finitely many angular modes of the initial data (thus of the solution) is crucial, see Section~\ref{massive.section}.  However, Theorem~\ref{thm.intro} combined with the boundedness of the energy also proves the \emph{decay without-a-rate} of $\phi$ in the general unbounded-$L$ case, i.e.\ $\lim_{t^{*}\rightarrow +\infty} \phi(r,t^{*},\theta,\varphi)=0$ for any fixed $r \geq r_+$, via a soft interpolation argument. 
	\end{rmk}
	
	\begin{rmk}
		The analysis leading to Theorem~\ref{thm.intro} is a building block for our subsequent work \cite{KGNew} in which we study solutions of \eqref{KG.intro} \emph{without any fixed angular mode restriction } and obtain uniform-in-$L$ estimates.
	\end{rmk}
	
	It follows from Theorem~\ref{thm.intro} that the decay rate $(t^{*})^{-\frac{5}{6}}$ is \emph{sharp for generic solutions}: as expected, there exists an hyperplane of initial data $(\phi_0,\phi_1)$ leading to faster -- at least $(t^{*})^{-1+\delta}$ -- decay (corresponding to the kernel of the linear form $\mathcal{L}^{*}(\phi_0,\phi_1)$).   For generic solutions (thus for initial data $(\phi_0,\phi_1)$ not in the kernel of $\mathcal{L}^{*}$),  it is moreover easy to derive  the following $L^1$ or $L^2$ lower bounds for fixed $r \ge r_+$: \begin{equation}\label{lower.bounds.intro}
		\int_1^{t^{*}} |\phi|(t',r) dt'  \gtrsim_r |\mathcal{L}(\phi_0,\phi_1)|\  (t^{*})^{\frac{1}{6}} \text{ and } \int_{t^{*}}^{+\infty} |\phi|^2(t',r) dt'  \gtrsim_r |\mathcal{L}(\phi_0,\phi_1)|^2\  (t^{*})^{-\frac{2}{3}}.
	\end{equation}
	
	Expanding \eqref{KG.intro}, we see that it is analogous to a long-range scattering problem, in the sense that \eqref{KG.intro} looks  like a wave equation with a Coulomb-type $O(r^{-1})$  potential: namely, for large enough $r$: \begin{equation}\label{V.intro}
		-\rd_{t t}^2 \phi + \Delta_{\RR^3} \phi = m^2 \left(1- \frac{2M}{r} + O(r^{-2}) \right) \phi.
	\end{equation} It is well-known that such potentials cannot be treated perturbatively, and in fact we do not know of any result showing asymptotic decay in time for analogous problems: this is in particular because such potentials admit \emph{bound states} in general. The two important features of the potential $V(r)= m^2 \left(1- \frac{2M}{r} + O(r^{-2}) \right)$  in the simplified equation \eqref{V.intro} (and also in \eqref{KG.intro}) are as follows: \begin{enumerate}[I.]
		\item \label{VI}$V(r)$ is positive and time-independent: by the energy estimate, we show that there are \emph{no bound states}. This aspect of our argument crucially relies on the existence of a black hole region.
		\item  \label{VII}$V(r)$ admits a \emph{turning point}: after taking the Fourier transform of $\phi$, it occurs when $V(r_c(\omega))-\omega^2=0$, i.e.\ $r_c(\omega)\approx \frac{2Mm^2}{m^2-\omega^2}$. As $\omega \rightarrow \pm m$, $r_c(\omega) \rightarrow+\infty$: the turning point is pushed to infinity.
		
	\end{enumerate} In fact, we will show that Theorem~\ref{thm.intro} still holds true if one adds a perturbation $G(r)=O(\log(r) r^{-2})$ to $V(r)$ (ensuring \eqref{VII} is satisfied),  provided that  the perturbed potential still satisfies assumption \eqref{VI}, in addition to regularity near the event horizon $r=r_+(M,e)$ (see Section~\ref{main.result.section}).
	The origin of the $(t^{*})^{-\frac{5}{6}}$ decay rate is explained in our proof by a branch-cut type singularity of the resolvent at frequencies $\omega \in  (-m,m)$. Specifically we will show that the resolvent $\mathcal{R}(\omega,r)$  for fixed $\omega \in (-m,m)$ admits the schematic form (up to errors that we do not write, see Section~\ref{proof.section}) with infinite oscillations as $\omega \rightarrow \pm m$:  for any $\omega \in (-m,m)$ \begin{equation}\label{resolvent.intro}
		\mathcal{R}(\omega,r) =    \frac{ \Gamma_+\ e^{i \pi  M m^2 (m^2-\omega^2)^{-\frac{1}{2}}} +  \Gamma_-\ e^{-i \pi  M m^2 (m^2-\omega^2)^{-\frac{1}{2}}}}{\bar{\Gamma}_+\ e^{-i \pi  M m^2 (m^2-\omega^2)^{-\frac{1}{2}}} +  \bar{\Gamma}_-\ e^{i \pi  M m^2 (m^2-\omega^2)^{-\frac{1}{2}}}}\   F(r) + \text{ errors}.
	\end{equation}  Here, $\Gamma_{\pm}$ are constants and $F(r)$ is a  function of $r$ only.  The above form in \eqref{resolvent.intro} is a result of the existence of the turning point (see Section~\ref{proof.section}). The $(t^{*})^{-\frac{5}{6}}$ rate then emerges after taking the inverse Fourier transform of the resolvent $\int_{-\infty}^{+\infty} e^{- i \omega t}  \mathcal{R}(\omega) d\omega $ because of the existence of a critical point $\omega_c(t^{*})= \pm m + O ((t^{*})^{-\frac{2}{3}}) $ of the phase $\omega \rightarrow -\omega t+ (m^2-\omega^2)^{-\frac{1}{2}}$, after a careful application of the stationary phase lemma. The main linear oscillation $mt^{*}$ in \eqref{decay.intro} comes from the fact that  $\omega_c(t^{*}) \rightarrow \pm m$ as $t^{*} \rightarrow +\infty$, and the $O((t^{*})^{\frac{1}{3}})$ sublinear corrections come from the  $O ((t^{*})^{-\frac{2}{3}})$ term in $\omega_c(t^{*})$;  finally the presence of a series comes from Taylor-expanding the denominator of \eqref{resolvent.intro}.
	
	To explain the contrast with the massless case, it may be useful to point out that in \eqref{KG.intro} for $m^2=0$,  \begin{equation}\label{massless}
		-\rd_{t t}^2 \phi + \Delta_{\RR^3} \phi = \left(\frac{2M}{r^3}+ O(r^{-4}) \right) \phi,
	\end{equation} where the $r^{-3}$ (short-range) decay of the potential ultimately accounts for  the $t^{-3}$ tail of Price's law using a physical space method \cite{AAG1}. In Fourier space, one can see \cite{Hintz} that the resolvent  fails to be $C^2$ at $\omega=0$: $$  \mathcal{R}(\omega,r) \approx \omega^{2}\log(\omega+i0)  \cdot F(r) + \text{ errors},$$ which is less singular than \eqref{resolvent.intro} (merely in $W^{\frac{1}{2}-,1}_{loc}$ near $\omega=\pm m$), and does not generate a critical point of the phase: the inverse Fourier transform of $\omega^{2}\log(\omega+i0)$ gives indeed a $t^{-3}$ decay rate exactly \cite{Hintz}.
	
	\paragraph{Outline of the rest of the introduction} In Section~\ref{massless.intro.section}, we will discuss the known results on massless scalar fields, including \eqref{KG.intro} with $m^2=0$, but also the charged (massless) scalar field case where the potential scales like $V(r)=O(r^{-2})$. In Section~\ref{massive.section}, we discuss past works on massive scalar fields on black holes: this includes a few works in the physics literature which have argued for the plausibility of a $(t^{*})^{-\frac{5}{6}}$ rate, albeit without predicting the complete form of Theorem~\ref{thm.intro}; and the rigorous work of the second author on Kerr black holes. In Section~\ref{interior.section}, we discuss the consequences of Theorem~\ref{KG.intro} on the interior of dynamical black holes -- in particular, the slowness of the $(t^{*})^{-\frac{5}{6}}$ rate (if it persists in the dynamical black hole setting) leads to a stronger singularity \cite{Moi,Moi4,r0} and  the oscillations involved in $\cos(mt^{*}+ ...)$ become decisive to the question of Strong Cosmic Censorship in spherical symmetry \cite{MoiChristoph,MoiChristoph2}; as such, Theorem~\ref{KG.intro} is an important first step towards retrieving the decay assumptions made  by the third author in \cite{Moi,Moi4,r0,MoiChristoph,MoiChristoph2}. Finally, in Section~\ref{proof.section}, we outline the method of proof in more details: in particular, how the turning point in the potential gives \eqref{resolvent.intro} and how the inverse Fourier transform of \eqref{resolvent.intro} gives the $(t^{*})^{-\frac{5}{6}}$ decay rate of Theorem~\ref{KG.intro}.

	\subsection{Decay for massless fields on black holes, comparison with Theorem~\ref{thm.intro}}
	\label{massless.intro.section} \paragraph{Uncharged scalar  fields} We   further discuss Price's law \cite{PriceLaw}, i.e.\ the analogue of Theorem~\ref{thm.intro} in the massless case \eqref{KG.intro} with $m^2=0$. Price's law on a Schwarzschild/Reissner--Nordstr\"om black hole is  
	\begin{equation}\label{Price.law}
		\phi_{\geq L}(t^{*},r,\theta,\varphi) \sim F(r,\theta,\varphi) \cdot (t^{*})^{-3-2L},
	\end{equation}
	where $\phi_{\geq L}$ denotes the restriction to angular modes greater or equal to $L$. As we explained earlier, \eqref{Price.law} has been proven (with various degrees of precision) in \cite{AAG1,AAG3,PriceLaw,Schlag1,Schlag2,Hintz,Tataru2,Tataru}  with the most definitive version obtained on a Schwarzschild/Reissner--Nordstr\"om black hole in \cite{AAG2}.  
	In comparing \eqref{Price.law} to \eqref{decay.intro}, we note that in the massive case $m^2>0$, \emph{the decay rate is the same for all fixed angular modes} $L$.

	\paragraph{Charged scalar fields} We now turn to the charged scalar field equation: for some $q_0 \in \RR-\{0\}$: \begin{equation}\label{charged.SF}
		g^{\mu \nu } D_{\mu} D_{\nu} \phi= m^2 \phi,\qquad D_{\mu} = \nabla_{\mu} + i q_0  A_{\mu},\qquad d A=F=\frac{e}{r^2}\ dt \wedge (1-\frac{2M}{r}+\frac{e^2}{r})dr,
	\end{equation} 
	\eqref{KG.intro} is thus the uncharged particular case of \eqref{charged.SF} (recall indeed that $\Box_g=   g^{\mu \nu } \nabla_{\mu} \nabla_{\nu} $), where $q_0=0$ (no coupling constant). In the charged and massive case (\eqref{charged.SF} with $m^2>0$), the massive contribution $\frac{m^2 M}{r}$ dominates the $\frac{iq_0 e}{r^2}$ charged contribution for large $r$, so one expects the rate of \eqref{decay.intro} to prevail (see \cite{KonoplyaZhidenko,Konoplya.Zhidenko.num} for numerics). 
	However,  the massless charged equation, i.e.\ \eqref{charged.SF} with $m^2=0$, schematically looks like:
	
	\begin{equation}\label{V.charged}
		-\rd_{t t}^2 \phi + \Delta_{\RR^3} \phi = \left[ \frac{iq_0 e}{r^2} +O(r^{-3}) \right]\phi+\text{first order terms.}
	\end{equation} The $O(r^{-2})$ fall-off rate of the potential in \eqref{V.charged} is to be compared to \eqref{V.intro} and \eqref{massless} where the fall-off rates of the potentials were $O(r^{-1})$ and $O(r^{-3})$ respectively. 
	The following tails  for \eqref{charged.SF} with $m^2=0$ on a Schwarzschild/Reissner--Nordstr\"{o}m black hole have been predicted in the physics literature
	\cite{HodPiran1,HodPiran2,OrenPiran}:
	\begin{equation}\label{Price.law.charged}
		\phi_{\geq L}(t^{*},r,\theta,\varphi) \sim G(r,\theta,\varphi) \cdot (t^{*})^{-1-\sqrt{1-4(q_0e)^2+4L(L+1)}}.
	\end{equation} 
	While the rate predicted by \eqref{Price.law.charged} remains inverse-polynomial, note that it is  slower than Price's law \eqref{Price.law}, but  faster than the massive case \eqref{decay.intro}.  Sharp upper bounds (with a $\ep$-loss)  corresponding to \eqref{Price.law.charged}   have been proven by the third author \cite{Moi2} in spherical\footnote{\cite{Moi2}  more generally treats the nonlinear Maxwell-charged-scalar-field equations, in which the Maxwell field $F$ is dynamical.} symmetry ($L=0$) and for $|q_0 e| \ll 1$. Since the leading order term in \eqref{V.charged} only involves $q_0 e$ and not the mass of the black hole, one expects the \eqref{Price.law.charged} to still hold  for \eqref{charged.SF} on Minkowski space: this has indeed been proven in \cite{MoiDejan}. We end this paragraph by contrasting \eqref{charged.SF} with the Maxwell-charged-scalar-field model on Minkowski spacetime, where $F$ becomes dynamical: in this case, $F$ disperses\footnote{On a black hole, $F$ solving the Maxwell-charged-scalar-field equations does \emph{not disperse } and converges to  $\frac{e}{r^2}\ dt \wedge (1-\frac{2M}{r}+\frac{e^2}{r})dr$ as $t \rightarrow +\infty$: thus, \eqref{charged.SF} is a good  model for the  Maxwell-charged-scalar-field equations on a black hole, but not on Minkowski.} to $0$, and thus the sharp decay is faster, 
	see   \cite{LindbladKG,Bieri.KG,ShiwuKG,YangYu,FangWangYang,KlainermanKG,LinbladKG2,LindbladKG3}.
	\paragraph{Wave with inverse-square potential} The wave equation with an inverse-square potential (i.e.\ with a $O(r^{-2})$ fall-off rate similar to \eqref{V.charged}) is of the following form, where $a>-\frac{1}{4}$ is \emph{real-valued}:	\begin{equation}\label{V.inverse}
		-\rd_{t t}^2 \phi + \Delta_{\RR^3} \phi = \frac{a}{r^2}\phi
	\end{equation} The sharp decay rate is known    \emph{on Minkowski spacetime} to be $t^{-1-\sqrt{1+4a+L(L+1)}}$ \cite{SchlagCostin,Dean,MoiDejan}. It turns out  the same decay  holds for equations of the form \eqref{V.inverse} on a  Schwarzschild/Reissner--Nordstr\"{o}m black hole \cite{inverse.Dejan}.
	
	\paragraph{Other linear decay results} For spin-weighted wave equations, the decay is faster  than \eqref{Price.law}, see \cite{Ma1,Ma3}. We mention specifically the case of Maxwell equations (spin $s=1$) 
	on a Schwarzschild black hole, see \cite{Ma3,Maxwell1,Maxwell2,Maxwell3, pasqualotto2019spin} for results in that direction. Also note that for wave equations with a potential decaying as $O(r^{-2-a})$  with $a>0$,  $O(t^{-2-a})$ decay holds, see \cite{LOOI,Faster2,Faster3}.

	\subsection{Previous works on massive scalar fields on black holes}\label{massive.section}
	
	\paragraph{Previous heuristics and numerics for Klein--Gordon fields on a Schwarzschild or Reissner--Nordstr\"{o}m black hole}

	Koyama and Tomimatsu predicted in the early 2000's \cite{KoyamaTomimatsu} that Klein--Gordon fields on a sub-extremal  Reissner--Nordstr\"{o}m black hole \emph{sufficiently near extremality} (i.e.\ for $M-|e|$ small) decay with the  leading tail \begin{equation}\label{decay.Koyama}
		(t^{*})^{-\frac{5}{6}}\cos(\Psi_1(t^{*}) +\Omega(t^{*})),
	\end{equation}   where $\Psi_1(t^{*}) $ is as in \eqref{decay.intro} and $\Omega(t^{*})$ is bounded. We note that the  formula~\eqref{decay.Koyama} does not exactly agree\footnote{More specifically, the work~\cite{KoyamaTomimatsu} (see also \cite{diracjing}) produces a formula for the leading order singular behavior as $|\omega| \rightarrow_{|\omega|<m} m$ of the cut-off resolvent which agrees with our corresponding formula. However, the authors of~\cite{KoyamaTomimatsu} use certain approximations in their stationary phase analysis which lead to discrepencies with our derived asymptotics for the scalar field.}    with~\eqref{decay.intro}; 
	however, the two formulas are qualitatively similar in that  the term $q=1$ in the sum of \eqref{decay.intro} corresponds to \eqref{decay.Koyama} with $\Omega(t^{*})=0$.
	In  subsequent work \cite{KoyamaTomimatsu2}, the authors claim  \eqref{decay.Koyama} also holds for a Schwarzschild black hole of small mass $M m \ll 1$ or large mass $M m \gg 1$.

	The works    \cite{KoyamaTomimatsu,KoyamaTomimatsu2} both rely on heuristic arguments consisting in  taking the Fourier transform in $t$, working with complex $\omega\in  \mathbb{C}$ and dropping certain terms in the equation to end up with an explicitly solvable Whittaker ODE (for large $r$). Solving this Whittaker ODE and using the algebraic properties of the Whittaker special function gives a formula of the form \eqref{resolvent.intro} (with all error terms being ignored).   It is not clear, however, how to control the errors involved in this approximation: first, the terms being ignored are not perturbative in the usual sense, and second, the algebraic properties of Whittaker special function are not stable to perturbations. Instead, our strategy identifies the key point as  being the turning point in the potential and  relies on a more robust analysis of this turning point which works in greater generality, see Section~\ref{proof.section}. 
	
	The case of a Schwarzschild black hole of intermediate mass $Mm \approx 1$ (or a  Reissner--Nordstr\"{o}m black hole far away from extremality)  was not treated in \cite{KoyamaTomimatsu,KoyamaTomimatsu2} due, in particular, to potential cancellations (that the authors could not plausibly exclude) in \eqref{resolvent.intro} leading to $|\Gamma_+|=|\Gamma_-|$, which would lead to decay faster than $(t^{*})^{-\frac{5}{6}}$. In \cite{KoyamaTomimatsu3}, the authors however provide a plausibility argument that such a cancellation does not occur on any Schwarzschild/Reissner--Nordstr\"{o}m black hole and thus claim that \eqref{decay.Koyama} holds in those cases as well. In our work, the issue of ruling out such cancellations is treated in a unified way on the whole Reissner--Nordstr\"{o}m sub-extremal range and for any mass, using the energy identity, see Section~\ref{proof.section}.

	We conclude this section mentioning that independent numerical work \cite{BurkoKhanna,Konoplya.Zhidenko.num}  also predicted a tail  with the following rate:  $(t^{*})^{-\frac{5}{6}} \sin(mt^{*})$  (albeit without the $O((t^{*})^{\frac{1}{3}})$ sublinear phase corrections).

	We also mention \cite{massivedirac} studying the solutions of   the massive Dirac equation  on a Kerr--Newman black hole, projected on a fixed spherical {harmonic $Y_L$ and finding a $O((t^{*})^{-\frac{5}{6}})$ tail in this situation. Specializing \cite{massivedirac} to a Reissner--Nordstr\"{o}m black hole,   we note that both the asymptotics and\footnote{While the   $(t^{*})^{-5/6}$ rate is identical, note that \eqref{decay.intro} involves  an infinite sum of terms $(-\gamma_m)^{q-1}q^{1/3}\cos(\Psi_q(t^{*})+  \check{\varphi}_H(r)+\theta[\phi_0,\phi_1])$ while the corresponding formula in~\cite{massivedirac} involves only one such term. At the level of the cut-off resolvent, \cite{massivedirac} finds a $e^{i (m^2-\omega^2)^{-1/2}}$ term as $|\omega| \rightarrow_{|\omega|<m} m$, to be compared with the non-trivial  rational function of  $e^{i (m^2-\omega^2)^{-1/2}}$ in \eqref{resolvent.intro}.} the singular behavior of the cut-off resolvent as $|\omega| \rightarrow_{|\omega|<m} m$ found in \cite{massivedirac} differ from those found in Theorem~\ref{KG.intro}. It is unclear to us how to adapt the methods used in \cite{massivedirac} to our situation.

		\paragraph{Non-decaying massive scalar fields on a black hole} On a Kerr black hole, the combination of \emph{superradiance} and the massive character of the scalar field leads to the existence of growing modes, i.e.\ solutions of the form $F(r,\theta,\varphi) e^{-i \omega  t^{*}}$, with $\Im(\omega)>0$ and $F(r,\theta,\varphi)$ smooth, as was proven in work of the second author \cite{Yakov} (for an open set of Klein--Gordon masses $m^2$). In particular, this means that decay in time cannot hold generically\footnote{Note that decay, however, could potentially hold on a Kerr black hole after projecting away from the growing modes, i.e.\ potentially for a countable-codimension set of initial data.}, in contrast to the situation of Theorem~\ref{thm.intro}. 
		
		In the Kerr setting, there exist oscillating modes \cite{Yakov}, i.e.\ solutions of the form $F(r,\theta,\varphi) e^{-i \omega  t^{*}}$ with $\omega \in \RR$.  On the nonlinear level, these oscillating modes lead to the existence of hairy black holes: specifically stationary bifurcations of the Kerr metric,  as was proven by the second author and Chodosh	 \cite{OtisYakov}.

		Relatedly, we also mention situations in spherical symmetry   where the massive field allows for stationary event horizon differing from Schwarzschild's or Reissner--Nordstr\"{o}m's, see \cite{fluctuating,violent} for an extended discussion.
		\subsection{Previous works on long range potentials}
		Long range potentials merely satisfy  weak decay $V(r)=O(r^{-\alpha})$  as $r\rightarrow +\infty$, for  $\alpha \in (0,1]$ (as opposed to the short-range case  $\alpha>1$). The case $\alpha=1$ includes the  important Coulomb potential $V(r)=  \pm\frac{ 1}{r}$. It is well known that basic properties of long-range scattering differ from the short-range case: however, it is a classical result that  modified wave operators exist and asymptotic completeness is available, see  e.g.\ \cite{HormanderIV}. 
		
		For comparison, our potential in \eqref{KG.intro} is of the form $V(r)= -\frac{1}{r}+O(r^{-2})$ at $r\rightarrow+\infty$, thus Coulomb-like of attractive type. The main difficulty in this case is to analyze the resolvent  near the threshold $E=0$ (corresponding to $\omega=\pm m$ in  the Fourier transform of \eqref{KG.intro}). There are two regimes: the positive energy case ($E>0$, corresponding to $|\omega|>m$ in \eqref{KG.intro}) and the negative energy case ($E<0$, i.e.\ $|\omega|<m$ in \eqref{KG.intro}). 
		
		The analysis is the regime $E\rightarrow 0^+$ is more straightforward  in the attractive case and has been carried out in details for a general attractive Coulomb-like potential in \cite{Ethan1} (corresponding to the $|\omega|>m$ regime in our case, i.e.\ the frequencies contributing the least to the decay \eqref{decay.intro}), see \cite{long1,long2,long3,long4} for earlier works.

		The  regime $E\rightarrow 0^-$ is more delicate due to the existence of \emph{bound states} in general, with energy $E<0$ arbitrarily close to $0$ (see \cite{Ethan2} for a proof under very general conditions). In fact, the existence of bound states is directly responsible on the nonlinear level for the existence of so-called \emph{boson stars} \cite{Boson}, which are spherically symmetric solutions of the Einstein--Klein--Gordon equations. Returning to the black hole case: for our equation \eqref{KG.intro}, we prove that \emph{there are no bound states}: this is a direct consequence of the energy identity and relies on the existence of an event horizon: it is thus specific to the black hole case (note that the bound states  of \cite{Ethan2} are constructed on asymptotically Schwarzschildian spacetimes, but with no event horizon). Nonetheless, the regime $E \rightarrow 0^{-}$ (i.e.\ $|\omega|<m$ in \eqref{KG.intro}) remains the most delicate in our analysis and the slow decay of \eqref{KG.intro} can be interpreted as the remnant of a bound state.   We note that for the Kerr black hole with non-zero angular momentum $aM \neq 0$, the regime $E<0$ does contain bound states \cite{Yakov,OtisYakov};  although these bound states do not persist  in the $a = 0$ Schwarzschild black hole limit, it may be useful to think that the slow decay corresponds to their $a=0$ limit.
		\subsection{The consequences of decay on the interior of black holes}\label{interior.section}
		
		Theorem~\ref{thm.intro} implies decay  at the event horizon $\mathcal{H}=\{r=r_+(M,e)\}$, evaluating \eqref{decay.intro} at $r=r_{+}$. For the Einstein equations in vacuum (respectively with matter), it is well-known that 
		the decay of gravitational waves (respectively and of the matter fields) on the event horizon have a decisive impact on the geometry of the dynamical black hole interior	\cite{KerrStab,KerrInstab,JonathanStab,JonathanWeakNull,Moi,Moi4,MoiChristoph,MoiChristoph2,MoiThesis,r0,JonathanICM,DejanAn,JanC1,gajicluk,mass,violent,fluctuating} which is important to the celebrated Strong Cosmic Censorship Conjecture \cite{PenroseSCC,Penroseblue}, see also \cite{kehle2020diophantine,JonathanInstab,KerrInstab,Ma2,Kehle2018,Kehle2019,Jan} for related linear results.
		We restrict our attention to the Einstein--Maxwell equations coupled with a massive scalar field  and we discuss results of the third author on the geometry  of the black hole interior and Strong Cosmic Censorship for that Einstein--Maxwell--Klein--Gordon model  in spherical symmetry. In what follows $v$ is an advanced-time coordinate which coincides with $t^{*}$  on a fixed Reissner--Nordstr\"{o}m black hole.

		\begin{thm}\label{int.thm}\cite{Moi,Moi4,MoiChristoph,MoiChristoph2}, Black hole interior in spherical symmetry] Consider a dynamical black hole solution of the Einstein--Klein--Gordon equations in spherical symmetry  converging asymptotically  to a sub-extremal  Reissner--Nordstr\"{o}m black hole as $v\rightarrow +\infty$ on its event horizon. Then:
			\begin{enumerate}[A.]
				\item\label{int1} (Cauchy horizon stability, \cite{Moi}) If $|\phi|_{|\mathcal{H}}(v)+ |\rd_{v}\phi|_{|\mathcal{H}}(v) \ls v^{-s}$ for some $s>\frac{1}{2}$, then the black hole terminal boundary near $i^+$ consists of a non-empty null boundary called a Cauchy horizon  $\mathcal{CH}_{i^+}$.\\ If $s>1$, then the Cauchy horizon is $C^0$-extendible and $\phi$ is bounded $W^{1,1}_{loc}\cap L^{\infty}$ norm.
				
				\item \label{int2} (Weak singularity at the Cauchy horizon, \cite{Moi,Moi2})  If moreover $ \int_v^{+\infty} |\rd_{v}\phi|^2_{|\mathcal{H}}(w) dw \gtrsim v^{1-2s}$, then the Cauchy horizon $\mathcal{CH}_{i^+}$ admits a curvature singularity, is $C^2$-inextendible, and $\phi$ blows up in $H^1_{loc}$  norm.
				\item \label{int3} (Slightly stronger Cauchy horizon singularity as a consequence of slow decay, \cite{MoiChristoph})  If moreover  $s>\frac{3}{4}$ and  $\int_1^{+\infty} |\rd_{v}\phi|_{|\mathcal{H}}(w) dw=~\infty$, then  $\phi$ blows up in $W^{1,1}_{loc}$ norm at the Cauchy horizon.
				
				\item \label{int4} (Decisive character of oscillations, \cite{MoiChristoph,MoiChristoph2})  If moreover $ \bigl|\int_1^{+\infty} \rd_{v}\phi_{|\mathcal{H}}(w) dw\bigr|<\infty$, then the Cauchy horizon $\mathcal{CH}_{i^+}$  is $C^0$-extendible and $\phi$ is bounded $ L^{\infty}$ norm. If however $\bigl|\int_1^{+\infty} \rd_{v}\phi_{|\mathcal{H}}(w) dw\bigr|=\infty$, then the Cauchy horizon $\mathcal{CH}_{i^+}$  does not admit any $C^0$-admissible extension and $\phi$ blows up in $ L^{\infty}$ norm.
			\end{enumerate}
		\end{thm} Formally \eqref{decay.intro}  at $r=r_+$ satisfies \ref{int1}, \ref{int2} and \ref{int3}. The main term involving $(t^{*})^{-\frac{5}{6}}$  in \eqref{decay.intro} also obeys the oscillation condition  $\bigl|\int_1^{+\infty} \rd_{t^{*}}\phi_{|\mathcal{H}}(w) dw\bigr|<\infty$ in \ref{int4}, but the error $O((t^{*})^{-1+\delta})$ is not integrable so strictly speaking \eqref{decay.intro} does not imply $\bigl|\int_1^{+\infty} \rd_{t^{*}}\phi_{|\mathcal{H}}(w) dw\bigr|<\infty$. We also emphasize that \eqref{decay.intro} from Theorem~\ref{thm.intro} is but the first step in proving that \ref{int1}, \ref{int2}, \ref{int3}, \ref{int4} are satisfied generically for the \emph{nonlinear} Einstein--Maxwell--Klein--Gordon model in spherical symmetry (a purpose which  requires ideas that go beyond the proof of Theorem~\ref{thm.intro}). A proof of such a result  combined with Theorem~\ref{int.thm} would, however, provide a definitive treatment of the Strong Cosmic Censorship Conjecture in spherical symmetry for the Einstein--Maxwell--Klein--Gordon model.

		\subsection{Methods of the proof}\label{proof.section} Our strategy relies on \emph{scattering theory}:  concretely, we take the Fourier transform in time $t$ of \eqref{KG.intro} and after decomposition the solution in angular modes, we end up with an ODE in $r$ for each\footnote{We emphasize that the heart of  the argument will only involve real-valued time-frequency $\omega$ instead of the usual approach involving contours on the complex plane (we will however briefly work in the complex plane to softly justify the Green's formula).} time-frequency $\omega\in \RR$ and angular mode $L$.  The two important extremities of the black hole correspond to event horizon $\mathcal{H}=\{r=r_+(M,e)\}$ ($r_+(M,e)$ is the positive root of $(1-\frac{2M}{r}+\frac{e^2}{r^2})$)  and infinity  $\mathcal{I}=\{r=+\infty\}$. The standard method  consists in relating $u_H$ -- the solution of \eqref{KG.intro} regular at $\mathcal{H}$, and $u_I$ the one 	regular at $\mathcal{I}$. In view of the discussion of the prequel, the proof of Theorem~\ref{thm.intro} comes down to answering five questions. \begin{enumerate}
			\item Why is the problem  essentially equivalent to taking the Fourier transform of $\mathcal{R}(\omega,r)$ the LHS of \eqref{resolvent.intro}?
			\item Which frequencies $\omega$ contribute the ``most'' to the $(t^{*})^{-\frac{5}{6}}$ tail of \eqref{decay.intro}?
			\item Why does the inverse-Fourier transform of \eqref{resolvent.intro} give \eqref{decay.intro} and the errors in $\omega$ give $ O((t^{*})^{-1+\delta})$ ?
			\item Why are there no bound states, and how do we know\footnote{Note indeed that, if we had  $\Gamma_+ = \Gamma_-$ in \eqref{resolvent.intro}, then $\mathcal{R}(\omega)$ would be constant and thus the decay  faster than the $O\left((t^{*})^{-\frac{5}{6}}\right)$.} that $\Gamma_+ \neq \Gamma_-$ in \eqref{resolvent.intro}?
			
			\item Why is \eqref{resolvent.intro} true? How does one control the errors in this formula?
			
		\end{enumerate} We will dedicate a paragraph of explanations of each of these questions below.
		
		\paragraph{The scattering approach at fixed time-frequency and Green's formula} 
		
		Taking the Fourier transform in $t$ of \eqref{KG.intro} for the angular mode $\phi_{L}$ gives the following ODE in terms of the variable $s(r)$ defined by $\frac{ds}{dr} = (1-\frac{2M}{r}+\frac{e^2}{r^2})^{-1}$, and the function $u(\omega,s) = \int_{-\infty}^{+\infty} e^{i \omega t} (r\phi_{L})(t,r)  dt$:	\begin{equation}\label{ODE.intro}
			\p^2_{s} u= - \omega^2 u + \Big(1-\frac{2M}{r}+\frac{e^2}{r^2}\Big)\Big(m^2 + \frac{L (L+1)}{r^2} +\frac{2M}{r^3}-\frac{2e^2}{r^4}\Big)u + \hat H(\omega,s)= \left[m^2-\omega^2 - \frac{2Mm^2}{r} + O(r^{-2}) \right]u + \hat{H}(\omega,s),
		\end{equation} where  $\hat{H}(\omega,s)$ is  the Fourier transform of ${H}(t,s)$, the source term accounting for the initial data of $\phi_{L}$ at $t^{*}=0$. (Note that we take the Fourier transform with respect to $t$, not $t^{*}$, as the ODE analysis is more straightforward using $(t,s)$ variables. However, to obtain a decay statement which is uniform in $r$ as $r\rightarrow r_+$, we switch to $(t^{*},r)$ coordinates at the very end of our argument.)

		We  show (using standard ODE analysis) the following Green's formula, with    $W(u_I,u_H)(\omega)= \rd_s u_I\ u_H -  \rd_s u_H\ u_I$  the Wronskian of $u_I$ the solution regular at infinity $\mathcal{I}=\{r=\infty\}$ and $u_H$ the solution regular at the event horizon $\mathcal{H}=\{r=r_+\}$: \begin{equation}\label{Green.intro1}
			u(\omega,s) =  \frac{u_H(\omega,s) \int_{s}^{+\infty} u_I(\omega,S) \cdot \hat{H}(\omega,S) dS +   u_I(\omega,s) \int^{s}_{-\infty} u_H(\omega,S) \cdot\hat{H}(\omega,S) dS}{W(u_I,u_H)(\omega)},
		\end{equation}
		The  scattering argument is to write $u_I$ (which is real-valued) in terms of $u_H$ and its complex conjugate as:
		\begin{equation}
			u_I(\omega,s)=  \frac{W(u_I,\bar{u}_H)(\omega)}{W(\bar{u}_H,u_H)(\omega)}\ u_H(\omega,s) + \frac{W(u_I,u_H)(\omega)}{W(u_H,\bar{u}_H)(\omega)}\ \bar{u}_H(\omega,s).
		\end{equation}
		which combined to \eqref{Green.intro1} gives a term of the form \begin{equation}\label{Green.intro3}
			u(\omega,s) = \frac{W(u_I,\bar{u}_H)(\omega)}{W(u_I,u_H)(\omega)} \frac{u_H(\omega,s) }{2i\omega} \int_{-\infty}^{+\infty} u_H(\omega,S) \hat{H}(\omega,S) dS+ \text{ error},
		\end{equation} where  all the terms are regular except $\frac{W(u_I,\bar{u}_H)(\omega)}{W(u_I,u_H)(\omega)}$ which is singular as $\omega \rightarrow \pm m$.
		
		Thus, $\phi$ is essentially -- up to errors  -- given by the inverse Fourier transform of $\frac{W(u_I,\bar{u}_H)(\omega)}{W(u_I,u_H)(\omega)}$ --  as it appears in the LHS of \eqref{resolvent.intro}.
		\paragraph{The faster-decaying contributions of the frequencies away and slightly above $\pm m$} Since \eqref{Green.intro3} is regular away from $\omega =\pm m$, one expects arbitrarily fast polynomial decay, thus we concentrate in regime $\omega \approx \pm m$. Since $\omega \in \RR$, there are two sub-regimes: $|\omega|>m$ and $|\omega|<m$. In this paragraph,  we discuss   $|\omega|>m$: note that in this case the main part  $(m^2-\omega^2 -\frac{2Mm^2}{r}) $ of the RHS of \eqref{ODE.intro} has a constant negative sign: because of this, it turns that a single WKB-type analysis allows to solve \eqref{ODE.intro} approximately in the regime $s\gg 1$. The result of this analysis is that  	$\frac{W(u_I,\bar{u}_H)(\omega)}{W(u_I,u_H)(\omega)} $ is absolutely continuous (i.e.\ its derivative is integrable near $\omega=\pm m$) and thus its inverse Fourier transform decays at least like $O((t^{*})^{-1})$. The conclusion is that frequencies $|\omega|>m$ do not contribute to the main decaying term in \eqref{decay.intro}.
		\paragraph{The singularity of the resolvent and its consequences} As explained above, the frequencies $|\omega|>m$ do not contribute. As  hinted in \eqref{resolvent.intro}, we will show that for $|\omega|<m$ \begin{equation}\label{resolvent.intro2}
			\frac{W(u_I,\bar{u}_H)(\omega)}{W(u_I,u_H)(\omega)} = \frac{ \Gamma_+\ e^{i \pi  M m^2 (m^2-\omega^2)^{-\frac{1}{2}}} +  \Gamma_-\ e^{-i \pi  M m^2 (m^2-\omega^2)^{-\frac{1}{2}}}}{\bar{\Gamma}_+\ e^{-i \pi  M m^2 (m^2-\omega^2)^{-\frac{1}{2}}} +  \bar{\Gamma}_-\ e^{i \pi  M m^2 (m^2-\omega^2)^{-\frac{1}{2}}}}  + \text{ errors}.
		\end{equation}
		Taking the inverse Fourier transform of $e^{i \pi  M m^2 (m^2-\omega^2)^{-\frac{1}{2}}}$ involves a phase $\Phi(\omega,t)=-i\omega t+ i \pi  M m^2 (m^2-\omega^2)^{-\frac{1}{2}}$, which turns out to have a critical point $\omega_c(t) = \pm m + O(t^{-\frac{2}{3}})$. Since $\rd_{\omega \omega}^2\Phi(\omega_c(t),t)= O (t^{\frac{5}{3}})$, a stationary phase argument shows that the inverse Fourier transform of $e^{i \pi  M m^2 (m^2-\omega^2)^{-\frac{1}{2}}}$ (suitably localized to the region $\{ |\omega|<m\}$) decays at the rate $t^{-\frac{5}{6}}$.
		
		To handle the fraction in \eqref{resolvent.intro2}, observe that if $|\Gamma_+|<|\Gamma_-|$, one can expand in denominator in a Taylor series: this is what explains the summation in \eqref{decay.intro}.	We will explain in the paragraph below how the energy identity actually implies that $|\Gamma_+|<|\Gamma_-|$ (in fact $|\gamma_m|= \frac{|\Gamma_-|(m)}{|\Gamma_+|(m)}$ in \eqref{decay.intro}). To explain where \eqref{resolvent.intro2} originates, one has to discuss the role of the turning point in the potential of \eqref{ODE.intro} (see paragraphs below).
		
		\paragraph{The energy identity and the explicit oscillating profile} The well-known energy identity exploits that the potential in $\eqref{ODE.intro}$ is real-valued and thus $\Im(\rd_{s}u\ \bar{u})$ is independent of $s$ for all solution $u$ of \eqref{ODE.intro}. Using the energy identity shows that bound states do not exist. Moreover, if one relates to $u_H(\omega,s)$ to  $\tilde{B}_{\pm}$ the two independent solutions of \eqref{ODE.intro} where $\omega=\pm m $ (indeed \eqref{ODE.intro} can be solved explicitly up to errors if $\omega= \pm m$:  $\tilde{B}_{\pm}$ at leading order behave like Bessel functions), one gets \begin{equation}
			u_H(\omega,s) =  E_+(\omega) \cdot \left[\tilde{B}_+(s)+ \text{error}\right] + E_-(\omega) \cdot \left[\tilde{B}_-(s) + \text{error}\right], 
		\end{equation} and the energy identity implies $|E_+|^2 - |E_-|^2 \gtrsim 1$ (for $\omega>0$): it turns out that this is equivalent to showing $|\Gamma_+|(\omega)<~|\Gamma_-|(\omega)$ (i.e.\ $|\gamma_m|= \frac{|\Gamma_-|(m)}{|\Gamma_+|(m)}<1$ in \eqref{decay.intro}), which is necessary for the above summation argument.
		
		\paragraph{The main contribution from the low frequencies and the role of the turning point} In the regime $|\omega|<m$, however, the potential in the RHS of \eqref{ODE.intro} admits a \emph{turning point} when $m^2-\omega^2-\frac{2Mm^2}{r}=0$. The key point is that this turning point is located at $r  \approx  [m^2-\omega^2]^{-1}$ and thus intervenes closer to infinity as $\omega$ approaches $\pm m$. The resolution of the turning point yields coefficients of the form \eqref{resolvent.intro2} which oscillate  near $\omega=\pm m $ (this is due to the behavior of the Airy function, see Appendix~\ref{app:airy}). The management of the errors in the turning point analysis, and a subsequent WKB analysis below the turning point constitute the technical heart of the paper which leads to the $O(t^{-1+\delta})$ error in \eqref{decay.intro}. We conclude with the following  technical comment:  the turning point analysis describes the solution around and also ``above'' the turning point, that is, for $r \gg \left(m^2-\omega^2\right)^{-1}$. Below the turning point, a WKB analysis is more accurate: we match the two regimes for   $r  \approx  [m^2-\omega^2]^{-1}$ slightly below the turning point. Then, the main error comes from matching the solution at $\omega=\pm m$ (described by Bessel functions) with the WKB regime: the former becomes less accurate for large $r$, while the latter for smaller $r$. As it turns out, minimizing the error in \eqref{resolvent.intro2} leads to evaluating the scattering coefficients in the region  $r  \approx  [m^2-\omega^2]^{-\frac{1}{2}}$,  which is still near infinity, but below the turning point. 
		
		\subsection{Outline of the paper}
		In Section~\ref{setup.section}, we express the Klein--Gordon equation \eqref{KG.intro} in various forms and introduce our notations and conventions, notably regarding the Fourier transform. In Section~\ref{main.result.section}, we give the precise statement of Theorem~\ref{thm.intro}. In Section~\ref{yakov.section} we start the proof and address the justification of the Green's formula \eqref{Green.intro1}, the ODE asymptotic analysis (``boundary conditions for the ODE \eqref{ODE.intro}), and frequencies least contributing ($|\omega|>m$)  to the tail \eqref{decay.intro}. Section~\ref{maxime.section} is the technical heart of the paper in which we address the frequencies  most contributing ($|\omega|<m$) to the tail \eqref{decay.intro} through a careful and quantitative turning point and WKB analysis. In Section~\ref{FT.section}, we take the inverse Fourier transform of expressions found in Section~\ref{yakov.section} and Section~\ref{maxime.section} and exploit the Green's formula and various stationary phase arguments to yield \eqref{KG.intro}. Finally, we have two appendices: one on special functions intervening in our problem (Bessel and Airy) and one on various Volterra-type ODE results.
		
		\subsection{Acknowledgements} YS acknowledges support from an Alfred P. Sloan Fellowship in Mathematics and from NSERC discovery grants RGPIN-2021-02562 and DGECR-2021-00093. MVdM gratefully acknowledges support from the NSF Grant	 DMS-2247376.

		\section{Set-up}\label{setup.section}
		
		\subsection{Equations}\label{eq.section}
		
		We study	\begin{equation}\label{eq:kg}
			\Box_g \phi - m^2 \phi = F,
		\end{equation}   
		where $F$ denotes a potential inhomogeneity, $g$ is the following spherically symmetric black hole metric \begin{equation}\label{metricstandard}
			g= -\Big(1-\frac{2M}{r}+\frac{\DD(r)}{r^2}\Big) dt^2 +  \Big(1-\frac{2M}{r}+\frac{\DD(r)}{r^2}\Big)^{-1} dr^2 + r^2 \left( d\theta^2+ \sin^2(\theta) d\varphi^2\right),
		\end{equation} on  $(t,r,\theta,\varphi) \in \mathbb{R} \times (r_+,\infty) \times \mathbb{S}^2$, where  $r \rightarrow \DD(r)$ is a smooth function such that \begin{enumerate}
			\item \label{D1}$1-\frac{2M}{r}+\frac{\DD(r)}{r^2}$ admits a largest  zero $r_+>0$ such that $1-\frac{2M}{r}+\frac{\DD(r)}{r^2}>0$ for all $r>r_+$.
			\item \label{D2}  $\DD(r)$, $\DD'(r)$ and $\DD''(r)$ admit an asymptotic expansion in $\frac{1}{r}$, i.e.\ there exists $(D_N)_{N\in \mathbb{Z}_{\geq 0}}$ such that  for all $N\in \mathbb{Z}_{\geq 0}$ such that for all $p \in \{0,1,2\}$, there exists $C_{N+1+p}>0$, for all $r > r_+$:
			$$ \big|\frac{d^p}{dr^p} \left(\DD(r) - \sum_{k=0}^{N}  \frac{D_k}{r^k}\right)\bigr| \leq \frac{C_{N+1+p}}{r^{N+1+p}}.$$
			\item  \label{D3} We assume that $\kappa(r)\doteq \frac{1}{2}\frac{d}{dr}\left(1-\frac{2M}{r} + \frac{\DD(r)}{r^2}\right)>0$ for all $r\geq r_+$.  
		\end{enumerate}  Note that assumption \ref{D3}.~for $r=r_+$ corresponds to a strictly positive surface gravity of the  event horizon $\{r=r_+\}$. Assumption \ref{D3}.~is also used to ensure that the potential $V_{\omega}(r)$ in Section~\ref{sec:ftransform} is strictly positive, which prevents the occurrence of a resonance at zero frequency. \begin{rmk}
			From the proof, it is clear that it is not necessary to require $\DD(r)$ to admit an asymptotic expansion in $\frac{1}{r}$ at all orders: it is in fact sufficient to have $|\DD|(r)+ r|\DD'|(r)+ |\DD''|(r)\ls 1$. We will however not pursue this to avoid tedious technicalities.
		\end{rmk}
		The case where $\DD$ is constant, i.e.\  $\DD(r) =e^2 \geq 0$ with $0 \leq |e|<M$ corresponds  to
		the sub-extremal Reissner--Nordstr\"{o}m metric  (and satisfies \ref{D1}, \ref{D2}, \ref{D3})\begin{equation}\label{RNstandard}
			g= -\Big(1-\frac{2M}{r}+\frac{e^2}{r^2}\Big) dt^2 +  \Big(1-\frac{2M}{r}+\frac{e^2}{r^2}\Big)^{-1} dr^2 + r^2 \left( d\theta^2+ \sin^2(\theta) d\varphi^2\right).
		\end{equation}  (note that we  allow $e=0$, corresponding to the Schwarzschild metric).  We assume that $\phi$ is supported on the $L$ spherical harmonics.  More concretely,
		\begin{equation}\label{formofphi}
			\phi\left(t,r,\theta,\varphi\right) = \tilde{\phi}\left(t,r\right)Y_L\left(\theta,\varphi\right),
		\end{equation}
		where $Y$ is a spherical harmonic associated to the eigenvalue $-L(L+1)$ of the spherical Laplacian where $L \in \mathbb{Z}_{\geq 0}$. Unless said otherwise, in what follows, we shall always assume that $\phi$ is of the form~\eqref{formofphi} for some $L \in \mathbb{Z}_{\geq 0}$. We furthermore allow all constants to depend on $L$, and, because it is unlikely to cause confusion, we will identify $\phi$ and $\tilde{\phi}$ in what follows.
		
		We also recall that for $\DD(r)=e^2$ (Reissner--Nordstr\"{o}m case), $r_+(M,e)$ and $\kappa_+:=\kappa(r_+)$ are given by
		\[r_{\pm} \doteq M \pm \sqrt{M^2-e^2}>0,\]
		\[\kappa_{+} \doteq \frac{1}{2}\frac{r_+-r_-}{r_+^2}> 0.\]
		The Klein--Gordon equation~\eqref{eq:kg} then becomes
		\begin{align}\label{KGinrcoord}
			&-\left(1-\frac{2M}{r}+\frac{\DD(r)}{r^2}\right)^{-1}\partial_t^2\phi + \left(1-\frac{2M}{r}+\frac{\DD(r)}{r^2}\right)\partial_r^2\phi 
			\\ \nonumber &\qquad + \left(\frac{2M}{r^2} +\frac{r\DD'(r)-2\DD(r)}{r^3} +\frac{2}{r}\left(1-\frac{2M}{r}+\frac{\DD(r)}{r^2}\right)\right)\partial_r\phi -\left(m^2 + \frac{L(L+1)}{r^2}\right)\phi = F. 
		\end{align}
		This becomes, letting 
		\begin{equation}\label{defpsiandH}
			\psi:= r \phi,\qquad H \doteq r\left(1-\frac{2M}{r}+\frac{\DD(r)}{r^2}\right)F,
		\end{equation}
		and the coordinate $s=\rs \in \mathbb{R}$ defined by $\frac{ds}{dr}= (1-\frac{2M}{r}+ \frac{\DD(r)}{r^2})^{-1} $:
		\begin{equation}\label{eq:main}
			- \p_t^2 \psi + \p^2_{s} \psi  = \Big(1-\frac{2M}{r}+\frac{\DD(r)}{r^2}\Big)\Big(m^2 + \frac{L (L+1)}{r^2} +\frac{2M}{r^3}+\frac{r\DD'(r)-2\DD(r)}{r^4}\Big)\psi + H.
		\end{equation}
		It is useful to keep in mind that
		\[s = 2\kappa_+\log(r-r_+) + O(1)\text{ as }r\to r_+.\]
		Lastly, as is well-known, we recall that the metric expression~\ref{metricstandard} is not valid along the event horizon. In order to cover the horizon we introduce a coordinate $t^*\left(t,r\right) = t + p(r)$, where we define $p(r)$ by, for $r \lesssim 1$,
		\begin{equation}\label{compartrdefp}
			\frac{dp}{dr} \doteq \sqrt{P(r)} \doteq \sqrt{\left(1-\frac{2M}{r}+\frac{\mathcal{D}}{r^2}\right)^{-2} - \epsilon_0(r)\left(1-\frac{2M}{r}+\frac{\mathcal{D}}{r^2}\right)^{-1} - \epsilon_1(r)\frac{2M}{r}},
		\end{equation}
		for suitable functions $\epsilon_0(r)$ and $\epsilon_1(r)$ which are everywhere positive and sufficiently small, setting $p(r)$ to be constant for $r \gg 1$, and otherwise requiring that $t^*$ is smooth and that $t^* = \text{const}$ are spacelike hypersurfaces.
		
		Then in $\left(t^*,r\right)$ coordinates, we have
		\begin{align}\label{metriceddi}
			&g = -\left(1-\frac{2M}{r}+\frac{\mathcal{D}}{r^2}\right)\left(dt^*\right)^2 + 2\sqrt{1-\epsilon_0(r)\left(1-\frac{2M}{r} + \frac{\mathcal{D}}{r^2}\right) - \epsilon_1(r)\frac{2M}{r}\left(1-\frac{2M}{r}+\frac{\mathcal{D}}{r^2}\right)^2}dt^*dr 
			\\ &\qquad + \left(\epsilon_0(r)+\epsilon_1(r)\left(1-\frac{2M}{r}+\frac{\mathcal{D}}{r^2}\right)\frac{2M}{r}\right)dr^2 + r^2\left(d\theta^2 + \sin^2\theta d\varphi^2\right).
		\end{align}
		Note that a computation yields that when $r \lesssim 1$
		\[\left|\nabla t^*\right|^2 = -\epsilon_0(r) -\epsilon_1(r)\frac{2M}{r}\left(1-\frac{2M}{r}+\frac{\mathcal{D}}{r^2}\right),\]
		so that the hypersurfaces of constant $t^*$ are spacelike. 
		
		It is useful to keep in mind that
		\[p(r) - s(r) = O(1)\text{ as }r\to r_+.\]
		We let $Q\left(r\right)$ denote the product of the square root of the determinant of $-g$ in the $\left(t^*,r\right)$ coordinates with the function $r^{-2}|\sin|^{-1}(\theta)$. We will have
		\[Q \sim 1,\qquad  \partial_rQ \sim r^{-2} \qquad \forall r \in [r_+,\infty).\]
		Finally, for later use, we compute the commutator of $\Box_g$ with $f(t^*)$ for any smooth function $f(t^*)$:
		\begin{align}\label{commutatorboxftstar}\begin{split}
				\left[\Box_g,f(t^*)\right]\phi &= \left(\Box_gf\right)\phi + 2\nabla f \cdot \nabla \phi
				\\  &= A_1\left(r\right) f''\phi + 2A_2(r) f'\phi + 2A_1\left(r\right)f'\partial_{t^*}\phi + 2A_3\left(r\right)f'\partial_r\phi,\end{split}
		\end{align}
		where, in the region where~\eqref{compartrdefp} holds,
		\begin{equation}\label{thisisa1cool}
			A_1(r) \doteq Q^{-2}\left(\epsilon_0(r) + \epsilon_1(r)\left(1-\frac{2M}{r}+\frac{\mathcal{D}}{r^2}\right)\frac{2M}{r^2}\right),
		\end{equation}
		
		\begin{equation}\label{thisisa2cool}
			A_2(r) \doteq -r^{-2}Q^{-1}\partial_r\left(Q^{-1} r^2 \sqrt{1-\epsilon_0(r)\left(1-\frac{2M}{r} + \frac{\mathcal{D}}{r^2}\right) - \epsilon_1(r)\frac{2M}{r}\left(1-\frac{2M}{r}+\frac{\mathcal{D}}{r^2}\right)^2}\right),
		\end{equation}
		
		\begin{equation}\label{thisisa3cool}
			A_3(r) \doteq -Q^{-2}\sqrt{1-\epsilon_0(r)\left(1-\frac{2M}{r} + \frac{\mathcal{D}}{r^2}\right) - \epsilon_1(r)\frac{2M}{r}\left(1-\frac{2M}{r}+\frac{\mathcal{D}}{r^2}\right)^2}.
		\end{equation}

		\subsection{Fourier transform}\label{sec:ftransform} 
		We now take the Fourier transform in time, to obtain, letting $u(\omega,s)= \int_{\RR} e^{ i \omega t}\psi(t,s) dt$:
		\begin{equation}
			\p^2_{s} u= - \omega^2 u + \Big(1-\frac{2M}{r}+\frac{\DD(r)}{r^2}\Big)\Big(m^2 + \frac{L (L+1)}{r^2} +\frac{2M}{r^3}+\frac{r\DD'(r)-2\DD(r)}{r^4}\Big)u.
		\end{equation} We will justify rigorously this operation in Section~\ref{yakov.section}. For now, we record the equations we will later be using in the fixed time-frequency $\omega$ analysis. 
		Calling $V_\omega(s)$ the RHS of the above equation, we can rewrite it as
		\begin{equation}\label{eq:mainV}
			\p^2_{s} u =V_\omega(s) u . 
		\end{equation} For any two solutions $f(\omega,s)$ and $h(\omega,s)$ to~\eqref{eq:mainV} we introduce the Wronskian by
		\[W\left(f,g\right)(\omega) \doteq \partial_sf\  h - \partial_s h\ f.\]
		One may easily check that $\partial_sW = 0$ and so that $W$ may be considered a function of $\omega$. Furthermore, $W$ vanishes if and only if $f$ is a (possibly $\omega$-dependent) constant multiple of $h$.
		
		We note that
		\begin{equation}
			V_\omega(s) = (m^2-\omega^2) -  \frac {2M m^2}{s} + 2M m^2 \Big(\frac{1}{r}-\frac{1}{s}\Big) + \Big(1-\frac{2M}{r}+\frac{\DD(r)}{r^2}\Big)\Big( \frac{L (L+1)}{r^2} +\frac{2M}{r^3}+\frac{r\DD'(r)-2\DD(r)}{r^4}\Big) +\frac{\DD(r) m^2}{r^2}.
		\end{equation}
		We define
		$$
		G(s):= s^2 \Big(  2M m^2 \Big(\frac{1}{r}-\frac{1}{s}\Big) +\Big(1-\frac{2M}{r}+\frac{\DD(r)}{r^2}\Big)\Big( \frac{L (L+1)}{r^2} +\frac{2M}{r^3}+\frac{rD'(r)-2\DD(r)}{r^4}\Big) +\frac{\DD(r) m^2}{r^2}\Big),
		$$
		so that 
		\begin{equation}
			V_{\omega}(s) = (m^2-\omega^2) - \frac {2M m^2}{s} 	+ \frac{G(s)}{s^2}.
		\end{equation}
		Note that for $s$ large enough: \begin{equation}\label{G.properties}
			|G|(s) \lesssim \log(s),\  \Big|\frac{dG}{ds}\Big|(s) \lesssim s^{-1}.
		\end{equation} In the large $s$ regime, \eqref{G.properties} are the only properties of $G$ we will be using.  Note, however, that the positivity of the potential $V_{\omega}(s)>0$ is crucial  in the global part of our argument which rules out the possibility of a resonance at $0$, and is implicit in our use of energy estimates when we establish the basic representation formula for $u$.
		
		\subsection{Regular solutions at the event horizon and infinity}\label{regsolsection}
		
		We introduce the following fundamental solutions (also known as ``Jost solutions'') $u_H(\omega,s)$ and  $u_I(\omega,s)$ as the only solutions respectively regular at the event horizon $\mathcal{H}=\{s=-\infty\}$ and infinity $\mathcal{I}=\{s=+\infty\}$ (up to a multiplicative constant). More precisely, we will show (see Section~\ref{yakov.section}) \begin{equation}\label{uH.boundary}
			u_H(\omega,s) \sim e^{-i \omega s}\text{ as } s \rightarrow -\infty,
		\end{equation}\begin{equation}\label{uI.boundary}
			u_I(\omega,s) \sim  s^{\frac{Mm^2}{\sqrt{m^2-\omega^2}} }e^{-\sqrt{m^2-\omega^2} s} \text{ as } s \rightarrow +\infty.
		\end{equation}

		Here, and in the rest of the paper, the ``$\sim$'' means the following:
		\[h\left(s\right) \sim l\left(s\right)\text{ as }s\to s_0 \text{ if and only if }\lim_{s\to s_0} \frac{h}{l} = 1.\]
		Here $s_0$ may be a finite number or $\pm\infty$.

		\section{Precise statement of the main result}\label{main.result.section}

		We are now ready to state our main theorem. 
		\begin{thm}\label{main.thm}
			Let $g$ be a Lorentzian metric on   $(t^{*},r,\theta,\varphi) \in \mathbb{R} \times (r_+,+\infty) \times \mathbb{S}^2$,  where $r_+>0$, given by \eqref{metriceddi} satisfying the assumptions \ref{D1}, \ref{D2}, \ref{D3}. Let  $r \rightarrow (\phi_0(r),\phi_1(r))$ be smooth, compactly supported on $\{t^{*}=0,\ r_+ \leq r \leq R,\ \omega \in \mathbb{S}^2\}$  for some $R>r_+$. Let $m^2>0$, $L\in \mathbb{Z}_{\geq 0}$ and $\phi= \phi_L(t^{*},r) Y_L(\theta,\varphi)$  to be the solution of the Klein--Gordon equation supported on a single angular mode $L$ with initial data $(\phi_0(r),\phi_1(r))$, i.e.\,  defining $n$ to be the unit normal to $\{t^{*}=0\}$: \begin{align*}
				&  \Box_g \phi = m^2 \phi,\\ & (\phi,n\phi)_{|t^{*}=0}(r,\theta,\varphi)= Y_L(\theta,\varphi) (\phi_0(r),\phi_1(r)). 
			\end{align*}
			Then, there exists  $\gamma_m(M,\DD,m^2) \in \mathbb{C}\setminus \{0\}$ with $|\gamma_m|<1$ such that  for every $\delta>0$, $N \in \mathbb{Z}_{\geq 0}$ and $R_0\geq r_+$, there exists $\check{C}_N(\delta,N,R_0,R,\phi_0,\phi_1,L)>~0$ such that for all $r_+ \leq r  \leq R_0$ and $t^{*}\geq 1$, $\phi_L(t^{*},r)$ obeys the following decay-in-time estimates: \begin{equation}\label{main.formula}\begin{split}
					&  \left|   \rd_{t^{*}}^N\left(r\phi_L(t^{*},r)- \tilde{C}(M,\DD,m^2)  \cdot \bigl| \mathcal{L}[\phi_0,\phi_1] \bigr| \cdot |\uhs|(r) \cdot (t^{*})^{-\frac{5}{6}}\cdot \left[\sum_{q=1}^{+\infty} (-\gamma_m)^{q-1}q^{\frac{1}{3}}\cos(\Psi_q(t^{*})+  \check{\varphi}_H(r)+\theta[\phi_0,\phi_1]) \right]\right) \right| \\ & \leq \check{C}_N \cdot (t^{*})^{-1+\delta},\end{split}
			\end{equation} where \begin{itemize}
				\item $\Psi_q(t^{*})$ is a linear phase with a sublinear correction given by $\Psi_q(t^{*})=  t^{*}\Phi_q(t^{*}) +  \phi_-(M,\DD,m^2)-\frac{\pi}{4}$,   where $\phi_-(M,\DD,m^2)\in \mathbb{R}$ and, defining $\omega_q(t^{*}):=m - \frac{(4\pi^2 m^2 M)^{\frac{1}{3}}}{2} q^{\frac{2}{3}} (t^{*})^{-\frac{2}{3}}$, one has: $$\Phi_q(t^{*})=\omega_q(t^{*})-2\pi m^2M q \cdot (t^{*})^{-1} \left(m^2-\omega_q^2(t^{*})\right)^{-\frac{1}{2}}= m  -\frac{3}{2} [2\pi M]^{\frac{2}{3}} m q^{\frac{2}{3}} (t^{*})^{-\frac{2}{3}}  +O_{m}((t^{*})^{-\frac{4}{3}}).$$ 
				\item $u_H^{*}(r)=|u_H^{*}|(r) e^{i\check{\varphi}_H(r)}:=u_H(\omega=m,r)$, where  $u_H^{*}(r)$  is a solution of \eqref{eq:mainV} which never vanishes and  satisfies  the following: $$u_H^{*}(r) \to 1  \text{ as } r \rightarrow r_+,\    |u_H^{*}|(r) \ls r^{\frac{1}{4}}
				\text{ for } r \geq  r_+.$$
				\item $\mathcal{L}(\cdot,\cdot)= |\mathcal{L}|(\cdot,\cdot) e^{i \theta(\cdot,\cdot)} $ is a (non-trivial)   linear form   defined on  couples $(\phi_0,\phi_1)$ of smooth compactly supported functions,   using the $A_i(r)$ defined in \eqref{commutatorboxftstar}:\begin{equation}
					\mathcal{L}[\phi_0,\phi_1]:=  \int_{r_+}^{+\infty} e^{-ip(r) m}u_H(m,s(r))\left(-imA_1(r) \phi +2A_2(r)\phi + A_1(r) \partial_{t^*}\phi + 2A_3(r)\partial_r\phi\right)|_{t^*=0}r\, dr
				\end{equation} 
				\item $\tilde{C}(M,\DD,m^2):= 2^{-\frac{1}{6}} 3^{-\frac{1}{2}} \pi^{\frac{1}{3}}[1-|\gamma_m|^2]  M^{\frac{2}{3}} m^{\frac{1}{6}}>0$ is a constant.
			\end{itemize}
		\end{thm}
		\begin{rmk}
			Our compact support assumption allows for $(\phi_0(r),\phi_1(r))$ to be supported near $\{r=r_+\}$. Otherwise said, we do not require the initial data to be supported away from the event horizon   $\{r=r_+\}$.
		\end{rmk}
		\begin{rmk}
			While for simplicity we restrict our initial data to be smooth and compactly supported, we reiterate that it is not difficult to adapt our methods to obtain a statement for data with finite regularity and polynomial decay in $r$ at infinity.
		\end{rmk}
		\begin{rmk}
			With little extra work, it follows from our methods that for all $N \in \mathbb{Z}_{\geq 0}$, $M \in \mathbb{Z}_{\geq 0}$ one  has \begin{equation}\label{main.formula2}\begin{split}
					&  \left|    \rd_{r}^M \rd_{t^{*}}^N\left(r\phi_L(t^{*},r)- C(M,\DD,m^2)  \cdot \bigl| \mathcal{L}[\phi_0,\phi_1] \bigr| \cdot |\uhs|(r) \cdot (t^{*})^{-\frac{5}{6}}\cdot \left[\sum_{q=1}^{+\infty} (-\gamma_m)^{q-1}q^{\frac{1}{3}}\cos(\Psi_q(t^{*})+  \check{\varphi}_H(r)+\theta[\phi_0,\phi_1]) \right]\right) \right| \\ & \leq \breve{C}_{N,M} \cdot (t^{*})^{-1+\delta}.\end{split}
			\end{equation} 
		\end{rmk}
		
		\section{Building $u_H$, $u_I$, and establishing the fundamental representation formula}\label{yakov.section}

		The main goal of this section is to prove the following two propositions (which involves the functions $u_H(\omega,s)$ and $u_I(\omega,s)$ constructed below in Section~\ref{wherewemakeuhui}).
		
		The first establishes a representation formula for the Fourier transform of suitable solutions to the inhomogeneous Klein--Gordon equation.
		\begin{prop}\label{repuform}Let $F$ be a smooth function on the Reissner--Nordtr\"{o}m spacetime which is compactly supported within $\{t > c\}$ for some $c \in \mathbb{R}$, and suppose that $\phi$ solves the Klein--Gordon equation~\eqref{KGinrcoord} and moreover vanishes for $t < c$. Then, if we define 
			\[u\left(\omega,s\right) \doteq \int_{-\infty}^{+\infty} e^{i\omega t}\psi\left(t,s\right)\, dt,\qquad \hat{H} \doteq \int_{-\infty}^{+\infty}e^{i\omega t}H\left(t,s\right)\, dt, \]
			(where $\psi$ and $H$ are defined by~\eqref{defpsiandH}) $u$ will satisfy, for $\omega \in [0,\infty)\setminus \{m\}$
			\begin{equation}\label{agoodrepforu}
				u(\omega,s) = W^{-1}(\omega)\left[u_I(\omega,s)\int_{-\infty}^su_H(\omega,S)\hat{H}(\omega,S)\, dS+ u_H(\omega,s)\int_s^{+\infty}u_I\left(\omega,S\right)\hat{H}\left(\omega,S\right)\, dS \right], 
			\end{equation}
			where $W(\omega)$ is the Wronskian of $u_I$ and $u_H$, and satisfies the following two estimates. Firstly, when $\omega > m$, we have for any small $\delta >0$, there exists a constant $C(\delta) > 0$ so that:
			\[\sup_{\omega \in (m,+\infty)}\frac{(1+\omega)^{\frac{3}{4}}\left|\omega-m\right|^{\frac{1}{4}}}{\left|W\right|}+  (1+\omega)^{-2\delta}\left|\omega-m\right|^{\frac{1}{2}+2\delta}\left|\frac{d}{d\omega}\left(\frac{W(\omega)}{(\omega^2-m)^{\frac{1}{4}}}\right)\right| \leq C(\delta).\]
			Secondly, for any small $\delta > 0$, there exists a constant $C'(\delta) > 0$ so that 
			\[\sup_{\omega \in [0,m-\delta]}\left|W\right|^{-1}(\omega)+ \left|\frac{dW}{d\omega}\right|(\omega) \leq C'(\delta).\]

			Finally, we have that for any compact sets $K_1,K_2 \subset \mathbb{R}$, $u\left(\omega,s\right) \in L^{\infty}_{\omega \in K_1}L^{\infty}_{s \in K_2}$.
			
		\end{prop}
		\begin{rmk}For simplicity of notation we restrict our attention to $\omega \geq 0$. However, since the scalar field $\psi$ is real valued, we have that $u\left(-\omega,s\right) = \overline{u\left(\omega,s\right)}$ and thus one may immediately obtain the corresponding formulas for $\omega < 0$.
		\end{rmk}
		\begin{rmk}By a density argument it is straightforward to see that one can replace the assumption that $F$ vanishes for large $r$ with, for example, the assumption that $F$ has a finite norm for which smooth compactly supported functions are dense, which guarantees that $\psi$ is uniformly bounded, and which guarantees that the integrals in~\eqref{agoodrepforu} converge absolutely.  
		\end{rmk}
		
		\begin{rmk}\label{horizonrepform}By working in regular coordinates  $(t^*,r)$ one can derive also an analogue of~\eqref{agoodrepforu} which holds up to the event horizon. Namely, in the regular $(t^*,r)$ coordinates, we may define
			\begin{equation}
				\check{u}\left(\omega,r\right) \doteq e^{i\omega p(r)}u\left(\omega,r\right) = \int_{-\infty}^{+\infty}e^{i\omega t^*}\psi\left(t^*,r\right)\, dt^*,\qquad \check{H}\left(\omega,r\right) \doteq  e^{i\omega p(r)}\hat{H} =  \int_{-\infty}^{+\infty}e^{i\omega t^*}H\left(t^*,r\right)\, dr.  
			\end{equation}
			We then have 
			\begin{align}\label{agoodrepformonthehorizon}
				&\check{u}\left(\omega,r\right) = W^{-1}(\omega)\left[\check{u}_I(\omega,r)\int_{-\infty}^{s(r)}u_H(\omega,S)\hat{H}(\omega,S)\, dS+ \check{u}_H(\omega,r)\int_{s(r)}^{+\infty}u_I\left(\omega,S\right)\hat{H}\left(\omega,S\right)\, dS \right],
			\end{align}
			where $\check{u}_I\left(\omega,r\right) \doteq  e^{i\omega p(r)}u_I\left(\omega,s\left(r\right)\right)$ and $\check{u}_H\left(\omega,r\right) \doteq  e^{i\omega p(r)}u_H\left(\omega,s\left(r\right)\right)$. A key point of this form of the representation formula is that we expect $\partial_{\omega}\check{u}_H$ uniformly bounded in any region $r \in [r_+,R_0]$ for $R_0 < \infty$. 
		\end{rmk}
		\begin{rmk}\label{switcharoo}It will be useful to observe the (immediate) fact that the formulas~\eqref{agoodrepforu} and~\eqref{agoodrepformonthehorizon} will continue to hold if $u_I$ and $\tilde{u}_I$ are both multiplied by any non-vanishing function of $\omega$. 
		\end{rmk}
		The second main result of this section uses the WKB method to establish estimates for the solution $u_I$ which cover the whole range of $\omega \in (m,+\infty)$.
		\begin{prop}\label{betterbounduI}For any $\omega \in (m,+\infty)$, let $u_I\left(\omega,s\right)$ be the unique solution to~\eqref{eq:mainV} which satisfies
			\[u_I =  r^{\frac{Mm^2}{-i\sqrt{\omega^2-m^2}}} e^{i\sqrt{\omega^2-m^2} s}\left(1+ O_{\omega}\left(r^{-1}\right)\right) \text{ as }s \to \infty.\]
			Furthermore, choose $A \gg 1$ sufficiently large (independent of $\omega$) so that, in particular, $V_{\omega}$ is strictly negative for all $s \in [A,\infty)$.
			
			Then, there exists a function $h\left(\omega,s\right)$, defined for $s \in [A,\infty)$, so that we may write
			\[u_I = \left(\omega^2-m^2\right)^{1/4}\left(-V_{\omega}\right)^{-1/4}\exp\left(i\int_A^s\left(-V_{\omega}\right)^{1/2}\, dx\right)\left(1+h\right),\]
			and so that for any $\delta > 0$, there exists a constant $C(\delta)$ independent of $\omega$ so that for all $s \in [A,\infty)$: (recall that $\xi=\omega$ and $\rd_{\xi}$ is expressed in the $(\xi,s)$ coordinates):
			\[\left|h\right| + \left|\partial_sh\right| \leq C(\delta)r^{-\frac{1}{2}},\qquad \left|\partial_{\xi }h\right| + \left|\partial_{\xi s}^2h\right| \leq C(\delta)\ \left(\left(\omega-m\right)^{-1/2-2\delta}r^{-\delta}+  r\right).\]
		\end{prop}
		\begin{rmk}Since this frequency range ultimately does not contribute to the leading order behavior as $t\to \infty$, we have made no attempt to provide sharp estimates in Proposition~\ref{betterbounduI}, either as $\omega \to_{\omega >m} m$ or as $r\to +\infty$. 
		\end{rmk}
		\subsection{Construction of $u_H$ and $u_I$}\label{wherewemakeuhui}
		In the following lemma we invoke the theory of regular singularities to define the function $u_H(\omega,s)$. The arguments are standard, but we provide a sketch of the proof for the convenience of the reader.
		\begin{lemma}\label{makeuh}Let $\mathbb{S} \doteq \left\{z \in \mathbb{C} : \Im(z) > -\kappa_+\right\}$. There exists a unique function $u_H\left(\omega,s\right) : \mathbb{S} \times \mathbb{R} \to \mathbb{C}$ which solves~\eqref{eq:mainV} and satisfies 
			\[u_H = e^{-i\omega s}\left(1+o\left(1\right)\right)\text{ as }s \to-\infty.\]
			When $\Im(\omega) > 0$, then $u_H$ is, up to multiplication by a constant, the unique solution which is uniformly bounded as $s\to -\infty$.  If $\omega \neq 0$, then $u_H$ never vanishes.
			
			Moreover, considering $s$ to be a function of $r$, we have that for each fixed $\omega$,  $e^{i\omega s(r)}u_H(\omega,s(r))$ is a  smooth function for $r \in [r_+,\infty)$, and for fixed $r$, $Y\left(\omega,r\right) \doteq e^{i\omega s(r)}u_H(\omega,s(r))$ is a holomorphic function of $\omega \in \mathbb{S}$. Finally, we have the following uniform bounds for $Y$ and its derivatives on any compact set $K \subset [r_+,\infty)$ and $\omega$ with $\Im(\omega) \geq 0$:
			\begin{equation}\label{boundforY}
				\sum_{i,j=0}^1\sup_{\Im(\omega) \geq 0, r\in K}\left|\partial_r^i\partial_{\omega}^jY\right| \lesssim 1. 
			\end{equation}
		\end{lemma}
		\begin{proof}
			From~\eqref{KGinrcoord} and a short calculation, we see that for $r \in (r_+,\infty)$
			\[\frac{d^2u}{dr^2} + \left(\frac{1}{r-r_+} + O\left(1\right)\right)\frac{du}{dr} + \left(\frac{\omega^2\kappa_+^{-2}}{4(r-r_+)^2} + O\left((r-r_+)^{-1}\right)\right)u = 0.\]
			Thus we have a regular singularity at $r = r_+$ (see Section 4 of Chapter 5 of~\cite{olver}). The corresponding indicial roots are $\pm \frac{i\omega}{2\kappa_+}$. It thus follows that as long as $\frac{i\omega}{\kappa_+}$ does not coincide with a negative integer, there will exist a solution of the form
			\[u\left(r,\omega\right) = (r-r_+)^{\frac{i\omega}{2\kappa_+}}R_H\left(r,\omega\right),\]
			for some $R_H$ which is  smooth  in $r \in [r_+,+\infty)$. It furthermore immediately follows that for fixed r, $R_H$  is holomorphic for $\omega \in \mathbb{S}$. Switching from the $r$ variable to the $s$ variable then easily leads to the definition of $u_H$.
			
			To see that $u_H$ never vanishes vanishes when $\omega \neq 0$, we consider $\mathscr{Q} \doteq \Im\left(\frac{du_H}{ds}\overline{u_H}\right)$. We have that $\frac{d\mathscr{Q}}{ds} = 0$ and that $\mathscr{Q}|_{s=-\infty} = -\omega \neq 0$. We immediately conclude that $u_H$ cannot vanish anywhere.
			
			It remains only to establish the bound~\eqref{boundforY}. For this it is convenient to instead use the equation~\eqref{eq:mainV}.  Let us define a function $h(\omega,s)$ by $u_{H}(\omega,s) = e^{-i\omega s}\left(1+h(\omega,s)\right)$. We have that $h$ is a solution to
			\begin{equation}\label{aniceeqnforhwow}
				\partial_{s}^2h - 2i\omega \partial_{s}h - \left(V_{\omega}-\omega^2\right)h = \left(V_{\omega}-\omega^2\right).
			\end{equation}
			In turn, this is equivalent to the Volterra integral equation
			\begin{equation}\label{anothereqnforhthingsaregood}
				h\left(\omega,s\right) = \int_{-\infty}^{s}G\left(\omega,s-S\right)\left(V_{\omega}(S)-\omega^2\right)\left(1+h\left(\omega,S\right)\right)\, dS,
			\end{equation}
			where
			\[G\left(\omega,y\right) \doteq \frac{1-e^{2i\omega y}}{-2i\omega}\text{ for }\omega \neq 0,\qquad G\left(0,y\right) = y.\]
			Keeping in the face that $\left(V_{\omega}-\omega^2\right)\left(S\right) \sim \exp\left(2\kappa_+ S \right)$ as $S \to -\infty$, the proof is then concluded  by an application of Theorems~\ref{thm:volterra} and~\ref{thm:volterraparam} (and Remarks~\ref{switchorder} and~\ref{switchorderparam}) and the formulas which relate $\frac{\partial}{\partial r}$ and $\frac{\partial}{\partial s}$. We note in particular that, in the notation of Theorems~\ref{thm:volterra} and~\ref{thm:volterraparam}, one may take $P_1$, $P_2$, and $Q$ to be constant and $\Phi  = \tilde{\Phi}= e^{2\kappa_+s}$.
		\end{proof}
		
		Next we turn to the construction of $u_I\left(\omega,s\right)$. Again the arguments are standard, but we provide a sketch of the proof for the convenience of the reader.
		\begin{lemma}\label{makeuI}Let $\mathbb{S} \doteq \left\{\omega \in \mathbb{C}\setminus\{m\} : \Re(\omega) \geq 0\text{ and } \Im(\omega) \geq 0\right\}$, and let us define a branch of $\sqrt{\cdot}$ by placing a branch cut along the negative imaginary axis and letting $\sqrt{\cdot}$ agree with its standard value along $\mathbb{R}_{>0}$. Then there exists a unique function $u_I\left(\omega,s\right) : \mathbb{S} \times \mathbb{R} \to \mathbb{C}$ which solves~\eqref{eq:mainV}, satisfies 
			\[u_I =  r^{\frac{Mm^2}{-i\sqrt{\omega^2-m^2}}} e^{i\sqrt{\omega^2-m^2} s}\left(1+ O\left(r^{-1}\right)\right) \text{ as }s \to +\infty,\]
			is holomorphic for $\omega \in {\rm Int}\left(\mathbb{S}\right)$ and satisfies that 
			\[\sup_{\omega \in K}\sup_{s \in [s_0,+\infty)}\sum_{i,j=0}^1\left|\partial_s^i\partial_{\omega}^j\left(r^{\frac{Mm^2}{i\sqrt{\omega^2-m^2}}} e^{-i\sqrt{\omega^2-m^2} s}u_I\right)\right| \lesssim_{K,s_0} 1,\]
			where $K$ is any compact subset of $\mathbb{S}$ and $s_0 \in (-\infty,+\infty)$.
			
			Finally, it is useful to observe that when $\Im(\omega) > 0$, then $u_I$ is the unique solution to~\eqref{eq:mainV} which is uniformly bounded as $s\to +\infty$. 
		\end{lemma}
		\begin{proof}We start by observing that the equation for $u$ may be written as
			\begin{equation}\label{largerformeqnforu}
				\frac{d^2u}{dr^2} +O(r^{-2})\frac{du}{dr}+ \left(\omega^2-m^2 + \frac{2M(2\omega^2-m^2)}{r} + O(r^{-2})\right)u = 0.
			\end{equation}
			One may then follow the well-known procedure for constructing formal expansions solving~\eqref{largerformeqnforu} (see Section 1 of Chapter 7 of~\cite{olver}). In particular, there exists $\{a_n\left(\omega\right)\}_{n=1}^{\infty}$, so that if we define a function $p_N(\omega,r)$ by
			\[p_N(\omega,r) \doteq e^{i\sqrt{\omega^2-m^2} r}r^{\frac{M(2\omega^2-m^2)}{-i\sqrt{\omega^2-m^2}}}\left(1+ \sum_{n=1}^N\frac{a_n}{r^n}\right),\]
			then we have that the functions $\{a_n\}_{n=1}^N$ are holomorphic in ${\rm Int}\left(\mathbb{S}\right)$, $\left(a_n,\partial_{\omega}a_n\right)$ are uniformly bounded on any compact set $K \subset \mathbb{S}$, and, if we denote by $\mathscr{L}$ the operator $\partial_s^2 - V_{\omega}$, we have
			\[\mathscr{L}p_N = e^{i\sqrt{\omega^2-m^2} r}r^{\frac{M(2\omega^2-m^2)}{-i\sqrt{\omega^2-m^2}}}B_N\left(\omega,r\right),\]
			where $B_N$ is a rational function of $r$, is holomorphic in $\omega \in {\rm Int}\left(\mathbb{S}\right)$, and satisfies
			\[\sup_K\left(\left|B_N\right| + \left|\partial_{\omega}B_N\right|\right)\lesssim_K r^{-N},\]
			where, again, $K$ is any compact subset of $\mathbb{S}$. We note further the useful fact that $\omega \in \mathbb{S}$ implies that
			\begin{equation}\label{someusefulpositivity}
				\Im\left(\sqrt{\omega^2-m^2}\right) \geq 0.
			\end{equation}
			
			Letting $N$ be sufficiently large, we thus look for a solution $u$ of the form $u = p_N + \epsilon_N$. Variation of parameters then leads to trying to find $\epsilon_N$ by solving a suitable Volterra equation:
			\begin{equation}\label{eqnforepsilonNtogetuI}
				\epsilon_N\left(\omega,r\right) = \int_r^{+\infty}G\left(\omega,r,R\right)\left(\left(\mathscr{L}p_N\right) + A_1\left(\omega,R\right)\frac{\partial \epsilon_N}{\partial R} + A_2\left(\omega,R\right)\epsilon_N\right)\, dR,
			\end{equation}
			where
			\[G\left(\omega,r,R\right) \doteq \frac{e^{-i\sqrt{\omega^2-m^2}(R-r)} - e^{i\sqrt{\omega^2-m^2}(R-r)}}{2i\sqrt{\omega^2-m^2}},\]
			and the functions $A_1$ and $A_2$ satisfy that 
			\[\sup_K\left(r\left|A_1\right| + \left|A_2\right| + r\left|\partial_{\omega}A_1\right|+\left|\partial_{\omega}A_2\right|\right)\lesssim_K r^{-1},\]
			for any compact subset $K\subset \mathbb{S}$.
			
			It suffices to solve the Volterra equation for all sufficiently large $r$. In order to do this, we apply Theorems~\ref{thm:volterra} and~\ref{thm:volterraparam} while also keeping in mind Remarks~\ref{volterramodify}, \ref{volterramodifyparam}, and~\ref{volterraholomorphic}. More specifically, letting $\tilde{N}$ satisfy $1 \ll \tilde{N} \ll N$ and $A \gg 1$ depending on the choice of compact set $K$, we may make the following choices:
			\[P_j(r) = \tilde{P}_j(r) = A\exp\left(-\Im\left(\sqrt{\omega^2-m^2}\right)r\right),\qquad Q(r) = \tilde{Q}(r) = A\exp\left(\Im\left(\sqrt{\omega^2-m^2}\right)r\right),\]
			\[\Phi = \tilde{\Phi} = A r^{-\tilde{N}}.\]

			Finally, the fact that $u_I$ is the unique solution which is bounded as $s\to+\infty$ is a consequence of the fact that any solution which is not a multiple of $u_I$ must blow-up exponentially fast as $s\to+ \infty$ (see Section 1 of Chapter 7 of~\cite{olver}). 
		\end{proof}
		
		\subsection{Wronskian estimates}
		In this section we establish various estimates for the Wronskian $W$ of $u_H$ and $u_I$. We start by showing that the Wronskian of $u_H$ and $u_I$ never vanishes for $\omega \in (m,\infty)$, and moreover establish a quantitative estimate for $W^{-1}$.
		\begin{lemma}\label{wronknotzerobigm}Let $\omega \in (m,+\infty)$. Then the Wronskian $W=W(u_I,u_H)$ of $u_I$ and $u_H$ satisfies the estimate
			\[\sup_{\omega \in (m,\infty) }\frac{\omega^{\frac{3}{4}}\left|\omega-m\right|^{\frac{1}{4}}}{\left|W\right|(\omega)} \lesssim 1.\]
		\end{lemma}
		\begin{proof}
			
			We have that $u_H$ and $\overline{u_H}$ form a basis of solutions to~\eqref{eq:mainV}, and moreover, one has 
			\[W\left[u_H,\overline{u_H}\right] = -2i\omega,\]
			and thus, there exists $\alpha(\omega) \in \mathbb{C}$ so that
			\begin{equation}\label{decompuitouh}
				u_I = \alpha u_H - (2i\omega)^{-1}W \overline{u_H}.
			\end{equation}
			We then set
			\[Q_T \doteq \Im\left(\rd_{s}u_I\overline{u}_I\right).\]
			A short calculation using $u_I$'s equation yields that $\rd_{s}Q_T(\omega,s) = 0$. On the other hand, in view of the asymptotics of $u_H$ and $u_I$ we have that
			\[\lim_{s\to +\infty}Q_T\left(s\right) = \sqrt{\omega^2-m^2} > 0,\qquad \lim_{s\to -\infty}Q_T\left(s\right) = -\omega \left|\alpha\right|^2  + \frac{1}{4\omega}\left|W\right|^2.\]
			Equating these two terms leads to the estimate
			\[\left|W\right|^{-1} \leq \omega^{-\frac{3}{4}}\left(\omega-m\right)^{-\frac{1}{4}}.\]
		\end{proof}

		Next we consider the case when $\omega \in [\delta,m-\delta]$ for some sufficiently small $\delta> 0$.
		\begin{lemma}\label{wronknotzerosmallbutnottoosmallm}Let $\delta > 0$ be any sufficiently small constant, and let $\omega \in [\delta,m-\delta]$. Then the Wronskian $W$ of $u_I$ and $u_H$ satisfies the estimate
			\[\sup_{\omega \in [\delta,m)}\left|W\right|^{-1} \leq C(\delta).\]
			The constant $C(\delta)$ is, in principle, possible to make explicit.
		\end{lemma}
		\begin{proof}As in the proof of Lemma~\ref{wronknotzerobigm} we have that~\eqref{decompuitouh} holds. In this case, however, because $u_I$ is real valued, we must in fact have that $\alpha = (2i\omega)^{-1}\overline{W}$:
			\begin{equation}\label{expanduiinrealcase}
				u_I = \frac{\overline{W}}{2i\omega}u_H - \frac{W}{2i\omega}\overline{u_H}.
			\end{equation}
			Note that~\eqref{expanduiinrealcase} already implies that $W$ cannot vanish, but does not yet provide a quantitative estimate for $W^{-1}$. In order to establish a quantitative estimate, we define
			\begin{equation}\label{qycurrent}
				\mathscr{Q}(\omega,s) \doteq y(s)\left|\rd_{s}u_I\right|^2 + \omega^2 y(s) \left|u_I\right|^2,
			\end{equation}
			where $y\left(s\right) \doteq \exp\left(-C\left(r(s)-r_+\right)\right)$ for a sufficiently large constant $C>0$ depending on $\delta$. Using the largeness of $C$ and that the fact that $V_{\omega} + \omega^2$ vanishes at $r=r_+$, we have that  
			\begin{align}\label{qyinc}
				\rd_s\mathscr{Q}(\omega,s) = y'(s)\left|\rd_s u_I\right|^2 + \omega^2 y'(s)\left|u_I\right|^2 + 2y(s)\left(V_{\omega}+\omega^2\right)\Re\left(\rd_s u_I\overline{u_I}\right) < 0.
			\end{align}
			
			On the other hand, using~\eqref{expanduiinrealcase}, we have
			\begin{equation}
				\lim_{s\to -\infty}\mathscr{Q}\left(\omega,s\right) = \left|W\right|^2(\omega) \Rightarrow \mathscr{Q}\left(\omega,s\right) < \left|W\right|^2(\omega),\qquad \forall s \in (-\infty,+\infty).
			\end{equation}
			To conclude the proof, it suffices to observe that as a consequence of Lemma~\ref{makeuI}, one may find a sufficiently large $s$, depending on $\delta$, so that for all $\omega \in [\delta,m-\delta]$, $\mathscr{Q}(s) \gtrsim 1$.
		\end{proof}

		Finally, we consider the case when $\omega \in [0,\delta]$.
		\begin{lemma}\label{wronknotzerosmallm}Let $\delta > 0$ be any sufficiently small constant, and let $\omega \in [0,\delta]$. Then the Wronskian $W=W(u_I,u_H)$ of $u_I$ and $u_H$ satisfies the estimate
			\[\sup_{\omega \in [0,\delta]}\left|W\right|^{-1}(\omega) \lesssim 1.\]
		\end{lemma}
		\begin{proof}It follows from Lemmas~\ref{makeuh} and~\ref{makeuI} that $W(\omega)$ is $C^1$ for $\omega \in [0,\delta]$. Thus it suffices to estimate $W^{-1}$ at $\omega = 0$.
			
			When $\omega = 0$, the formula~\eqref{expanduiinrealcase} is clearly invalid. In order to find a replacement, we define a function
			\[\check{u}_H\left(s\right) \doteq \lim_{\omega \to 0^+}\frac{\overline{u_H}\left(\omega,s\right) - u_H\left(\omega,s\right)}{2i\omega}.\]
			It is straightforward to see that $\check{u}_H$ is a well-defined real-valued solution to~\eqref{eq:mainV} when $\omega = 0$, that
			\[W\left[u_H,\check{u}_H\right] = -1, \]
			and that $\check{u}_H \sim s$ as $s \to -\infty$. The analogue of~\eqref{expanduiinrealcase} is then (keeping in mind that $u_H$ is real valued when $\omega = 0$)
			\begin{equation}\label{omega0uinearh}
				u_I = \alpha u_H +W \check{u}_H,
			\end{equation}
			for some constant $\alpha(\omega) = W\left[\check{u}_H,u_I\right] \in \mathbb{R}$, which furthermore satisfies
			\begin{equation}\label{keepthisinmindalpha}
				\left|\alpha(\omega)\right| \lesssim 1.
			\end{equation}
			This implies, in particular, that 
			\[ \rd_s u_Iu_I = \left(s+O(1)\right) W^2 + \alpha W O\left(se^{2\kappa_+ s}\right)+ \alpha^2\cdot O\left(e^{2\kappa_+s}\right)\text{ as }s\to -\infty.\]
			
			(note that we have used the fact that $r\rightarrow u_H(0,r)$ is smooth, hence $|u_H(0,s)-1| = O(e^{2\kappa_+s })$ as $s\rightarrow -\infty$).  We also have, in view of the decay of $u_I$ at $s = +\infty$, that
			\[\lim_{s\to +\infty} \rd_s u_Iu_I = 0.\]

			We next observe that when $\omega = 0$ we have
			\[V_0 = \left(1-\frac{2M}{r}+\frac{\mathcal{D}(r)}{r^2}\right)\left(m^2 + \frac{L(L+1)}{r^2}+\frac{2M}{r^3}+\frac{r\DD'(r)-2\DD(r)}{r^4}\right) > 0.\]
			In particular, the asymptotics of $u_I$ imply that 
			\[\int_0^{+\infty} \left(\left|\partial_su_I\right|^2+ V_0\left|u_I\right|^2\right)\, ds \gtrsim 1.\]
			
			Thus, for any $s_0 \in (-\infty,+\infty)$
			\begin{align}
				0 = \int_{s_0}^{+\infty}\left(-\partial_s^2u_I + V_0u_I\right)u_I\, ds = \int_{s_0}^{+\infty}\left[\left(\partial_su_I\right)^2 +V_0(u_I)^2\right]\, ds + u'_Iu_I|_{s=s_0}.
			\end{align}
			This then implies that, for any $s$  sufficiently negative  (depending on $\omega$)
			\begin{equation}\label{isthisthend}
				\left((-s)+O(1)\right) W^2(\omega) + \alpha(\omega) W(\omega) O\left(se^{2\kappa_+s}\right)+ \alpha^2(\omega) \cdot O\left(e^{2\kappa_+s}\right) \gtrsim 1,
			\end{equation}
			from which, keeping in mind~\eqref{keepthisinmindalpha}, the proof is easily concluded. 
		\end{proof}

		\subsection{Exponential damping and  derivation of the representation formula}
		In this section we will prove Proposition~\ref{repuform}.
		\begin{proof}
			We first note that it is a straightforward consequence of energy estimates and Sobolev inequalities that $\psi$ is pointwise uniformly bounded on the spacetime. With this in mind, for every $\epsilon > 0$, we define 
			\[\psi_{\epsilon} \doteq e^{-\epsilon t}\psi,\qquad H_{\epsilon} \doteq e^{-\epsilon t}H,\]
			recalling the notation $H$ from \eqref{defpsiandH}, and note that we will have for all $t\geq 0$ (recalling that $\psi_{\ep}=0$ for $t<0$)
			\[\left|\psi_{\ep}\right| \lesssim e^{-\epsilon t}.\]
			Keeping also in mind the fact that $\psi$ and $H$ is supported for $\{t > 0\}$, we may define
			\[u_{\epsilon}\left(\omega,s\right) \doteq \int_{-\infty}^{+\infty} e^{i\omega t}\psi_{\epsilon}\left(t,s\right)\, dt = \int_{-\infty}^{+\infty} e^{i\left(\omega+i\epsilon\right) t}\psi\left(t,s\right)\, dt,\]
			\[\hat{H}_{\epsilon}(\omega,s) \doteq \int_{-\infty}^{+\infty}e^{i\omega t}H_{\epsilon}\left(t,s\right)\, dt =\int_{-\infty}^{+\infty}e^{i\left(\omega+i\epsilon\right) t}{H}\left(t,s\right)\, dt.\]
			For each $\omega \in \mathbb{R}_{>0}\setminus \{m\}$ it is convenient to set $\omega_{\epsilon} \doteq \omega + i\epsilon$. Then, for each $\epsilon > 0$, we have that $u_{\epsilon}(\omega,s)$ is a continuous function of $\omega$ which satisfies
			\begin{equation}\label{uepeqn}
				\partial_s^2u_{\epsilon} - V_{\omega_{\epsilon}}u_{\epsilon} = H_{\epsilon}.
			\end{equation}
			On the other hand, $u\left(\omega,s\right) = \lim_{\epsilon \to 0}u_{\epsilon}\left(\omega,s\right)$ is, a priori, only defined as a distribution in $\omega$. 
			
			Now fix $\omega \in \mathbb{R}_{>0}\setminus \{m\}$ and let $\epsilon > 0$ be sufficiently small depending on $\omega$ so that $W_{\epsilon} \doteq \partial_su_I\left(\omega_{\epsilon},s\right) u_H\left(\omega_{\epsilon},s\right) - \partial_su_H\left(\omega_{\epsilon},s\right) u_I\left(\omega_{\epsilon},s\right) \neq 0$. Then we may define a function $\check{u}_{\epsilon}(s)$ by
			\begin{equation}\label{hatuthisisit}
				\check{u}_{\epsilon}(\omega,s) \doteq W_{\epsilon}^{-1}(\omega)\left[u_I\left(\omega_{\epsilon},s\right)\int_{-\infty}^su_H\left(\omega_{\epsilon},S\right)\hat{H}_{\epsilon}\left(\omega,S\right)\, dS + u_H\left(\omega_{\epsilon},s\right)\int_s^{+\infty}u_I\left(\omega_{\epsilon},S\right)\hat{H}_{\epsilon}\left(\omega,S\right)\, dS\right]. 
			\end{equation}
			Then $\underline{u}_{\epsilon} \doteq u_{\epsilon} - \check{u}_{\epsilon}$ will solve
			\begin{equation}\label{yougetthis}
				\partial_s^2\underline{u}_{\epsilon} - V_{\omega_{\epsilon}}\underline{u}_{\epsilon} = 0.
			\end{equation}
			Furthermore, a straightforward analysis shows that $\check{u}_{\epsilon}$ and hence also $\underline{u}_{\epsilon}$ is uniformly bounded as $s \to \pm\infty$. In turn, this implies that $\check{u}_{\epsilon}$ is both a constant multiple of $u_I$ and a constant multiple of $u_H$ (note that $u_I(\omega_\ep,s)$ (respectively $u_H(\omega_\ep,s)$) is the only bounded solution of \eqref{yougetthis} at $s=+\infty$ (respectively $s=-\infty$) up to a multiplicative factor). However, since the Wronskian of $u_I$ and $u_H$ is non-zero, we conclude that $\check{u}_{\epsilon}$ vanishes and thus that
			\begin{equation}\label{notjustforhatnow}
				u_{\epsilon}(s) \doteq W_{\epsilon}^{-1}\left[u_I\left(\omega_{\epsilon},s\right)\int_{-\infty}^su_H\left(\omega_{\epsilon},x\right)H_{\epsilon}\left(\omega,x\right)\, dx + u_H\left(\omega_{\epsilon},s\right)\int_s^{\infty}u_I\left(\omega_{\epsilon},x\right)H_{\epsilon}\left(\omega,x\right)\, dx\right]. 
			\end{equation}
			We then conclude the proof when $\omega \neq m$ by taking $\epsilon \to 0$.
			
			It remains to establish an $L^{\infty}$ bound for $u$.  In view of the representation formulas we have established, we may restrict attention to $\omega$ close to $m$. For any such $\omega$ the above arguments show that when $\epsilon > 0$, the formula~\eqref{notjustforhatnow} holds. Our goal is to establish uniform estimates for $u_{\epsilon}$ as $\epsilon \to 0$.  
			
			We start by writing the equation~\eqref{uepeqn} as
			\begin{equation}\label{eqnwithcomplexomega}
				\partial^2_su_{\epsilon} + \left(\omega_{\ep}^2 - \mathring{V}\right)u_{\epsilon} = H_{\epsilon},
			\end{equation}
			where $\mathring{V} = V_{\omega_{\epsilon}} +\omega_{\epsilon}^2$, is real valued and satisfies $\mathring{V} = O\left(r-r_+\right)$ as $r\to r_+$,  $\sup_{s \in \mathbb{R}}\left|\mathring{V}\right| \lesssim 1$, and that $\mathring{V} > 0$ for all $s \in \mathbb{R}$. We start by multiplying both sides of~\eqref{eqnwithcomplexomega} with $\overline{\omega_{\epsilon}u_{\epsilon}}$, integrating by parts, and taking the imaginary part. We obtain
			\begin{equation}\label{energyestincomplexplane}
				\epsilon \int_{-\infty}^{+\infty}\left[\left|\partial_su_{\epsilon}\right|^2 + \left(\left|\omega_{\ep}\right|^2 +\mathring{V}\right)\left|u_{\epsilon}\right|^2\right]\, dS = \int_{-\infty}^{+\infty}\Im\left(H\overline{\omega_{\epsilon}u_{\epsilon}}\right)\, dS.
			\end{equation}
			Another useful estimate is obtained by multiplying~\eqref{eqnwithcomplexomega} with $-y \overline{\partial_su_{\epsilon}}$ for a suitable $C^1$ function $y(s)$, and then integrating by parts and taking the real part. Assuming that $y$ is bounded as $s \to \pm\infty$, we obtain
			\begin{equation}\label{babycomplexmora}
				\int_{-\infty}^{+\infty}\left[\frac{1}{2}y'\left|\rd_{s} u_{\epsilon}\right|^2 + \frac{1}{2}y'\omega^2\left|u_{\epsilon}\right|^2 + yV\Re\left(u_{\epsilon}\overline{\partial_su_{\epsilon}}\right) - y\Re\left(\left(\omega_{\epsilon}^2-\omega^2\right)u_{\epsilon}\overline{\partial_su_{\epsilon}}\right)\right]\, dS = -\int_{-\infty}^{+\infty} y\Re\left(H\overline{\partial_su_{\epsilon}}\right)\, dS.
			\end{equation}
			There exists a function $\lambda(s)$ such that $s^{-2}\lambda^{-1}(s) \sim  1$ for $s \leq -1$ and $\sup_{s \geq -1}\left|\lambda^{-1}\right| \lesssim 1$, so that for  any large constant $C_0>0$ we can set $y(s) = -\exp\left(-C_0 \int_{-\infty}^s \lambda(S)\, dS\right)$ and we will have 
			\begin{equation}\label{whatcouldyevenbewhoknows}
				y'(s) = -y(s) C_0 \lambda(s) > 0,\qquad y(-\infty) = -1,\qquad y(+\infty) = 0.
			\end{equation}
			In particular, since $\left|\omega_{\epsilon}^2-\omega^2\right| \lesssim \epsilon$, and  $H$ vanishes for large $s$, one may combine~\eqref{energyestincomplexplane},~\eqref{babycomplexmora}, and~\eqref{whatcouldyevenbewhoknows} to obtain, for any compact set $K \subset \mathbb{R}$,
			\[\sup_{\epsilon > 0}\int_K\left[\left|\partial_su_{\epsilon}\right|^2 + \left|u_{\epsilon}\right|^2\right]\, dS \lesssim_{H,K} 1.\]
			This establishes that $u$ is bounded  as a map from $\omega$ to $L^2_{\rm loc}\left(s\right)$. It is then straightforward establish commuted versions of this estimate and  conclude the proof with Sobolev inequalities in $s$.
			
		\end{proof}
		\subsection{An improved estimate for $u_I$ when $\omega > m$}
		In this section we prove Proposition~\ref{betterbounduI}.
		\begin{proof}We start with estimates for the potential $V_{\omega}$ (defined in Section~\ref{setup.section}) under the assumption that $\omega \in (m,+\infty)$ and $s \in [A,+\infty)$ for sufficiently large $A$ (note that $s(r)\sim r$ in this range). We start with the trivial estimates that
			\[\left(-V_{\omega}\right)^{-1} \lesssim r,\qquad \left(-V_{\omega}\right)^{-1} \lesssim \frac{1}{\omega^2-m^2}.\]
			Keeping this in mind, we may easily establish that   for any $\delta>0$ small:
			\begin{equation}\label{potentialestimatesforWKBomegabigm}
				\left|\left(-V_{\omega}\right)^{-1/4}\frac{\partial^2}{\partial s^2}\left(-V_{\omega}\right)^{-1/4}\right| \lesssim r^{-3/2},
			\end{equation}
			\begin{equation}\label{potentialestimatesWKBwithomega}
				\left|\left(-V_{\omega}\right)^{1/2}\frac{\partial}{\partial \omega}\left(\left(-V_{\omega}\right)^{-3/4}\frac{\partial^2}{\partial s^2}\left(\left(-V_{\omega}\right)^{-1/4}\right)\right)\right| \lesssim \left(\omega^2-m^2\right)^{-1/2-\delta}r^{-1-\delta}, \qquad \text{ if }\omega \lesssim 1,
			\end{equation}
			\begin{equation}\label{potentialestimatesWKBwithomega2}
				\left|\left(-V_{\omega}\right)^{1/2}\frac{\partial}{\partial \omega}\left(\left(-V_{\omega}\right)^{-3/4}\frac{\partial^2}{\partial s^2}\left(\left(-V_{\omega}\right)^{-1/4}\right)\right)\right| \lesssim \omega^{-4}r^{-3},\qquad \text{ if }\omega - m \gtrsim 1,
			\end{equation}
			\begin{equation}\label{potentialestimatesWKBwithomega3}
				\omega\left|\left(\int_A^s\left(-V_{\omega}\right)^{-1/2}\, dS\right) \frac{\partial}{\partial s}\left(\left(-V_{\omega}\right)^{-3/4}\frac{\partial^2}{\partial s^2}\left(\left(-V_{\omega}\right)^{-1/4}\right)\right)\right| \lesssim \left(\omega^2-m^2\right)^{-1/2-2\delta}r^{-1-\delta}, \qquad \text{ if }\omega \lesssim 1,
			\end{equation}
			\begin{equation}\label{potentialestimatesWKBwithomega4}
				\omega \left|\left(\int_A^s\left(-V_{\omega}\right)^{-1/2}\, dS\right) \frac{\partial}{\partial s}\left(\left(-V_{\omega}\right)^{-3/4}\frac{\partial^2}{\partial s^2}\left(\left(-V_{\omega}\right)^{-1/4}\right)\right)\right| \lesssim \omega^{-4}r^{-3},\qquad \text{ if }\omega - m \gtrsim 1,
			\end{equation}
			
			Now we proceed as in Section 2 of Chapter 6 of~\cite{olver}. We define a new variable
			\[\xi\left(s,\omega\right) \doteq \int_A^s\sqrt{-V_{\omega}}\, dS. \]
			It is also convenient to use $\tilde{\omega}$ to refer to $\omega$ in the $\left(\xi,\omega\right)$ coordinate system. We then have the following equation for $h$:
			\begin{equation}\label{quiteaniceeqnforhwhenyouthinkaboutit}
				\frac{\partial^2h}{\partial \xi^2} + 2i\frac{\partial h}{\partial \xi} + \mathcal{E} h = -\mathcal{E},
			\end{equation}
			where
			\[\mathcal{E} \doteq \left(-V_{\omega}\right)^{-3/4}\frac{\partial^2}{\partial s^2}\left(\left(-V_{\omega}\right)^{-1/4}\right).\]
			In order solve~\eqref{quiteaniceeqnforhwhenyouthinkaboutit} we set-up the corresponding the Volterra equation:
			\begin{equation}\label{tofindhvolt}
				h\left(\xi,\tilde{\omega}\right) = -\frac{1}{2i}\int_{\xi}^{+\infty}\left(1-e^{2i(\eta-\xi)}\right)\left(-\mathcal{E}\left(\eta,\tilde{\omega}\right) - \mathcal{E}\left(\eta,\tilde{\omega}\right)h\left(\eta,\tilde{\omega}\right)\right)\, d\eta.
			\end{equation}
			The estimate~\eqref{potentialestimatesforWKBomegabigm} is enough to conclude by standard Volterra estimates from  Theorem~\ref{thm:volterra} (specifically, we  take $P_j$ and $Q$ to all be constant and set $\Phi(\xi) = A r^{-1/2}$ for a suitably large constant $A$) that 
			\[\left|h\right| + \left|\rd_{\xi} h\right|+ \left|\rd^2_{\xi \xi} h\right| \lesssim r^{-1/2}.\]
			The estimate for $\partial_{\xi\xi}^2h$ is then obtained directly by the equation.

			Keeping in mind the change of variables (and the fact that  $\rd_{s}\xi = \sqrt{-V_{\omega}}$ and $\rd_{\omega} \xi = \omega \int_A^{s}  (-V_{\omega})^{-\frac{1}{2}}dS$):
			\[\partial_{\tilde{\omega}} = \partial_{\omega} - \left(\partial_{\omega}\xi\right)\left(\partial_s\xi\right)^{-1}\partial_s= \partial_{\omega} - \left(\partial_{\omega}\xi\right)\partial_{\xi},\]
			the estimates~\eqref{potentialestimatesWKBwithomega},~\eqref{potentialestimatesWKBwithomega2},\eqref{potentialestimatesWKBwithomega3}, and~\eqref{potentialestimatesWKBwithomega4} and standard Volterra estimates from Theorem~\ref{thm:volterraparam} (setting, for a sufficiently large constant $A$, $\tilde{\Phi} = A\left(\omega^2-m^2\right)^{-1/2-2\delta}r^{-\delta}$ if $\omega \lesssim 1$ and $\tilde{\Phi} = A\omega^{-4}r^{-2}$ if $\omega \gtrsim 1$) imply that 
			\[\left|\partial_{\tilde{\omega}}h\right| + \left|\partial^2_{\xi\tilde{\omega}}h\right| \lesssim \left(\omega^2-m^2\right)^{-1/2-2\delta}r^{-\delta},\qquad \text{ if }\omega \lesssim 1,\]
			\[\left|\partial_{\tilde{\omega}}h\right| + \left|\partial^2_{\xi\tilde{\omega}}h\right| \lesssim \omega^{-4}r^{-2},\qquad \text{ if }\omega - m \gtrsim 1.\]

			The proposition then follows by changing variables from $\left(\xi,\tilde{\omega}\right)$ back to $\left(s,\omega\right)$.
			
		\end{proof}

		\section{Analysis of the low-frequency regime near $|\omega|=m$}\label{maxime.section}

		In this section, we restrict to the case $|\omega| < m$. It is useful to introduce a function $G(s)$ so that the ODE~\eqref{eq:mainV} becomes
		\begin{equation}\label{odewithG}
			\partial_s^2u = \left(m^2-\omega^2\right)u - \frac{2Mm^2}{s}u + \frac{1}{s^2}G\left(s\right)u.
		\end{equation}
		The function $G(s)$ will satisfy the following two inequalities for $s>1$ large enough \begin{equation}
			|G|(s) \lesssim \log(s),\ |G'|(s) \lesssim s^{-1}.
		\end{equation}
		
		We rewrite the ODE~\eqref{eq:mainV} with the following change of variables. We let $k:= Mm^2\left(m^2-\omega^2\right)^{-1/2}$, and $x$ such that
		$$
		s = \frac{2Mm^2}{m^2 - \omega^2} x.
		$$
		We obtain:
		\begin{equation}\label{eq:KGnice}
			- \partial_x^2u + 4k^2\Big(1 - \frac 1 x \Big) u + \frac{G(d^2k^2x)}{x^2} = 0,
		\end{equation}
		where we let $\nofa^2 := 2M^{-1}m^{-2}$, so that $s = \nofa^2 k^2 x$. Throughout this section we will assume that $k \gg 1$. We will also renormalize $u_I$ defined in Section~\ref{setup.section} and define $\tuI= C(\omega) u_I$ such that $\tilde{u}_I$ matches with the ``natural'' normalization of the airy functions involved  in the turning point analysis.

		The main objective of this section is to obtain a formula (including errors which become smaller as $\omega \rightarrow \pm m$, equivalently, $k\to \infty$) for $W(\tuI,u_H)= \partial_s\tuI u_H - \partial_su_H \tuI$ (recall the definitions \eqref{uI.boundary}, \eqref{uH.boundary}). The key difficulty is the fact that the original ODE~\eqref{eq:KGnice} has a turning point which goes to $r = \infty$ as $\omega$ tends to either $+m$ or $-m$. 
		We summarize the outcome of this section in the following proposition:
		\begin{prop}\label{ODE.interior.prop} 
			$\frac{W(\tuI,u_H)}{\overline{W}(\tuI,u_H)}(\omega)$ takes the following form in the $s$-coordinate system
			
			\begin{equation} \label{W.uI.uH.ratio}
				\frac{\overline{W}(\tuI,u_H)}{W(\tuI,u_H)}(\omega)= \gamma_m e^{i \phi_{-}(m)}+e^{i\phi_-(m)}[1-|\gamma_m|^2]\sum_{q=1}^{+\infty}  (-\gamma_m)^{q-1} e^{2i\pi k q}+\sum_{N=0}^{+\infty} F_N(\omega) e^{2i \pi k N},
			\end{equation} 
			where $\gamma_m := \gamma(m)$ and for some $\theta \in(0,1)$, the following estimates are true for all $(1-\theta)m<\omega <m$ \begin{align}
				&|\gamma(\omega)-\gamma_m|
				+ |\phi_{\pm}(\omega)-\phi_{\pm}(m)|  \ls k^{-\frac{1}{2}}\log(k)\\ &|\rd_{\omega}\gamma|(\omega)
				\ls  k^2\log^2(k) , \\ &  |F_N|(\omega)
				\ls \eta_0^{N}\cdot k^{-\frac{1}{2}}\log(k),\\ & |\rd_{\omega}F_N|(\omega) 
				\ls \eta_0^{N} \cdot k^2\log^2(k) . 
			\end{align} where $\eta_0 \in (0,1)$ is independent of $\omega$.  The functions $\gamma(\omega)$ and $\phi_{\pm}\left(\omega\right)$ are introduced in~\eqref{eq:gammas} and~\eqref{thisiswherephipmisdefined}.
			
		\end{prop}
		\subsection{Turning point analysis}\label{TP.analysis}
		In this first lemma, we use the turning point analysis  from~\cite{olver} (Theorem 3.1, Chapter 11) to construct the aforementioned function $\tuI$. Following~\cite{olver} we will frequently use the weight functions $E$, $N$, and $M$. We review the relevant properties of these in Appendix~\ref{app:airy}.
		\begin{lemma}\label{TP.lemma} 
			The only (up to rescaling by a constant) solution of~\eqref{eq:KGnice} which is bounded as $x\rightarrow +\infty$ may be written as
			\begin{equation}
				\tuI(x)= \Big|1-\frac{1}{x}\Big|^{-\frac{1}{4}}  |\zeta|^{\frac{1}{4}}(x) \cdot\left( Ai( [2k]^{2/3} \zeta(x))  	+\epsilon^{TP}(\omega,x)\right)\end{equation} where $\lim_{x\to +\infty}\ep^{TP} = 0$, we recall the definition of the Airy functions in Appendix~\ref{app:airy}, and in addition we have
			\begin{align}
				&\label{zeta1}\frac{2}{3} \zeta(x)^{\frac{3}{2}}= \sqrt{x}\sqrt{x-1}-\ln(\sqrt{x}+\sqrt{x-1}) \text{ for } x \geq 1,\\ & \label{zeta2} \frac{2}{3} [-\zeta(x)]^{\frac{3}{2}}= \arccos(\sqrt{x})-\sqrt{x}\sqrt{1-x} \text{ for }  0 \leq x \leq 1.
			\end{align} 
			
			Moreover, for all $x \in (1/4, 3/4)$,
			\begin{align}
				\label{error.largeneg}
				&k^{\frac{1}{6}}\ |\epsilon^{TP}(\omega,x)|,\ \  k^{-\frac{5}{6}}\ |\partial_x \epsilon^{TP}|(\omega,x) \lesssim   k^{-1} \log(k).
			\end{align}

		\end{lemma}
		
		\begin{proof}
			Throughout the proof, we assume that $x \geq 1/4$. We apply the theory of turning points from~\cite{olver} (Theorem 3.1, Chapter 11): defining $u=2 k$ and
			\begin{equation}\label{f.def}
				f(x) = 1-\frac{1}{x},\ g(x)= \frac{G(\nofa^2 k^2 x)}{  x^2},\ \zeta(x) (\frac{d\zeta}{dx})^2 = f(x),\ \hat{f}(x) = \frac{f(x)}{\zeta(x)},\  \psi(\zeta) = - (\hat{f})^{-\frac{3}{4}} \frac{d^2}{dx^2} (\hat{f})^{-\frac{1}{4}}  + \frac{g(x)}{\hat{f}}.
			\end{equation}One can check that $ \zeta(x) (\frac{d\zeta}{dx})^2 = f(x)$ is solved by \eqref{zeta1}, \eqref{zeta2} (where we choose $\zeta$ so that $\zeta(x) f(x) \geq 0$ for all $x$), and note that $x \rightarrow \zeta(x)$ thus defined is a $C^1$ function. Note moreover that $(-\zeta(0))^{\frac{3}{2}}=\frac{3\pi}{4}$ and $\zeta(x)- \zeta(0) \sim 2(\frac{x}{-\zeta(0)})^{\frac{1}{2}}$: \begin{equation}\label{zeta.asymp}
				\frac{2}{3} \zeta^{\frac{3}{2}}(x)=  x - \frac{\ln(x)}{2} - \frac{1}{2}-\ln(2) + O(x^{-1}) \text{ as } x\rightarrow +\infty,\ \zeta(x) \sim (x-1) \text{ as } x\rightarrow1.
			\end{equation}
			We introduce the following error-control function \begin{equation}
				H(x) = -\int_0^{\zeta(x)}|v|^{-\frac{1}{2}} \psi(v) dv =\pm\int_x^{1} \frac{|f|^{\frac{1}{2}}(x) }{\zeta(x)}\psi(\zeta(x)) dx.
			\end{equation} 
			Using \eqref{zeta.asymp}, we have that, using $TV_{a,b}$ to denote the total variation of a function on $(a,b)$:
			
			\begin{align}\label{TVH1}
				& TV_{x,+\infty}(H) \lesssim  x^{-1}\log(k^2 x)  \text{ for all } x \geq 1/4.
			\end{align}
			Using Theorem 11.3.1 in~\cite{olver}, we have  a solution $\tuI\left(\omega,x\right)$ defined for $x > 1/4$ so that \begin{equation}
				\tuI(\omega,x)= \left|1-\frac{1}{x}\right|^{-\frac{1}{4}} \zeta^{\frac{1}{4}}(x) \left( Ai( [2k]^{2/3} \zeta(x))+ \epsilon^{TP}(\omega,x)\right)
			\end{equation} and \begin{equation}\label{ep2}
				E( [2k]^{2/3} \zeta(x))\cdot \left(\frac{\epsilon^{TP}(\omega,x)}{M( [2k]^{2/3} \zeta(x))},\ \frac{\partial_x \epsilon^{TP}(\omega,x)}{ (2k)^{\frac{2}{3}} \hat{f}^{\frac{1}{2}}(x)N([2k]^{2/3} \zeta(x))}\right)  \lesssim  k^{-1} TV_{x,+\infty}(H)
			\end{equation} Moreover, the solution $\tuI$ is the unique solution to~\eqref{eq:KGnice} (up to a multiplicative constant) which remains uniformly bounded as $x\to +\infty$. After combining with \eqref{TVH1} we obtain immediately \eqref{error.largeneg}.
			
		\end{proof}

		In this next lemma, we extend the analysis of Lemma~\ref{TP.lemma} by also considering the $\omega$-derivatives of $\tuI$.
		\begin{lemma}\label{TP.lemma2}  Let $\tuI$ be as in Lemma~\ref{TP.lemma}. For $x \in (1/4,3/4)$, we have
			\begin{align}\label{derror.largeneg}
				k^{\frac{1}{6}}\ |\partial_{\omega}\epsilon^{TP}|(\omega,x)+  k^{-\frac{5}{6}}\ |\partial_{\omega x}^2\epsilon^{TP}|(\omega,x) \lesssim  k^2\log^2(k).
			\end{align}
		\end{lemma}

		\begin{proof} In the whole proof, we will assume that $x\geq \frac{1}{4}$. Now note (see \cite{olver}, proof of Theorem 11.3.1 and the proof of Lemma~\ref{TP.lemma} above) that $\epsilon^{TP}$ satisfies the following Volterra equation
			\begin{equation}\label{etpthisisthevolt}
				\epsilon^{TP}\left(\zeta,\omega\right) = \int_{\zeta}^{+\infty}\left[K\left(\zeta,v,\omega\right)\psi(v,\omega)\left(Ai((2k)^{2/3}v) + \epsilon^{TP}\left(v,\omega\right)\right)\right]\, dv,
			\end{equation}
			with 
			\begin{equation}
				K(\zeta,v,\omega)= (2k)^{-2/3}\left(Bi( [2k]^{2/3} \zeta)Ai( [2k]^{2/3} v)-Ai( [2k]^{2/3} \zeta)Bi( [2k]^{2/3} v)\right),
			\end{equation}
			and $\psi$ is as in the proof of Lemma~\ref{TP.lemma}. Our plan will be to apply Theorem~\ref{thm:volterraparam}.
			
			Now, we clearly have $$ \left|\frac{dk}{d\omega}\right| \lesssim  k^3,\ \left|\frac{d(k^2)}{d\omega}\right| \lesssim  k^4,\ \left|\partial_{\omega}\psi\right|= \left|\frac{ \partial_{\omega}g}{\hat{f}}\right|\lesssim \frac{k^2}{x^2 \cdot \hat{f}(x)}  \tilde{G}(d^2k^2x),$$ where we defined $\tilde{G}(s)= s G'(s)=O(1)$.

			We next note the following various upper bounds for the Kernel $K$ (cf.~(11.3.14) in \cite{olver}) that hold for all $v \geq \zeta$: \begin{align}
				&|K|(\zeta,v,\omega) \leq    \underbrace{(2k)^{-2/3}E^{-1}([2k]^{\frac{2}{3}} \zeta)   M([2k]^{\frac{2}{3}} \zeta)}_{\doteq \tilde{P}_0}  \underbrace{E([2k]^{\frac{2}{3}} v) M([2k]^{\frac{2}{3}} v)}_{\doteq Q(v)}, \\
				&\left|\frac{\partial K}{\partial \zeta}\right|(\zeta,v,\omega) \leq \underbrace{E^{-1}([2k]^{\frac{2}{3}} \zeta)   N([2k]^{\frac{2}{3}} \zeta)}_{\doteq \tilde{P}_1}Q(v).
			\end{align}
			For the function $\psi(v,\omega)$ we have the following bounds
			\begin{align}\label{psibounds}
				&\left|\psi\left(v,\omega\right)\right| \lesssim \log(k)\left(1+\left|x\left(v\right)\right|\right)^{-4/3}\
				\\ \nonumber &\left|\partial_{\omega}\psi\left(v,\omega\right)\right| \lesssim k^2\left(1+|x(v)|\right)^{-\frac{4}{3}}, 
			\end{align}
			where $x(v)$ is determined by inverting the formulas~\eqref{zeta1} and~\eqref{zeta2}. We recall, in particular, that $v \sim x^{2/3}$ when $v \gg 1$ (note that we do not need to study the behavior as $x\rightarrow 0$ as we assumed $x\geq \frac{1}{4}$).
			
			Now, in the notation of Theorem~\ref{thm:volterraparam}, we have
			\[R\left(\zeta,\omega\right) = (2k)^{-2/3}\int_{\zeta}^{+\infty}K\left(\zeta,v,\omega\right)\psi\left(v,\omega\right)Ai\left((2k)^{2/3}v\right)\, dv.\]
			When estimating $\partial_{\omega}R$ it is useful to note that the following estimates holds (see Appendix~\ref{app:airy}):
			\begin{equation}\label{whenaimeetsbiitsok}
				\left|\frac{d}{dx}\left(Ai(x)Bi(x)\right)\right| \lesssim M^6(x),\qquad \left||v|^{1/4}M\left((2k)^{2/3}v\right)\right| \lesssim k^{-1/6}.
			\end{equation}
			Then, keeping~\eqref{whenaimeetsbiitsok} in mind, one may check that:
			\begin{equation}\label{tpRomegaest1}
				\left|\partial_{\omega}R\right| \lesssim  k^{8/3}\log(k)\tilde{P}_0\cdot .
			\end{equation}
			Similarly,
			\begin{equation}\label{tpRomegaest2}
				\left|\partial^2_{\omega\zeta}R\right| \lesssim k^{8/3}\log(k)\tilde{P}_1\cdot .
			\end{equation}
			
			Now, again in the notation of Theorem~\ref{thm:volterraparam}, we need to estimate $\tilde{R}_1$ and $\tilde{R}_2$. In view of Lemma~\ref{TP.lemma} we  take $\psi_0 = \psi$, $\psi_1 =0$, we have that 
			\[ [\log(k)]^{-1} \left|\Psi_0\right|,   k^{\frac{2}{3}} |\kappa_0| \lesssim 1 ,\qquad \Psi_1 = 0,\]
			we may take $\Phi(\zeta) = TV_{x(\zeta),+\infty}(H)$ for $H$ as in Lemma~\ref{TP.lemma}, and we may take $P_0 =    k^{-\frac{2}{3}}   E^{-1}\left((2k)^{2/3}\zeta\right)M\left((2k)^{2/3}\zeta\right)$. Keeping in mind that $\left|Ai\left((2k)^{2/3}\zeta\right)\right| \sim k^{2/3}P_0$ it is clear that by using essentially the same estimates as leads to~\eqref{tpRomegaest1} and~\eqref{tpRomegaest2}, that one obtains:
			\[\left|\tilde{R}_1\right| \lesssim k^{8/3}\log^2(k)\tilde{P}_0\]
			\[\left|\tilde{R}_2\right| \lesssim k^{8/3}\log^2(k)\tilde{P}_1.\]
			(Note that these estimates are strictly easier since $\Phi$, at the cost of an additional $\log(k)$, now provides extra decay in all the integrals.) The lemma then follows from Theorem~\ref{thm:volterraparam}.

		\end{proof}

		\subsection{Control of the oscillation error} In this section, we will take $x \in (1/4,3/4)$.	We denote, for $y > 0$,
		$$ \epsilon^{osc}(y) = Ai(-y)- \pi^{-\frac{1}{2}} y^{-\frac{1}{4}} \cos( \frac{2}{3} y^{\frac{3}{2}}-\frac{\pi}{4})$$ Standard facts on Airy functions (see \cite{olver} (1.09) p392 in Chapter 11 and Appendix~\ref{app:airy}) give that  \begin{equation}\epsilon^{osc}(y)= O(y^{-\frac{7}{4}}),\ \frac{d\epsilon^{osc}}{dy}(y) = O(y^{-\frac{5}{4}}),\ \frac{d^2\epsilon^{osc}}{dy^2}(y) = O(y^{-\frac{3}{4}})  \text{ as } y \rightarrow+\infty	\end{equation}	Now note that for $x \in (1/4,3/4)$ we have \begin{equation*}	Ai( [2k]^{\frac{2}{3}} \zeta(x))= \frac{\pi^{-\frac{1}{2}}}{2} (2k)^{-\frac{1}{6}} [-\zeta(x)]^{-\frac{1}{4}} \left(  e^{-i\frac{\pi}{4}} \exp(ik \frac{4[-\zeta(x)]^{\frac{3}{2}}}{3} )+ e^{i\frac{\pi}{4}}\exp(-ik \frac{4[-\zeta(x)]^{\frac{3}{2}}}{3} )\right)+ \epsilon^{osc}( -[2k]^{\frac{2}{3}} \zeta(x)).	\end{equation*}
		We may thus conclude that \begin{equation}\begin{split}\label{uIformula}
				\tuI(x)=& \frac{\pi^{-\frac{1}{2}}}{2} (2k)^{-\frac{1}{6}} |1-x^{-1}|^{-\frac{1}{4}}\left(  e^{-i\frac{\pi}{4}} \exp(ik \frac{4[-\zeta(x)]^{\frac{3}{2}}}{3} )+ e^{i\frac{\pi}{4}}\exp(-ik \frac{4[-\zeta(x)]^{\frac{3}{2}}}{3} )\right)\\ + &  |1-x^{-1}|^{-\frac{1}{4}} |\zeta|^{\frac{1}{4}}(x)\underbrace{\left[ \ep^{TP}(\omega,x) +\epsilon^{osc}( -[2k]^{\frac{2}{3}} \zeta(x))\right]}_{:=\eta^{TP}(\omega,x)}\end{split}
		\end{equation}

		\begin{lemma}\label{lemma.eta}
			Gathering the earlier estimates we get, for $ x \in (1/4,3/4) $: \begin{align}
				& k^{\frac{1}{6}}\left[|\eta^{TP}|(\omega,x)+k^{-1} |\rd_{x}\eta^{TP}|(\omega,x)\right]\ls k^{-1}  \log(k),\\
				& k^{\frac{1}{6}}\left[|\partial_{\omega }\eta^{TP}|(\omega,x)+ k^{-1} |\rd_{x \omega}^2\eta^{TP}|(\omega,x)\right] \ls   k^2 \log^2(k).
			\end{align}
		\end{lemma}

		\subsection{WKB below the turning point}

		In this section we use the WKB method to construct a convenient set of solutions to~\eqref{eq:KGnice} for which we have good control in the region $x \in [k^{-2+\delta},1/2]$, for any small constant $\delta > 0$ (independent of $k$). The key point is that in this region we are protected against the turning point of the potential.
		\begin{lemma}\label{WKB.lemma}In the region $x \in [ k^{-2+\delta},1/2]$,
			the general solution of~\eqref{eq:KGnice} is a linear combination of the following two functions
			
			\begin{equation}\begin{split} \label{w+WKB}
					w_{+}(x)=   (x^{-1}-1)^{-\frac{1}{4}} \exp[i \frac{4k}{3}[-\zeta]^{\frac{3}{2}}(x)] \cdot(1+\ep^{WKB}(\omega,x))\end{split}
			\end{equation}\begin{equation}\begin{split} \label{w-WKB}
					w_{-}(x)=  \overline{w_+}(x) =  (x^{-1}-1)^{-\frac{1}{4}} \exp[-i \frac{4k}{3}[-\zeta]^{\frac{3}{2}}(x)] \cdot(1+\bar{\ep}^{WKB}(\omega,x))\end{split}
			\end{equation} 
			where $\zeta$ is defined by~\eqref{zeta2}, and $\ep^{WKB}$ and $\bar{\ep}^{WKB}$ satisfy the following estimates for $x \in [k^{-2+\delta},1/2]$ 
			\begin{align}\label{WKB.error1}
				&|\ep^{WKB}|(\omega,x),\ k^{-1} x^{\frac{1}{2}}|\partial_{x} \ep^{WKB}|(\omega,x)\lesssim k^{-1}  \log(k)x^{-\frac 12},  \\
				\label{WKB.error2}
				&|\partial_{\omega}\ep^{WKB}|(\omega,x),\ k^{-1} x^{\frac{1}{2}}|\partial_{x \omega}^2 \ep^{WKB}|(\omega,x)  \lesssim  k^2\log(k) \log(x^{-1})  ,\\
				&\label{WKB.error3}|\partial_{x x}^2\ep^{WKB}|(\omega,x)\lesssim  \log(k) x^{-2} + k\log(k)x^{-3/2}.
			\end{align}
			Moreover, we may assume that
			\begin{equation}\label{wronskian.WKB}
				W_x(w_+,w_-) = -4i k, 
			\end{equation}
			where $W_x$ denotes the Wronskian computed in the $x$-coordinates:
			\[W_x\left(f,g\right) \doteq \partial_xf\ g - f\ \partial_xg.\]
			(We will suppress the $x$ subscript, when the meaning should be clear from context.)
			
			Finally, for later convenience we list the estimates~\eqref{WKB.error1},~\eqref{WKB.error2}, and~\eqref{WKB.error3} in the special case when $x\sim k^{-1}$:
			\begin{align}\label{WKB.error1sp}
				&|\ep^{WKB}|\lesssim k^{-1/2}\log(k),\qquad  |\partial_{x} \ep^{WKB}|  \lesssim k\log(k), \\
				\label{WKB.error2sp}
				&|\partial_{\omega}\ep^{WKB}|\lesssim k^2\log^2(k),\qquad |\partial_{x \omega}^2 \ep^{WKB}|(\omega,x)  \lesssim k^{7/2}\log^2(k)   ,\\
				&\label{WKB.error3sp}|\partial_{x x}^2\ep^{WKB}|(\omega,x)\lesssim  k^{5/2}\log(k).
			\end{align}
		\end{lemma}
		\begin{proof} We start by observing that 
			\[\left(\frac{\partial}{\partial x}\left(\frac{4k}{3}\left(-\zeta(x)\right)^{3/2}\right)\right)^2 = 4k^2\left(\frac{1}{x}-1\right).\]
			In particular, from Section 2.4 of Chapter 6 in~\cite{olver} (and a simple rescaling) we have $w_{\pm}$ will exist if $\epsilon^{WKB}$ solves the following Volterra equation
			\begin{equation}\label{tosolvethewkbwkb}
				\epsilon^{WKB}\left(x,\omega\right) = -\frac{1}{2ik}\int_{\xi(x)}^{\xi(1/2)}\left(1-e^{2ik\left(v-\xi(x)\right)}\right)\psi(v)\left(1+\epsilon^{WKB}\right)\, dv,
			\end{equation}
			where $dv$ refers to an integration with respect to the $\xi = \frac{4}{3}(-\zeta(x))^{3/2}$ variable, and where $\psi$ is defined by
			\[\psi(v) \doteq \hat{f}^{-3/4}\frac{\partial^2}{\partial x^2}\left(\hat{f}^{-1/4}\right) - \frac{g}{\hat{f}},\]
			for 
			\[\hat{f}(v) = 4\left(\frac{1}{x(v)} - 1\right),\qquad g(v) = \frac{G(\nofa^2 k^2 x(v))}{  x(v)^2}.\]
			The following estimates are straightforward to establish:
			\[\left|\psi(v)\right| \lesssim \frac{\log(k^2x)}{x},\qquad \left|\partial_{\omega}\psi(v)\right| \lesssim  k^2 x^{-1}.\]
			Then we apply Theorem~\ref{thm:volterra} where we may take $K = -(2ik)^{-1}(1-e^{2ik(v-\xi)})$, $P_0 = k^{-1}$, $P_1$ and $Q$ to be constants, and take $\Phi = \log(k)x^{-1/2}$. The estimate~\eqref{WKB.error1} then follows after changing variables from $\partial_{\xi}$ to $\partial_x$. Next we will apply Theorem~\ref{thm:volterraparam} with $\tilde{P}_0 = P_0$, $\tilde{P}_1 = P_1$, and $\tilde{Q} = Q$, and $\tilde{\Phi} = k^3\log(k) \log(x^{-1})+ k^2 \log(k)x^{-1/2}$. Then the estimate~\eqref{WKB.error2} follow after switching into $x$-coordinates from $\xi$-coordinates. Finally,~\eqref{WKB.error3} follows by differentiating the Volterra equation for $\epsilon^{WKB}$ and again converting from $\xi$ to $x$-coordinates.
			
			It remains to verify~\eqref{wronskian.WKB}. To do this, we note that $\epsilon^{WKB}$ and $\partial_x\epsilon^{WKB}$ will both vanish at $x = 1/2$, and thus, computing at $x = 1/2$: 
			\begin{align}
				W\left(w_+,w_-\right) &= W\left((x^{-1}-1)^{-\frac{1}{4}} \exp[i \frac{4k}{3}[-\zeta]^{\frac{3}{2}}(x)],(x^{-1}-1)^{-\frac{1}{4}} \exp[-i \frac{4k}{3}[-\zeta]^{\frac{3}{2}}(x)]\right) 
				\\ &= 2i\Im\left(\partial_x\left((x^{-1}-1)^{-\frac{1}{4}} \exp[i \frac{4k}{3}[-\zeta]^{\frac{3}{2}}(x)]\right)(x^{-1}-1)^{-\frac{1}{4}} \exp[-i \frac{4k}{3}[-\zeta]^{\frac{3}{2}}(x)]\right)
				\\ \nonumber &= -4ik
			\end{align}
		\end{proof}
		
		We now express the function $\tilde{u}_I$ in terms of the functions $w_{\pm}$. This will then provide us with a good approximate form of $\tilde{u}_I$ which is valid in the entire region $x \in [k^{-2+\delta},1/2]$.

		\begin{prop} \label{A.prop}We can write for all $k^{-2+\delta}\leq x \leq \frac{1}{2}$: \begin{equation}\label{writetildeuintermsofwpm}\tuI(\omega,x)= A(\omega) w_+(\omega,x)+ \bar{A}(\omega) w_-(\omega,x),\end{equation}
			with 
			
			\begin{align}\label{A.est}
				A(\omega) &= \frac{\pi^{-\frac{1}{2}}}{2} (2k)^{-\frac{1}{6}} \left[e^{-i\frac{\pi}{4}}+  \ep_{a}(\omega) \right],
			\end{align}
			with the following estimates: \begin{align}\label{error.AB}
				&|\ep_{a}|(\omega) \lesssim k^{-1}  \log(k),\\ \label{derror.AB}  & |\rd_{\omega}\ep_{a}|(\omega) \lesssim k^2\log^2(k).
			\end{align}
			
		\end{prop}
		
		\begin{proof} 
			
			Since $w_{\pm}$ are two independent solutions to the equation~\ref{eq:KGnice} and $\tilde{u}_I$ is real valued, it follows immediately that there exists a function $A(\omega)$, depending only on $\omega$, so that~\eqref{writetildeuintermsofwpm} holds. Our goal is then to estimate $A(\omega)$ and $\bar{A}(\omega)$. We will always  assume in what follows that $\frac{1}{4}\leq x \leq \frac{1}{2}$.

			We have  by \eqref{uIformula}, \eqref{w+WKB}, \eqref{w-WKB}\begin{equation*}
				\tuI(x)= \frac{\pi^{-\frac{1}{2}}}{2} (2k)^{-\frac{1}{6}}\left[e^{-i\frac{\pi}{4}}\frac{w_+(\omega,x)}{1+\ep^{WKB}(\omega,x)}+e^{i\frac{\pi}{4}}\frac{w_-(\omega,x)}{1+\bar{\ep}^{WKB}(\omega,x)}\right] + |1-x^{-1}|^{-\frac{1}{4}} |\zeta|^{\frac{1}{4}}(x) \eta^{TP}(\omega,x),
			\end{equation*}
			hence 
			\begin{equation*} \begin{aligned}
					W_x(\tuI, w_{+})&= \frac{\pi^{-\frac{1}{2}}}{2} (2k)^{-\frac{1}{6}} \Big[e^{i\frac{\pi}{4}}\frac{W(w_-,w_+)}{1+\bar{\ep}^{WKB}(\omega,x)}-e^{-i\frac{\pi}{4}}\partial_x \ep^{WKB}(\omega,x)\cdot  (x^{-1}-1)^{-\frac{1}{2}} e^{i \frac{8k}{3}[-\zeta]^{\frac{3}{2}}(x)}  \\& \qquad \qquad -e^{i\frac{\pi}{4}}\partial_x \bar{\ep}^{WKB}(\omega,x) \frac{1+{\ep}^{WKB}(\omega,x) }{1+\bar{\ep}^{WKB}(\omega,x)}\cdot  (x^{-1}-1)^{-\frac{1}{2}}\Big] \\
					& +  (x^{-1}-1)^{-\frac{1}{2}} W(|\zeta|^{\frac{1}{4}} \eta^{TP}(\omega,x), e^{i \frac{4k}{3}[-\zeta(x)]^{\frac{3}{2}}(x)} [1+\ep^{WKB}(\omega,x)]).
				\end{aligned}
			\end{equation*}
			We also find immediately that \begin{equation}\label{A.formula}
				\bar{A}(\omega) =  \frac{W(\tuI, w_{+})}{W(w_-, w_{+})}.
			\end{equation}
			
			Let us write	\begin{align}
				&\bar{\ep}_{a_0}(\omega,x):= \frac{e^{i\frac{\pi}{4}}\frac{W(w_-,w_+)}{1+\bar{\ep}^{WKB}(\omega,x)}  -e^{i\frac{\pi}{4}}\partial_x \bar{\ep}^{WKB}(\omega,x) \frac{1+{\ep}^{WKB}(\omega,x) }{1+\bar{\ep}^{WKB}(\omega,x)}\cdot  (x^{-1}-1)^{-\frac{1}{2}}}{ W(w_-, w_{+})}	-  e^{i\frac{\pi}{4}}, \\ & \bar{\ep}_{a_1}(\omega,x):=   \frac{2 (2k)^{\frac{1}{6}}} {\pi^{-\frac{1}{2}} }\ \frac{  (x^{-1}-1)^{-\frac{1}{2}} W(|\zeta|^{\frac{1}{4}} \eta^{TP}(\omega,x), e^{i \frac{4k}{3}[-\zeta]^{\frac{3}{2}}(x)} [1+\ep^{WKB}(\omega,x)])}{ W(w_-, w_{+})} ,\\ & \bar{\ep}_{a_2}(\omega,x):=  e^{i \frac{8k}{3}[-\zeta]^{\frac{3}{2}}(x)}  \frac{ -e^{-i\frac{\pi}{4}}\partial_x \ep^{WKB}(\omega,x)\cdot  (x^{-1}-1)^{-\frac{1}{2}}}{ W(w_-, w_{+})}
			\end{align} so that \begin{equation*}
				\bar{A}(\omega)= \frac{\pi^{-\frac{1}{2}}}{2} (2k)^{-\frac{1}{6}} \left[e^{i\frac{\pi}{4}}+  \bar{\ep}_{a_0}(\omega,x) + \bar{\ep}_{a_2}(\omega,x)+ \bar{\ep}_{a_1}(\omega,x) \right]
			\end{equation*}
			For now, notice from \eqref{wronskian.WKB}, \eqref{WKB.error1}, \eqref{WKB.error2} 
			that \begin{align}\label{EB02.proof}
				&|\ep_{a_0}|(\omega,x),\ |\ep_{a_2}|(\omega,x) \lesssim k^{-1} \log(k), \\ &  	|\partial_{\omega}\ep_{a_0}|(\omega,x),\ |\partial_{\omega}\ep_{a_2}|(\omega,x) \lesssim k^2\log^2(k).\label{dEB02.proof}
			\end{align}
			for $\ep_{a_1}$,   note that in view of Lemma~\ref{lemma.eta} and \eqref{WKB.error1}, \eqref{WKB.error2} we have (noting that the most singular contributions come from  $\rd_{\omega} \eta^{TP}$  and $\rd_{x \omega}^2 \eta^{TP}$): \begin{align}\label{EB1.proof}
				&|\ep_{a_1}|(\omega,x) \lesssim k^{-\frac{5}{6}}  \left[ |\rd_x \eta^{TP}(\omega,x)|+ |\eta^{TP}(\omega,x)| [k  + |\rd_x \ep^{WKB}|(\omega,x)] \right]\lesssim k^{-1} \log(k),\\ &  |\rd_{\omega}\ep_{a_1}|(\omega,x) \lesssim[  k^{-\frac{5}{6}}  |\frac{dk}{d\omega}|  ] \left[ |\rd_x \eta^{TP}(\omega,x)|+ |\eta^{TP}(\omega,x)| [k  + |\rd_x \ep^{WKB}|(\omega,x)] \right]\nonumber\\+&  k^{-\frac{5}{6}}  |\rd_{\omega} \ep^{WKB}|(\omega,x) \left[ |\rd_x \eta^{TP}(\omega,x)|+ |\eta^{TP}(\omega,x)|k   \right]\nonumber\\+& k^{-\frac{5}{6}} \left[ |\rd_{x \omega}^2 \eta^{TP}(\omega,x)|+ |\rd_{\omega}\eta^{TP}(\omega,x)| [k + |\rd_x \ep^{WKB}|(\omega,x)] \right] \nonumber \\+ & k^{-\frac{5}{6}}   |\eta^{TP}(\omega,x)| [k^3  + |\rd^2_{x \omega} \ep^{WKB}|(\omega,x)]    
				\lesssim  k^2\log^2(k)\label{dEB1.proof}.
			\end{align}

			We evaluate  \eqref{EB02.proof}, \eqref{dEB02.proof}, \eqref{EB1.proof}, \eqref{dEB1.proof} at $x=\frac{1}{2}$: this gives \eqref{error.AB}, \eqref{derror.AB} for the $\bar{A}$-terms. The similar estimates for the $A$-terms follow identically.

		\end{proof}

		Now, we write the solution in an even more explicit way, with in mind the future matching with the modified Bessel functions. We will make use of the following Taylor expansion \begin{align}
			&\frac{4}{3} [-\zeta(x)]^{\frac{3}{2}}= \pi - 4x^{\frac{1}{2}}+ \ep^{AT}(x),\\ & \ep^{AT}(x)= O(x^{\frac{3}{2}}),\   \frac{d}{dx}\ep^{AT}(x)= O(x^{\frac{1}{2}}) \text{ as } x \rightarrow 0. \label{ep.AT}
		\end{align}
		We thus obtain immediately from Lemma~\ref{WKB.lemma}:
		\begin{equation}\begin{split}
				w_{+}(x)=   (x^{-1}-1)^{-\frac{1}{4}} e^{i \pi k} e^{-4ik  x^{\frac{1}{2}} } e^{ik \cdot \ep^{AT}(x)} \cdot(1+\ep^{WKB}(\omega,x)),\end{split}\label{w+}
		\end{equation}\begin{equation}\begin{split}
				w_{-}(x)=    (x^{-1}-1)^{-\frac{1}{4}}  e^{-i \pi k} e^{4ik  x^{\frac{1}{2}} } e^{-ik \cdot \ep^{AT}(x)} \cdot(1+\bar{\ep}^{WKB}(\omega,x))\end{split}\label{w-}.
		\end{equation}

		\subsection{Matching with Bessel functions}\label{Bessel.section}

		We note that the fundamental solutions for~\eqref{odewithG} when $\omega=\pm m $ (i.e.\ $k=+\infty$)  and $G=0$  are the following renormalized Bessel functions (see Appendix~\ref{bessel.appendix})
		\begin{align}
			\label{cB1}
			&\check{B}_1(s)= \sqrt{s}\ J_1(\sqrt{8Mm^2}\sqrt{s}) = \sqrt{2}\pi^{-\frac{1}{2}} s^{\frac{1}{4}}\left(8Mm^2\right)^{-1/4} \cos(\sqrt{8Mm^2}\sqrt{s}-\frac{3\pi}{4})+\delta^{B}_1(s)\\ &\check{B}_2(s)= \sqrt{s}\ Y_1(\sqrt{8Mm^2}\sqrt{s}) = \sqrt{2}\pi^{-\frac{1}{2}} s^{\frac{1}{4}} \left(8Mm^2\right)^{-1/4}\sin(\sqrt{8Mm^2}\sqrt{s}-\frac{3\pi}{4})+\delta^{B}_2(s)\label{cB2},\\ &|\delta^{B}_i|(s)+ s^{\frac{1}{2}} |\frac{d}{ds}\delta^{B}_i|(s) +  s |\frac{d^2}{ds^2}\delta^{B}_i|(s) \lesssim  s^{-\frac{1}{4}}\label{deltaB}.
		\end{align} Note that $W(\check{B}_1,\check{B}_2)=-\pi^{-1}$. 
		
		In the following lemma we will construct a pair of solutions to~\eqref{odewithG} by perturbing these rescaled Bessel functions. We will have good control of these new solutions in the range $\Delta k^{-2} \leq x \leq k^{-2/3}$ for a sufficiently large constant $\Delta > 0$ (independent of $k$). In particular, since $s \sim k^2x$, we will have good control both in the region $\{s \gg 1\}$ and also in a region where the $w_{\pm}$ functions are well controlled.
		
		\begin{lemma}\label{Bessel.lemma}
			There exists two independent solutions of the equation~\eqref{odewithG} given by in $(\omega,x)$ coordinates as \begin{align*}
				& \tilde{B}_+(\omega,x) =  \pi^{-1/2}\left(Mm^2\right)^{-1/2} k^{\frac{1}{2}} x^{\frac{1}{4}} e^{-i\frac{3\pi}{4}} e^{4ik\sqrt{x}}  + \eta_+(\omega,x)  ,\\ &  \tilde{B}_-(\omega,x) =  \overline{\tilde{B}_+} = \pi^{-1/2}\left(Mm^2\right)^{-1/2}k^{\frac{1}{2}} x^{\frac{1}{4}} e^{i\frac{3\pi}{4}} e^{-4ik\sqrt{x}}  + \eta_-(\omega,x)  .
			\end{align*}
			Assuming $k^{-2}\Delta \leq x \leq k^{-2/3}$, for a sufficiently large constant $\Delta$, we have
			\begin{align}
				&|\eta_{\pm}|(\omega,x)+ k^{-1} \cdot x^{\frac{1}{2}}|\rd_{x}\eta_{\pm}|(\omega,x)
				\lesssim k^{\frac{3}{2}} x^{\frac{7}{4}}+k^{-\frac{1}{2}}\log(k)  x^{-\frac{1}{4}},\label{eta1}\\  	&|\rd_{\omega}\eta_{\pm}|(\omega,x)\lesssim    x^{7/4}k^{7/2} + x^{1/4}k^{5/2}\log(k) + x^{9/4}k^{9/2},
				\label{eta2}\\ & |\rd_{x \omega}^2 \eta_{\pm}|(\omega,x)\lesssim    x^{5/4}k^{9/2} +k^{11/2}x^{7/4}.
			\end{align} 
			For later convenience, we rewrite these estimates at the point $x = k^{-1}$:
			\begin{align}
				&|\eta_{\pm}| \lesssim k^{-1/4} \log(k),\qquad |\rd_{x}\eta_{\pm}|
				\lesssim k^{5/4}\log(k)
				,\\  	&|\rd_{\omega}\eta_{\pm}|\lesssim      k^{9/4} \log(k),
				\\ & |\rd_{x \omega}^2 \eta_{\pm}|(\omega,x)\lesssim    k^{15/4}.
			\end{align}
			
			Finally we note  that  $W_s(\tilde{B}_+,\tilde{B}_-)=2i\pi^{-1}$ hence  \begin{equation}\label{Wronskian.B}
				W_x(\tilde{B}_+,\tilde{B}_-)= 2i \pi^{-1} d^2  k^2.
			\end{equation} 
			
		\end{lemma}
		\begin{proof}We start by constructing a pair of solutions in the region $s \gg 1$, to
			\[\frac{d^2}{ds^2}v(s)+[2Mm^2s^{-1}-G(s)s^{-2}]v(s)=0,\]
			which is~\eqref{odewithG} when $\omega = m$. Namely, we find $B_i(s)$ by solving the Volterra equation:
			\begin{equation*}\begin{split}
					&B_i(s)= \check{B}_i(s) + h_i(s),\\ 
					& h_i(s)= \int_s^{+\infty}\pi\left(\check{B}_2(s)\check{B}_1(y)-\check{B}_1(y)\check{B}_2(s)\right)g(y)\left(\check{B}_i(y)+h_i(y)\right)\, dy,
				\end{split}
			\end{equation*} where, for convenience, we have set $g(s) \doteq G(s)s^{-2}$. Now we apply Theorem~\ref{thm:volterra} in the region $s \gg 1$, with 
			\[K(s,y) = \pi\left(\check{B}_2(s)\check{B}_1(y)-\check{B}_1(s)\check{B}_2(y\right)y^{-1/2},\]
			$P_0(s)= As^{1/4}$,  $P_1(s) = As^{-1/4}$  for a suitable constant $A$, $Q(y) = y^{-1/4}$, and $\Phi(s) =   s^{-1/2}\log(s)$. We then obtain the following bounds:
			\begin{equation}\label{h.bounds}\begin{split}
					&|h_i|(s),\ s^{\frac{1}{2}}|\frac{d}{ds}h_i|(s) \lesssim s^{-\frac{1}{4}}\log(s),\\  &|\frac{d^2}{ds^2}h_i|(s)\lesssim s^{-\frac{5}{4}}\log(s). \end{split}	\end{equation}
			Furthermore, in view of the asymptotics as $s\to +\infty$, we will have that $W(B_1,B_2)=-\pi^{-1}$.

			We will next introduce the dependence on $\omega$ and find two solutions to  
			\[\frac{d^2}{ds^2}v(s)+[\omega^2-m^2 +2Mm^2s^{-1}-g(s)]v(s)=0,\]
			
			by setting $\tilde{B}_i(s)=B_i(s)+ \ep_i(s)$ and solving the equation 
			\[\partial_s^2\epsilon_i + \left(\frac{2Mm^2}{s}-g(s)\right)\epsilon_i = -k^{-2}\left(B_i + \epsilon_i\right),\]
			via the following Volterra equation, where $s_0$ is sufficiently large (independent of $k$): 
			\[\epsilon_i(s) = \pi k^{-2}\int_{s_0}^s\left(B_2(s)B_1(y) - B_2(y)B_1(s)\right)\left(B_i(y) + \epsilon_i(y)\right)\, dy.\]

			We now apply Theorem~\ref{thm:volterra} (with Remark~\ref{switchorder} in mind) in a region $s \in [s_0,s_f]$ with $s_f \doteq d^2k^{4/3}$ and, for a suitable constant $A > 0$ (independent of $s_f$ and $k$):
			\[K\left(s,y\right) = \pi  \left(B_2(s)B_1(y) - B_2(y)B_1(s)\right),\]
			\[P_0(s) = As^{1/4},\qquad P_1(s) = As^{-1/4},\qquad Q(y) = y^{1/4},\qquad \Phi(s) = k^{-2}s^{3/2}.\]
			We obtain
			\[\left|\epsilon_i\right| + s^{1/2}\left|\partial_s\epsilon_i\right| \lesssim s^{7/4}k^{-2}\exp\left(s_f^{3/2}k^{-2}\right) \lesssim s^{7/4}k^{-2}.\]
			Next, we turn to the derivatives with respect to $\omega$. In order to prevent confusion we will $\partial_{\mathring{\omega}}$ to denote $\omega$ derivatives computed in the $\left(s,\omega\right)$ coordinate system and reserve $\partial_{\omega}$ for the $\omega$ derivative computed in the $(x,\omega)$ coordinate system. We now apply Theorem~\ref{thm:volterraparam} with, for a suitable constant $A$ (independent of $k$): 
			\[\tilde{P}_0(s) = As^{1/4},\qquad \tilde{P}_1(s) = As^{-1/4},\qquad Q(y) = y^{1/4},\qquad \tilde{\Phi}(s) = s^{3/2}.\]
			We obtain
			\[\left|\partial_{\mathring{\omega}}\epsilon_i\right| + s^{1/2}\left|\partial^2_{\mathring{\omega}s}\epsilon_i\right| \lesssim s^{7/4}.\]
			Appealing directly to $\epsilon_i$'s equations, we also obtain
			\[\left|\partial_s^2\epsilon_i\right| \lesssim s^{3/4}k^{-2} + k^{-4}s^{7/4} \lesssim s^{3/4}k^{-2}.\]

			Note that once again 
			\begin{equation}\label{thisisactuallywhatsused}
				W(\tilde{B}_1,\tilde{B}_2)=-\pi^{-1}.
			\end{equation}
			
			Finally, we may define $\tilde{B}_{\pm}$ by
			\begin{align*}
				& \tilde{B}_+(s) = \tilde{B}_1(s)+i  \tilde{B}_2(s) = \sqrt{\frac{2}{\pi}}\left(8Mm^2\right)^{-1/4} s^{\frac{1}{4}} e^{-i\frac{3\pi}{4}} e^{i\sqrt{8Mm^2 s}}  + \underbrace{\delta_{+}^B(s)+ h_+(s)+  \ep_+(\xi,s)}_{\eta_+},\\ &  \tilde{B}_-(s) = \tilde{B}_1(s)-i  \tilde{B}_2(s) = \sqrt{\frac{2}{\pi}}\left(8Mm^2\right)^{-1/4} s^{\frac{1}{4}}e^{i\frac{3\pi}{4}} e^{-i\sqrt{8Mm^2s}}  + \underbrace{\delta_{-}^B(s)+  h_-(s)+ \ep_-(\xi,s)}_{\eta_-},
			\end{align*}  where $\delta_{\pm}^{B}= \delta_{1}^{B}\pm i \delta_{2}^{B} $, $h_{\pm}= h_1 \pm h_2$ and  $\ep_{\pm}= \ep_1 \pm i \ep_2$. The estimates for $\eta_{\pm}$ are then obtained by translating the previously obtained estimates in $\left(\mathring{\omega},s\right)$ coordinates into $\left(\omega,x\right)$ coordinates.
			
			Finally, we note that~\eqref{Wronskian.B} is straightforward to establish.
		\end{proof}

		In the next lemma, we express the function $\tilde{u}_I$ in terms of the perturbed Bessel functions $\tilde{B}_{\pm}$.
		
		\begin{lemma}\label{Bessel.error.lemma}Since the Wronskian of $\tilde{B}_+$ and $\tilde{B}_-$ is non-zero and $\tilde{u}_I$ is real valued, we may find functions $a_{\pm}(\omega)$ with $a_- = \overline{a}_+$ so that 
			\begin{equation}
				\tuI(\omega,x)= a_+(\omega)\tilde{B}_+(\omega,x)+a_-(\omega)\tilde{B}_-(\omega,x).
			\end{equation} 
			We then have the following estimates for $a_{\pm}$:
			\begin{align}
				&a_{-}(\omega) = \bar{a}_+(\omega)= A(\omega) k^{-\frac{1}{2}}\pi^{\frac{1}{2}} (Mm^2)^{\frac{1}{2}} e^{i \pi k  - i\frac{3\pi}{4}} \left[1 +\alpha_{-}(\omega)\right],\\  & |\alpha_{-}|(\omega)\ls  k^{-\frac{1}{2}} \log(k),\\ & |\frac{d \alpha_-(\omega)}{d\omega}|\ls  k^2 \log^2(k).
			\end{align} 
		\end{lemma}
		
		\begin{proof}
			In this proof we will  always work in a region where  $C^{-1}k^{-1} \leq x \leq C k^{-1}$ for some large constant $C$ which is independent of $k$.We need to compute
			\[a_- = \frac{W\left(\tilde{u}_I,\tilde{B}_+\right)}{W\left(\tilde{B}_-,\tilde{B}_+\right)}.\]
			We write \begin{equation}\begin{split}
					&	\tuI(x) = (1-x)^{-\frac{1}{4}} k^{-\frac{1}{2}} \pi^{\frac{1}{2}}(Mm^2)^{\frac{1}{2}} \Re\left[ A(\omega)e^{i\pi k - i\frac{3\pi}{4}} e^{ik \ep^{AT}(x)} (1+\ep^{WKB}(\omega,x)) \tilde{B}_{-}(\omega,x)\right]- \ep_F(\omega,x),\\& \ep_F(\omega,x):= \Re\left[  (1-x)^{-\frac{1}{4}} k^{-\frac{1}{2}}\pi^{\frac{1}{2}}(Mm^2)^{\frac{1}{2}}\cdot A(\omega)e^{i\pi k - i\frac{3\pi}{4}} e^{ik \ep^{AT}(x)} (1+\ep^{WKB}(\omega,x))\eta_{-}(\omega,x)\right]. \end{split}
			\end{equation}
			We take advantage of the fact that $\tilde{B}_+= \overline{		\tilde{B}_-}$ and the formulas $W(Hf,f)= H' f^2$, $W(Hf,\bar{f})= H' |f|^2 + H W(f,\bar{f})$ to obtain	\begin{align*}
				& W\left( [1-x]^{-\frac{1}{4}} e^{ik\ep^{AT}(x)}[1+\ep^{WKB}(\omega,x)]\tilde{B}_-,\tilde{B}_+\right)= [1-x]^{-\frac{1}{4}} e^{ik\ep^{AT}(x)}[1+\ep^{WKB}(\omega,x)] W(\tilde{B}_-,\tilde{B}_+)\\+\ & e^{ik\ep^{AT}} (1-x)^{-\frac{1}{4}}\left[ \frac{[1+\ep^{WKB}](1-x)^{-1}}{4} + \rd_{x}\ep^{WKB} +i[1+\ep^{WKB}]k \frac{d}{dx} \ep^{AT}\right] |\tilde{B}_{+}|^2,\\ &  W\left( [1-x]^{-\frac{1}{4}} e^{-ik\ep^{AT}(x)}[1+\bar{\ep}^{WKB}(\omega,x)]\tilde{B}_+,\tilde{B}_+\right)= \\ & e^{-ik\ep^{AT}} (1-x)^{-\frac{1}{4}}\left[ \frac{[1+\bar{\ep}^{WKB}](1-x)^{-1}}{4} + \rd_{x}\bar{\ep}^{WKB} -i[1+\bar{\ep}^{WKB}]k \frac{d}{dx} \ep^{AT}\right] \tilde{B}_{+}^2.
			\end{align*} Now note the following estimates, which follow from \eqref{Wronskian.B}, \eqref{ep.AT}, \eqref{WKB.error1} and the fact that $|\tilde{B}_+|\ls k^{1/4}$ (keeping in mind the restriction in the proof to $x\sim k^{-1}$): \begin{align*}
				&\bigl| [1-x]^{-\frac{1}{4}} e^{ik\ep^{AT}(x)}[1+\ep^{WKB}(\omega,x)]-1 \bigr|\lesssim  x+ k|\ep^{AT}|(x)+ |\ep^{WKB}|(\omega,x)\ls  k^{-1/2}\log(k), \\& \bigl|\rd_{\omega}\left( [1-x]^{-\frac{1}{4}} e^{ik\ep^{AT}(x)}[1+\ep^{WKB}(\omega,x)]\right) \bigr|\lesssim |\frac{dk}{d\omega}| |\ep^{AT}|(x)+|\rd_{\omega}\ep^{WKB}|(\omega,x) \ls k^2\log^2(k),\\ & \bigl| (1-x)^{-\frac{1}{4}} \tilde{B}_{+}^2\frac{ \frac{[1+\bar{\ep}^{WKB}](1-x)^{-1}}{4} + \rd_{x}\bar{\ep}^{WKB} -[1+\bar{\ep}^{WKB}]k \frac{d}{dx} \ep^{AT} }{W(\tilde{B}_+,\tilde{B}_-)}\bigr| \ls k^{-1/2}\log(k),\\ & \bigl|\rd_{\omega}\left( (1-x)^{-\frac{1}{4}}e^{-ik\ep^{AT}(x)}\tilde{B}_{+}^2\frac{ \frac{[1+\bar{\ep}^{WKB}](1-x)^{-1}}{4} + \rd_{x}\bar{\ep}^{WKB} -[1+\bar{\ep}^{WKB}]k \frac{d}{dx} \ep^{AT} }{W(\tilde{B}_+,\tilde{B}_-)}\right)\bigr|\\ & \ls |\frac{d}{d\omega}\log[W(\tilde{B}_+,\tilde{B}_-)]|k^{-1/2}\log(k)  + \frac{ |\rd_{\omega}\overline{\ep}^{WKB}|[1+k|\frac{d}{dx}\ep^{AT}|(x)]+|\rd_{x \omega}^2\ep^{WKB}|+|\frac{dk}{d\omega}||\frac{d}{dx}\ep^{AT}|(x) }{|W(\tilde{B}_+,\tilde{B}_-)|}|\tilde{B}_{+}|^2\\ &+ \left|\frac{d\log\left(\tilde{B}_+^2\right)}{d\omega}\right|k^{-1/2}\log(k)+ |\frac{dk}{d\omega}| |\ep^{AT}|(x)k^{-1/2}\log(k) \ls k^2\log^2(k),\\ & \bigl|\rd_{\omega}\left( (1-x)^{-\frac{1}{4}}e^{-ik \epsilon^{AT}(x)}|\tilde{B}_{+}|^2\frac{ \frac{[1+\ep^{WKB}](1-x)^{-1}}{4} + \rd_{x}\ep^{WKB} -[1+\ep^{WKB}]k \frac{d}{dx} \ep^{AT} }{W(\tilde{B}_+,\tilde{B}_-)}\right)\bigr|   \ls k^2\log^2(k),
			\end{align*} 
			where we have  used the fact which may be deduced from Lemma~\ref{Bessel.lemma} that \begin{align}
				\bigl|\frac{d\log\tilde{B}_+}{d\omega} \bigr|\ls k^{5/2},\qquad \left|\partial_{\omega}\tilde{B}_+\right| \lesssim k^{11/4}.\label{B+.omega}
			\end{align}
			It is also useful to keep in mind that
			\[\left|\partial_x\tilde{B}_+\right| \lesssim k^{7/4},\qquad \left|\partial^2_{x\omega}\tilde{B}_+\right| \lesssim k^{17/4}.\]
			Now we turn to estimating the Wronskians involving $\ep_{F}(\omega,x)$. We start with the following explicit computation: \begin{align*}
				& W\left( [1-x]^{-\frac{1}{4}} e^{ik\ep^{AT}(x)}[1+\ep^{WKB}(\omega,x)]\eta_+,\tilde{B}_{+}\right)= [1-x]^{-\frac{1}{4}} e^{ik\ep^{AT}(x)}[1+\ep^{WKB}(\omega,x)] W(\eta_+,\tilde{B}_+)\\+\ & [1-x]^{-\frac{1}{4}}e^{ik\ep^{AT}} \left[ \frac{[1+\ep^{WKB}](1-x)^{-1}}{4} + \rd_{x}\ep^{WKB} +[1+\ep^{WKB}]k \frac{d}{dx} \ep^{AT}\right] \eta_+\  \tilde{B}_+.
			\end{align*} 
			
			We then estimate this by the following: \begin{align*}
				& \left|  \frac{W\left( [1-x]^{-\frac{1}{4}} e^{ik\ep^{AT}(x)}[1+\ep^{WKB}(\omega,x)]\eta_+,\tilde{B}_{+}\right)}{W(\tilde{B}_{+},\tilde{B}_{-})}\right|\lesssim k^{-2} |\eta_+ | |\tilde{B}_{+}|\left( 1+ |\rd_x \ep^{WKB}|(\omega,x)+ k|\frac{d}{dx}\ep^{AT}|(x)\right)\\ & + k^{-2} \left( |\rd_x \eta_+| |\tilde{B}_+|+  | \eta_+| |\rd_x \tilde{B}_+|   \right) \lesssim k^{-1/2}\log(k),
				\\ &  \left|  \rd_{\omega}\left(\frac{W\left( [1-x]^{-\frac{1}{4}} e^{ik\ep^{AT}(x)}[1+\ep^{WKB}(\omega,x)]\eta_+,\tilde{B}_{+}\right)}{W(\tilde{B}_{+},\tilde{B}_{-})}\right)\right|\lesssim 
				\\ \nonumber &\qquad k^2 \left|  \frac{W\left( [1-x]^{-\frac{1}{4}} e^{ik\ep^{AT}(x)}[1+\ep^{WKB}(\omega,x)]\eta_+,\tilde{B}_{+}\right)}{W(\tilde{B}_{+},\tilde{B}_{-})}\right|\\&+ k^{-2} \left[|\rd_{\omega}\eta_+ | |\tilde{B}_{+}|+ |\eta_+ | | \rd_{\omega}\tilde{B}_{+}|\right]\left( 1+ |\rd_x \ep^{WKB}|(\omega,x)+ k|\frac{d}{dx}\ep^{AT}|(x)\right)
				\\ \nonumber &\qquad + k^{-2} |\eta_+ | |\tilde{B}_{+}|\left(  |\rd_{x \omega}^2 \ep^{WKB}|(\omega,x)+ k^3|\frac{d}{dx}\ep^{AT}|(x)\right)\\ & + k^{-2} \left( |\rd_x \eta_+| |\tilde{B}_+|+  | \eta_+| |\rd_x \tilde{B}_+|   \right)\left(  |\rd_{\omega} \ep^{WKB}|(\omega,x)+ k^3|\ep^{AT}|(x)\right)\\ &+  k^{-2} \left( |\rd^2_{x \omega} \eta_+| |\tilde{B}_+|+  | \eta_+| |\rd^2_{x \omega}  \tilde{B}_+| + |\rd_x \eta_+| |\rd_{\omega}\tilde{B}_+|+  |\rd_{\omega} \eta_+| |\rd_x \tilde{B}_+|  \right) \lesssim k^2\log^2(k).
			\end{align*}  which concludes the proof.
			
		\end{proof}
		\subsection{The energy identity and the scattering of the solution regular at the event horizon}
		Recall the solutions $\tilde{B}_1 = \Re\left(\tilde{B}_+\right)$ and $\tilde{B}_2 = \Im\left(\tilde{B}_+\right)$ which where constructed in the course of  the proof of Lemma~\ref{Bessel.lemma}. Since these are linearly independent, we may define uniquely $E_{1}(\omega)$ and $E_{2}(\omega)$  by the formula $$ u_H = E_1(\omega) \tilde{B}_1 +    E_2(\omega) \tilde{B}_2.$$
		In the following lemma we establish an important estimate for the $E_i$'s which eventually is used to show that a potential cancellation in the most singular part of the Green's formula does not occur. For the first time in this section, the proof will involve an argument which is \emph{not} restricted to the region $r \gg 1$.
		\begin{lemma}\label{uH.scat.lemma} For $i=1,2$, $\omega \rightarrow E_i(\omega)$ is analytic  for $\omega \in \mathbb{R}$  and satisfies the following energy identity:
			\begin{equation}\label{energy}
				|E_1(\omega)+ i E_2(\omega)|^2-|E_1(\omega)- i E_2(\omega)|^2  = 4\pi \omega.
			\end{equation}
		\end{lemma}
		
		\begin{proof}  Recalling that $\tilde{B}_{\pm} =\tilde{B}_{1} \pm i\tilde{B}_{2}$.  We define also $E_{\pm}(\omega)$  by the formula \begin{align*}
				&  E_1= E_+ + E_ -, \\ & E_2= i [E_+ - E_ -]. 
			\end{align*}  so that $$ u_H = E_+(\omega) \tilde{B}_+ +    E_-(\omega) \tilde{B}_-.$$ 
			The energy identity can be derived from the following equality: $$ \frac{d}{ds}\left[\Im ( \frac{d u_H}{ds}\bar{u}_H )  \right] =0$$ Recall that $u_H(\omega,s) \sim e^{-i \omega s}$ as $s\rightarrow -\infty$, and therefore $$\Im ( \frac{d u_H}{ds}\bar{u}_H ) =-\omega.$$ 
			
			Using the fact that $\overline{\tilde{B}_+}=\tilde{B}_-$, we get that \begin{align*}
				&-\omega = \left[|E_+|^2- |E_-|^2 \right] \Im[\frac{d}{ds}\tilde{B}_+\tilde{B}_-]= \left[|E_+|^2- |E_-|^2 \right] \frac{W(\tilde{B}_+,\tilde{B}_-)}{2i} =\left[|E_+|^2- |E_-|^2 \right] \pi^{-1}.
			\end{align*}  where for the one before last equality we have used the identity $W_s(\tilde{B}_+,\tilde{B}_-) =2i\pi^{-1}$ from Lemma~\ref{Bessel.lemma}. We conclude that \begin{equation*}
				-|E_+|^2+ |E_-|^2 = \pi \omega,
			\end{equation*} which in view of the fact that $E_{\pm}= \frac{E_1 \mp i E_2}{2}$ gives \eqref{energy}.
			
			The analyticity of $\omega \rightarrow E_i(\omega)$ follows from  Lemma~\ref{makeuh} which yields the analyticity in $\omega$ of $u_H$, and the immediate fact that $\tilde{B}_1$ and $\tilde{B}_2$ depend analytically on $\omega$.

		\end{proof}

		\subsection{Putting everything together}
		We now are ready to give the proof of Proposition~\ref{ODE.interior.prop}.
		\begin{proof}

			We write $u_H = E_1 \tilde{B}_1 + E_2 \tilde{B}_2$, and we know that $\tuI= (a_+ + a_-) \tilde{B}_1+ i (a_+ -  a_-) \tilde{B}_2$. We also have already seen that $a_+ = \bar{a}_-$. We thus write  $\tuI= 2\Re(a_-) \tilde{B}_1 +2 \Im(a_-) \tilde{B}_2$. In view of the above Proposition~\ref{A.prop} and  Lemma~\ref{Bessel.error.lemma}, we get 
			
			\begin{equation}\begin{split}
					\tuI= -2^{-1/6}k^{-2/3}(Mm^2 )^{-1/2} \left[\cos(\pi k)+ \Re(    \tilde{\ep}_{a} e^{i\pi k })\right] \tilde{B}_1\\- 2^{-1/6}k^{-2/3}(Mm^2)^{1/2}  \left[\sin(\pi k)+  \Im(    \tilde{\ep}_{a} e^{i\pi k })\right] \tilde{B}_2\end{split}
			\end{equation} where $\tilde{\ep}_{a}(\omega)$ is chosen so that $\left[1+  \tilde{\ep}_{a}(\omega) \right]= [1+\alpha_-(\omega)] \left[1+   e^{i\frac{\pi}{4}}\ep_{a} (\omega) \right]$. Now we get, using \eqref{thisisactuallywhatsused} that \begin{equation*}\begin{split}
					&W\left(\tuI,u_H\right) = W_s(\tuI,u_H)\\&= \pi^{-1} 2^{-1/6}k^{-2/3}(Mm^2)^{-1/2} \left[ 
					E_2\cos(\pi k)- E_1 \sin(\pi k )+ E_2\Re(    \tilde{\ep}_{a} e^{i\pi k })-E_1\Im(    \tilde{\ep}_{a} e^{i\pi k })\right]\\&= \pi^{-1} 2^{-7/6}k^{-2/3}(Mm^2)^{-1/2}  \left[ (E_2 +iE_1+ \nu_1(\omega)) e^{i \pi k }+(E_2 -iE_1+\nu_{-1}(\omega)) e^{-i \pi k }
					\right]. \end{split}
			\end{equation*} where \begin{align} & \nu_1(\omega):= \frac{(E_2+ i E_1) }{2} \tilde{\ep}_{a},\\&  \nu_{-1}(\omega):= \frac{(E_2- i E_1)}{2}  \overline{\tilde{\ep}}_{a}.
			\end{align} 
			
			Finally, introduce
			\begin{equation}\label{eq:gammas}
				\begin{aligned}
					&  \Gamma_+(\omega)= E_2(\omega) +iE_1(\omega)+ \nu_1(\omega), \\  &   \Gamma_-(\omega)= E_2(\omega) -iE_1(\omega)+ \nu_{-1}(\omega),\\ & .
				\end{aligned}
			\end{equation}
			
			so that we may write \begin{align*}
				&W(\tuI,u_H)(\omega)= \pi^{-1} 2^{-1/6}k^{-2/3}(Mm^2)^{-1/2} \left[ \Gamma_+(\omega) e^{i \pi k}+ \Gamma_-(\omega) e^{-i \pi k} \right].
			\end{align*}
			As a consequence of our previous estimates on $\alpha_-$ and $\epsilon_a$ (see Proposition~\ref{A.prop} and Lemma~\ref{Bessel.error.lemma}), we have that 
			\begin{equation}\label{sowecanignorethisstuff}
				\left|\tilde{\epsilon}_a\right| \lesssim k^{-1/2}\log(k).
			\end{equation}
			In particular, the following limits exist
			\[\Gamma_{\pm}^m \doteq \lim_{k\to \infty}\Gamma_{\pm}\left(\omega\right) = -E_1(m) \pm iE_2(m).\]
			It is convenient to introduce also a function $\gamma(\omega)$
			\[\gamma(\omega)= \frac{\Gamma_+(\omega)}{\Gamma_-(\omega)}.\]
			In view of~\eqref{sowecanignorethisstuff} and~\eqref{energy} we have the fundamental facts that
			\begin{equation}\label{fundamentalfactyay}
				\left|\gamma\right| < 1,\qquad \left|\Gamma_-\right| > 0.
			\end{equation}
			We also introduce \begin{align}
				& \label{thisiswherephipmisdefined} e^{i \phi_-(\omega)} = \frac{\bar{\Gamma}_-(\omega)}{\Gamma_-(\omega)},\\
				& \gamma_m=\frac{\Gamma_{+}^{m}}{\Gamma_{-}^{m}},\\ & \delta  \gamma(\omega)= \gamma(\omega)-\gamma_m,
				\\& F_{0}(\omega)=  \bar{\gamma}(\omega) e^{i \phi_{-}(\omega)}-\bar{\gamma_m}e^{i \phi_{-}(m)}, \\ & F_{N}(\omega)=   e^{i \phi_{-}(\omega)} (1-|\gamma|^2(\omega)) (-\gamma)^{N-1} (\omega)-e^{i \phi_{-}(m)} (1-|\gamma_m|^2) (-\gamma_m)^ {N-1}.\\
			\end{align}

			Our next goal is to show that there exists $\eta_0 \in (0,1)$ independent of $\omega$ and $\theta \in (0,1)$ independent of $\omega$ such that for all $(1-\theta)m<\omega<m$:
			\begin{align}
				&\label{F.est}    |F_N|(\omega) \ls \eta_0^N \cdot k^{-\frac{1}{2}} \log(k),\\ &  \label{dF.est}  |\frac{dF_N}{d\omega}|(\omega) \ls  \eta_0^N \cdot k^2\log^2k.
			\end{align}
			
			We will make use of the fact that $\omega \rightarrow E_i(\omega)$ is real-analytic  for $(1-\theta)m<\omega<m$ (by Lemma~\ref{uH.scat.lemma}). 	First note that from  Proposition~\ref{A.prop} and Lemma~\ref{Bessel.error.lemma} \begin{equation}
				|\delta \gamma|(\omega)+ |e^{i\phi_-(\omega)} -e^{i\phi_-(m)}|\lesssim k^{-\frac{1}{2}} \log(k),
			\end{equation} so clearly \begin{equation}
				|F_N|(\omega) \lesssim  |\gamma|^N (\omega )k^{-\frac{1}{2}} \log(k) \ls  \eta_0^N k^{-\frac{1}{2}} \log(k),
			\end{equation} for some $|\gamma|_m < \eta_0 <1$. Then, taking the $\frac{d}{d\omega}$ derivative we obtain
			\begin{equation}
				|\frac{d\delta \gamma(\omega)}{d\omega}|+ |\frac{de^{i\phi_-(\omega)}}{d\omega} |\lesssim k^{2} \log^2(k),
			\end{equation} so similarly \begin{equation}
				|\frac{dF_N(\omega)}{d\omega}| \lesssim \eta_0^N  k^{2} \log^2(k).
			\end{equation}

			Keeping~\eqref{fundamentalfactyay} in mind, one can write, after some algebraic manipulation:  \begin{equation*}
				\frac{\bar{W}(\tuI,u_H)}{W(\tuI,u_H)}(\omega)= \frac{\bar{\Gamma}_-(\omega)}{\Gamma_-(\omega)} \left( \bar{\gamma}(\omega) + (1-|\gamma|^2(\omega))\sum_{N=1} ^{+\infty}  (-\gamma)^{N-1}(\omega) e^{2i\pi k N }\right).
			\end{equation*}
			
			Then, finally, we obtain
			\begin{equation*}
				\frac{\bar{W}(\tuI,u_H)}{W(\tuI,u_H)}(\omega)=\bar{\gamma}_me^{i \phi_{-}(m)}+ e^{i\phi_-(m)}[1-|\gamma_m|^2]\sum_{q=1}^{+\infty}  (-\gamma_m)^{q-1} e^{2i\pi k q}+\sum_{N=0}^{+\infty} F_N(\omega) e^{2i \pi k N}.
			\end{equation*}

		\end{proof}
		\section{Inverse Fourier transform and stationary phase}\label{FT.section}
		\subsection{Set-up}\label{setup.inverse}
		Throughout the section we will use the notation $\check{u}_H$, $\check{u}_I$, $\check{u}$, and $\check{H}$ from Remark~\ref{horizonrepform}. We  first consider the inhomogeneous Klein--Gordon equation~\eqref{eq:main} with a  smooth spacetime compactly supported source $H(t^*,r)$, and where both $H$ and $\psi$ vanish for $\{t < 0\}$. In Proposition~\ref{repuform} and Remark~\ref{horizonrepform} we have established a representation formula~\eqref{agoodrepformonthehorizon} for $\tilde{u}\left(\omega,r\right)$, the Fourier transform $\psi\left(t^*,r\right)$. When $\omega < m$ it is useful to use Remark~\ref{switcharoo} to re-write~\eqref{agoodrepforu} so that it involves $\tilde{u}_I$ instead of $u_I$. Setting then $\mathring{u}_I\left(\omega,r\right) \doteq e^{i\omega s(r)}\tilde{u}_I\left(\omega,s(r)\right)$, we have that
		\begin{equation}\label{first.split}
			u(\omega,s) =  \frac{u_H(\omega,s) \int_{s}^{+\infty} \tuI(\omega,S) \cdot \hat{H}(\omega,S) dS +   \tuI(\omega,s) \int^{s}_{-\infty} u_H(\omega,S) \cdot\hat{H}(\omega,S) dS}{W(\tuI,u_H)(\omega)}
		\end{equation}
		It is furthermore useful to slightly re-write~\eqref{first.split}. We begin by noting (keeping in mind that $\tuI$ is real valued) that whenever $\omega \in (0,m)$:
		$$\tuI(\omega,s) = \alpha(\omega) u_H(\omega,s)+ \bar{\alpha}(\omega) \bar{u}_H(\omega,s),$$ where
		\[\bar{\alpha} = \frac{W\left(\tilde{u}_I,u_H\right)}{W\left(\bar{u}_H,u_H\right)} = \left(2i\omega\right)^{-1}W\left(\tilde{u}_I,u_H\right).\]
		We then obtain
		\begin{align}\label{second.split}\begin{split}
				&u(\omega,s) =  \frac{-1}{2i\omega} \frac{\overline{W}(\tuI,u_H)(\omega)}{W(\tuI,u_H)(\omega)}  u_H\left(\omega,s\right)\int_{-\infty}^{+\infty} u_H(\omega,S) \cdot \hat{H}(\omega,S) dS\\
				\\ &\qquad + \frac{1}{2i\omega}  \left(  u_H(\omega,s) \int_s^{+\infty} \bar{u}_H(\omega,S)  \hat{H}(\omega,S) dS+  \bar{u}_H(\omega,s) \int^s_{-\infty} u_H(\omega,S)  \hat{H}(\omega,S) dS\right).\end{split}
		\end{align}
		The first line of~\eqref{second.split} contains the crucial ratio of Wronskians which was the focus of Proposition~\ref{ODE.interior.prop}. Finally, we re-write~\eqref{second.split} in terms of $\check{u}$ and the $r$-coordinate:
		\begin{align}\label{third.split}\begin{split}
				&\check{u}(\omega,r) =  \frac{-1}{2i\omega} \frac{\overline{W}(\tuI,u_H)(\omega)}{W(\tuI,u_H)(\omega)}  \check{u}_H\left(\omega,r\right)\int_{-\infty}^{+\infty} u_H(\omega,S) \cdot \hat{H}(\omega,S) dS\\
				\\ &\qquad + \frac{1}{2i\omega}  \left(  \check{u}_H(\omega,r) \int_{s(r)}^{+\infty} \bar{u}_H(\omega,S)  \hat{H}(\omega,S) dS+  e^{i\omega s(r)}\bar{u}_H(\omega,s(r)) \int^{s(r)}_{-\infty} u_H(\omega,S)  \hat{H}(\omega,S) dS\right).\end{split}
		\end{align}

		\subsection{Stationary phase lemmas}\label{stat.section} 
		
		We start by recalling a few standard stationary phase type lemmas. The first is proven with a simple integration by parts whose proof we omit.
		\begin{lemma}\label{stat.phase.basic}
			Let  $a<b$ with $a \in [-\infty,\infty)$, $b \in (-\infty,\infty]$, and $h  \in C^1(a,b) \cap L^1(a,b)$. When $ t  \gtrsim 1$,  we have
			\begin{align}
				\bigl|\int_{a}^{b}e^{-it\omega} h(\omega)\, d\omega\bigr| \ls  t^{-1} \left(\sup_{ a < \omega< b} |h|(\omega) + \int_{a}^{b} |\frac{dh}{d\omega}|(\omega)d\omega\right).\end{align}
		\end{lemma}
		
		Next we recall the Van der Corput lemma.
		\begin{lemma}\label{vandercorput}Let $\phi(x) : (a,b) \to \mathbb{R}$ be smooth and suppose that there exists $k \in \mathbb{Z}_{\geq 1}$ and $\lambda > 0$ such that at least one of the following holds
			\begin{enumerate}
				\item $k \geq 2$ and ${\rm inf}_{(a,b)}\left|\frac{d^k\phi}{dx^k}\right| \geq \lambda$.
				\item $k = 1$, ${\rm inf}_{(a,b)}\left|\phi'(x)\right| \geq \lambda$, and $\phi'$ is monotonic.
			\end{enumerate}
			Then there exists a constant $c_k$ independent of $\phi$, $\lambda$, and $(a,b)$ so that
			\[\left|\int_a^be^{i\phi(x)}\, dx\right| \leq c_k \lambda^{-1/k}.\]
		\end{lemma}
		\begin{proof}See Chapter VIII of~\cite{BigStein}.
		\end{proof}
		
		Finally, we have the following two useful corollaries of Lemma~\ref{vandercorput}.
		\begin{cor}\label{vanitI}Let $\phi(x):[a,b] \to \mathbb{R}$ be smooth, suppose ${\rm inf}_{(a,b)}\left|\phi'(x)\right| \geq \lambda$, and let $\psi: (a,b) \to \mathbb{C}$ be $C^1$ with a continuous extension to $[a,b]$. Then
			\[\left|\int_a^be^{i\phi(x)}\psi(x)\, dx\right| \lesssim  \lambda^{-1}\left[\sup_{x\in (a,b)}\left|\psi(x)\right| + \int_a^b\left|\psi'(x)\right|\, dx\right].\]
			The implied constant is independent of $\phi$, $(a,b)$, $\lambda$, and $\psi$.
		\end{cor}
		\begin{proof}See Chapter VIII of~\cite{BigStein}.
		\end{proof}
		\begin{cor}\label{vanitII}Let $\phi:[a,b] \to \mathbb{R}$ be smooth, suppose $\left|\phi'\right|$ never vanishes and is monotonically decreasing, and let $\psi: (a,b) \to \mathbb{C}$ be $C^1$ with a continuous extension to $[a,b]$. Then
			\[\left|\int_a^be^{i\phi(x)}\psi(x)\, dx\right| \lesssim \frac{\limsup_{x\to b}\left|\psi(x)\right|}{\left|\phi'(b)\right|} +  \int_a^b\frac{\left|\psi'(x)\right|}{\left|\phi'(x)\right|}\, dx. \]
			The implied constant is independent of $\phi$, $(a,b)$, and $\psi$.
		\end{cor}
		\begin{proof}We set $F(x) = \int_a^xe^{i\phi(y)}\, dy$. Then we have
			\begin{align}
				\left|\int_a^be^{i\phi(x)}\psi(x)\, dx\right| &= \left|\int_a^b F'(x) \psi(x)\, dx\right|
				\\ \nonumber &\leq \left|F(b)\right|\limsup_{x\to b}\left|\psi(x)\right| + \int_a^b\left|\psi'(x)\right|\left|F(x)\right|\, dx.
			\end{align}
			The proof then follows from Lemma~\ref{vandercorput}.
		\end{proof}

		We now apply the above stationary phase lemmas from above in order to prove a lemma which will help to control the errors appearing in \eqref{main.formula}.
		\begin{lemma}\label{stat.phase.hard}Let  $A$ be a positive constant satisfying $A\gtrsim 1$, $\delta > 0$ be sufficiently small, and $F :[0,m) \to \mathbb{C}$ be $C^1$. For any $\theta \in [0,1]$, when $ \frac{t}{A}  \gg_{\theta} 1$, we have
			\begin{align}\label{expandityay}
				\bigl|\int_{(1-\theta)m}^{m}e^{-it\omega}e^{i A \left(m^2-\omega^2\right)^{-1/2}}F(\omega)\, d\omega\bigr| \ls_{\delta}  [\frac{t}{A}]^{-1+\delta} \left(\sup_{0 \leq \omega < m} \frac{k^{\frac{1}{2}}}{\log(k)} |F|(\omega)+ \sup_{ 0 \leq \omega < m} k^{-2-\delta} |\rd_{\omega}F|(\omega) \right).
			\end{align}
		\end{lemma}
		\begin{proof}Throughout we assume that $\frac{t}{A} \gg 1$. We emphasize that the implied constants (as per the convention of the paper) may depend on $m$, but will not depend on $A$ or $\omega$. We also allow, in this proof, for all implied constants to depend on $\delta$. It is convenient to set
			\[D \doteq \sup_{0 \leq \omega < m} \frac{k^{\frac{1}{2}}}{\log(k)} |F|(\omega)+ \sup_{ 0 \leq \omega < m} k^{-2-\delta} |\rd_{\omega}F|(\omega).\]
			
			We define the phase function
			\[\Phi\left(t,\omega\right) \doteq \omega - \frac{A}{t} \left(m^2-\omega^2\right)^{-1/2}.\]
			We have
			\begin{equation}\label{phiprimeprime}	
				\rd_{\omega}\Phi(t,\omega) = 1 - t^{-1} \omega A\left(m^2-\omega^2\right)^{-3/2},
			\end{equation}
			\[\rd_{\omega \omega}^2\Phi=- \frac{A}{t} \left(\left(m^2-\omega^2\right)^{-\frac{3}{2}} + 3\omega^2 \left(m^2-\omega^2\right)^{-\frac{5}{2}}\right).\]
			
			In particular, we note that $\partial_{\omega\omega}^2\Phi < 0$ everywhere. One then easily shows that for each value of $t$, there exists a unique value $\omega_c\left(t\right)$, so that
			\[\rd_{\omega}\Phi\left(t,\omega_c\left(t\right)\right) = 0.\]
			And moreover, $\omega_c(t)>0$ obeys the following asymptotics	 	\[\omega_c(t)  = m - \frac{1}{2m}(\frac{t}{Am})^{-\frac{2}{3}} + O\left(\frac{t}{A}\right)^{-\frac{4}{3}}.\]
			In what follows, we will use the notation $\tilde{t}=\frac{t}{A}$. As a consequence of the formulas for $\partial_{\omega}\Phi$ and $\partial^2_{\omega\omega}\Phi$ we can find a small constant $b > 0$ and a large constant $B > 0$ (both independent of $A$ and $t$) so that
			\begin{align}\label{lowerboundphi1}
				\omega \in [m-B\tilde{t}^{-2/3},m-b\tilde{t}^{-2/3}] &\Rightarrow \left|\partial_{\omega}\Phi\right| \gtrsim \tilde{t}^{2/3}\left|\omega-\omega_c\right|,
				\\ \label{lowerboundphi2} \omega \in [0,m-B\tilde{t}^{-2/3}] &\Rightarrow \left|\partial_{\omega}\Phi\right| \gtrsim 1,
				\\ \label{lowerboundphi3} \omega \in [m-b\tilde{t}^{-2/3},m] &\Rightarrow \left|\partial_{\omega}\Phi\right| \gtrsim \tilde{t}^{-1}\left(m-\omega\right)^{-3/2}.
			\end{align}

			We split the interval $[0,m]$ into three sub-intervals as such, defining $\omega_{\pm}(t) = \omega_c(t)\pm \tilde{t}^{-\beta} $: \begin{align*}
				& 	I_1 = [(1-\theta)m, \omega_-(t)],\\   & I_2 = [\omega_-(t), \omega_+(t)],\\  &I_3 = [\omega_+(t), m],
			\end{align*}
			for some $\beta> \frac{2}{3}$ to be determined later (recall that $m-\omega_c(t)=O(\tilde{t}^{-\frac{2}{3}})$). Starting with $I_1$, note that for all $\omega \in I_1$ we have \begin{align*}&  |\rd_{\omega} \Phi|(t,\omega) \gtrsim  {\rm min}\left(1,\tilde{t}^{\frac{2}{3}} |\omega-\omega_c(t)|\right).
			\end{align*} Hence,  Corollary~\ref{vanitII} yields
			\begin{align}
				&\left|\int_{I_1}e^{-it\Phi(\omega)}F\left(\omega\right)\, dw\right| \lesssim t^{-1}\tilde{t}^{-2/3+\beta}\left|F\left(\omega_-(t)\right)\right| + t^{-1}\int_{(1-\theta)m}^{\omega_-(t)}\frac{\left|\partial_{\omega}F\right|}{\left|\partial_{\omega}\Phi\right|}\, d\omega
				\\ \nonumber &\lesssim D\left[t^{-1}\tilde{t}^{-2/3+\beta}\tilde{t}^{-1/6+\delta} + t^{-1}\left(\int_{(1-\theta)m}^{m-B\tilde{t}^{-2/3}} +\int_{m-B\tilde{t}^{-2/3}}^{\omega_-(t)} \right)\left(\omega-m\right)^{-1-\delta/2}\left({\rm min}\left(1,\tilde{t}^{2/3}\left|\omega-\omega_c(t)\right|\right)\right)^{-1}\, d\omega\right]
				\\ \nonumber &\lesssim D \left(\tilde{t}^{-11/6+\beta + \delta} + \tilde{t}^{-1+\frac{\delta}{3}}\right).
			\end{align}
			
			Now we turn to $I_2$. From~\eqref{phiprimeprime}, one may easily check that
			\[{\rm inf}_{\omega \in [\omega_-,\omega_+]} \left|\partial^2_{\omega\omega}\Phi\right|\gtrsim \tilde{t}^{2/3}.\]
			Thus, we may apply Corollary~\ref{vanitI} to obtain
			\begin{align}
				\left|\int_{I_2}e^{-it\Phi}F(\omega)\, d\omega \right| &\lesssim t^{-1/2} \tilde{t}^{-1/3}\left[\sup_{\omega \in [\omega_-,\omega_+]}\left|F(\omega)\right| + \int_{\omega_-}^{\omega_+}\left|\partial_{\omega}F\right|\, d\omega\right]
				\\ \nonumber &\lesssim D  t^{-1/2}\tilde{t}^{-1/3-1/6+\delta} + D  t^{-1/2}\tilde{t}^{-1/3}\tilde{t}^{(2/3)(1+\delta/2)-\beta}
				\\ \nonumber &\lesssim D\left(\tilde{t}^{-1+\delta} + \tilde{t}^{-5/6 + (2/3-\beta) + \delta}\right)
			\end{align}
			
			Now we come to $I_3$.  For this we may apply Corollary~\ref{vanitII} (after sending $x\mapsto -x$) to obtain that (keeping~\eqref{lowerboundphi2} and~\eqref{lowerboundphi3} in mind):
			\begin{align}
				&\left|\int_{\omega_+}^me^{-it\Phi}F(\omega)\, d\omega\right| \lesssim  t^{-1}\tilde{t}^{-2/3+\beta}\left|F\left(\omega_+\right)\right| + t^{-1}\int_{\omega_+}^m\frac{\left|\partial_{\omega}F\right|}{\left|\partial_{\omega}\Phi\right|}\, d\omega
				\\ \nonumber &\lesssim D\Bigg[t^{-1}\tilde{t}^{-2/3+\beta}\tilde{t}^{-1/6+\delta} + t^{-1}\tilde{t}^{(2/3)(1+\delta/2) - (2/3)}\int_{\omega_+}^{m-b\tilde{t}^{-2/3}}\left|\omega-\omega_c\right|^{-1}\, d\omega
				\\ \nonumber &\qquad \qquad \qquad \qquad \qquad \qquad \qquad \qquad \qquad \qquad \qquad \qquad + \int_{m-b\tilde{t}^{-2/3}}^m\left(m-\omega\right)^{1/2-\delta/2}\, d\omega\Bigg]
				\\ \nonumber &\lesssim D\left[\tilde{t}^{-11/6+\beta + \delta} + \tilde{t}^{-1+\delta}\right]
			\end{align}

			We now choose $\beta=\frac{5}{6}$, which concludes the proof of the lemma.
			
		\end{proof}
		Now we turn to the generation of the main term in  \eqref{main.formula}.
		\begin{lemma}\label{stat.phase.medium} Let $m > 0$, $A>0$, $\theta \in (0,1)$. When $ \frac{t}{A}  \gg 1$,  we have, for any small $\delta>0$:
			\begin{align}
				& \int_{(1-\theta)m}^{m}e^{-it\omega}e^{i A \left(m^2-\omega^2\right)^{-\frac{1}{2}}}\, d\omega = e^{-it\omega_c(t)+i A \left(m^2-\omega_c^2(t)\right)^{-\frac{1}{2}}}  (\frac{2\pi i}{3 m^{\frac{1}{3}} A})^{\frac{1}{2}} (\frac{t}{A})^{-\frac{5}{6}} + O_{\delta}\left(\tilde{t}^{-1+\delta}\right),\\ 
			\end{align}
			where the function $\omega_c(t)$ corresponds to the unique critical point of the phase function
			\[\Phi\left(t,\omega\right) \doteq \omega - \frac{A}{t} \left(m^2-\omega^2\right)^{-1/2},\]
			and satisfies
			\begin{align}
				& \omega_c(t)= m - \frac{1}{2m}(\frac{t}{Am})^{-\frac{2}{3}} + O\left(\frac{t}{A}\right)^{-\frac{4}{3}},
				\\  & -t\omega_c(t) +t A \left(m^2-\omega_c^2(t)\right)^{-\frac{1}{2}}= -m t  +\frac{3}{2 }( \frac{A^2 t}{m})^{\frac{1}{3}} +O_{m}(t^{-\frac{1}{3}}).
			\end{align}
		\end{lemma}\begin{proof}
			We assume familiarity with the proof and notation of  Lemma~\ref{stat.phase.hard}. Let $\delta > 0$ be sufficiently small. By Corollary~\ref{vandercorput}, we have
			\[\int_{(1-\theta)m}^me^{-it\Phi}\, d\omega =\int_Ie^{-it\Phi}\, d\omega + O\left(\tilde{t}^{-1+\delta}\right),\]
			where $I \doteq [w_c(t) - \tilde{t}^{-2/3-\delta},w_c(t) + \tilde{t}^{-2/3-\delta}] \doteq [\tilde{\omega}_-,\tilde{\omega}_+].$

			Evaluating at $\omega_c\left(t\right)$ yields:
			
			\begin{equation}\label{secondder}
				\rd_{\omega \omega}^2	\Phi\left(t,\omega_c\right)
				= -3m^{\frac{1}{3}}(\frac{t}{A})^{\frac{2}{3}}\left(1+O(\frac{t}{A})^{-\frac{2}{3}}\right).
			\end{equation}
			Furthermore, for each $n \geq 2$, we have the upper bound
			\[\sup_{\omega \in I}\left|\left(\frac{\partial}{\partial \omega}\right)^n\Phi\right| \lesssim_n \tilde{t}^{\frac{2(n-1)}{3}}.\]
			Now let  $N$ be a positive integer, depending on $\delta$, so that $N\delta \gg 1$. Taylor expanding yields that, for any $\omega \in I$, we have that
			\begin{equation}\label{taylor2}
				\Phi\left(t,\omega\right) = \Phi\left(t,\omega_c\right) + \frac{1}{2}\partial_{\omega\omega}^2\Phi\left(t,\omega_c\right)\left(\omega-\omega_c\right)^2 + a_n(t)\sum_{n=3}^N\left(\omega-\omega_c\right)^n + e_N(t,\omega),
			\end{equation}
			where 
			\[\left|a_n(t)\right| \lesssim_n \tilde{t}^{\frac{2(n-1)}{3}},\qquad \left|e_N(t)\right| \lesssim_N \tilde{t}^{\frac{2N}{3}}\left(\omega-\omega_c\right)^{N+1}.\]
			
			In view of~\eqref{taylor2} and~\eqref{secondder}, we have that $\Phi\left(t,\omega\right) - \Phi\left(t,\omega_c\right)$ is negative on $I$. In particular, we may define a new variable of integration $x$ by
			\begin{equation}\label{formforx}
				x \doteq \sqrt{\frac{\Phi\left(t,\omega\right) - \Phi\left(t,\omega_c\right)}{\partial^2_{\omega\omega}\Phi\left(t,\omega_c\right)}}.
			\end{equation}
			We then have 
			\begin{equation}\label{whatnow}
				\int_Ie^{-it\Phi}\, d\omega = 2e^{-it\Phi\left(t,\omega_c\right)}\int_{x\left(\tilde{\omega}_-\right)}^{x\left(\tilde{\omega}_+\right)}e^{-i t x^2\partial^2_{\omega\omega}\Phi\left(t,\omega_c\right)}\, \underbrace{ \frac{\left|\Phi\left(t,\omega\right) - \Phi\left(t,\omega_c\right)\right|^{1/2}\partial^2_{\omega\omega}\Phi\left(t,\omega_c\right)}{\Phi_{\omega}\left(t,\omega\right)\left|\partial^2_{\omega\omega}\Phi\left(t,\omega_c\right)\right|^{1/2}}}_{\doteq \tilde{H}\left(t,x\right)} dx
			\end{equation}
			Note that we have
			\[\tilde{H}\left(t,x\left(\omega_c\right)\right) =\tilde{H}\left(t,0\right)= \frac{\sqrt{2}}{2},\]
			\begin{equation}\label{derboundderbound}
				\sup_{\omega \in I}\left|\frac{\partial \omega}{\partial x}\right| =  \sup_{\omega \in I}\left|\tilde{H}\right| \lesssim 1.
			\end{equation}
			Furthermore, a calculation with~\eqref{taylor2} and~\eqref{formforx} yields 
			\begin{equation}\label{taylortildeH}
				\tilde{H}\left(t,x\right) = \frac{\sqrt{2}}{2} + \sum_{n=1}^{N-2}\tilde{a}_n(t)\left(\omega(x)-\omega_c\right)^n + \tilde{e}_{N-2}(t,\omega(x)),
			\end{equation}
			where
			\[\left|\tilde{a}_n(t)\right| \lesssim_n \tilde{t}^{\frac{2n}{3}},\qquad \left|\tilde{e}_{N-2}\right|\lesssim \tilde{t}^{\frac{2(N-1)}{3}}\left(\omega(x)-\omega_c\right)^{N-1}.\]
			We then split~\eqref{whatnow} into two pieces
			\begin{align}\label{splittingit}
				&2e^{-it\Phi\left(t,\omega_c\right)}\int_{x\left(\tilde{\omega}_-\right)}^{x\left(\tilde{\omega}_+\right)}e^{-i t x^2
					\partial^2_{\omega\omega}\Phi\left(t,\omega_c\right)}\,\tilde{H}\left(t,x\right) dx = \sqrt{2} e^{-it\Phi\left(t,\omega_c\right)} \int_{x\left(\tilde{\omega}_-\right)}^{x\left(\tilde{\omega}_+\right)}e^{-i t x^2\partial^2_{\omega\omega}\Phi\left(t,\omega_c\right)}dx 
				\\ \nonumber &\qquad + 2e^{-it\Phi\left(t,\omega_c\right)}\int_{x\left(\tilde{\omega}_-\right)}^{x\left(\tilde{\omega}_+\right)}e^{-i t x^2\partial^2_{\omega\omega}\Phi\left(t,\omega_c\right)}\, \left(\tilde{H}\left(t,x\right)-\tilde{H}\left(t,x\left(\omega_c\right)\right)\right)dx \doteq I_1+I_2.
			\end{align}
			
			For the first integral, a standard argument (see Proposition 3 in Chapter VIII of~\cite{BigStein}) yields that 
			\[I_1 = e^{-it\omega_c(t)+i A \left(m^2-\omega_c^2(t)\right)^{-\frac{1}{2}}}  (\frac{2\pi i}{3 m^{\frac{1}{3}} A})^{\frac{1}{2}} (\frac{t}{A})^{-\frac{5}{6}} + O\left(\tilde{t}^{-1+\delta}\right).\]
			Next, setting $x_{\pm} \doteq x\left(\tilde{\omega}_{\pm}\right)$, we consider the integrals
			\[J_n \doteq \tilde{a}_n(t)\int_{x_-}^{x_+}e^{-i t x^2\partial^2_{\omega\omega}\Phi\left(t,\omega_c\right)}\left(\omega-\omega_c\right)^ndx.\]
			We carry out a change of variables $u = t\partial_{\omega\omega}^2\Phi\left(t,\omega_c\right)x^2$ on $[x_-,0]$ and $[0,x_+]$ separately, and we obtain that 
			\begin{align}\label{firstboundyayi2}
				\left|J_n\right| &\lesssim \tilde{t}^{-\frac{n}{6}-\frac{5}{6}}\left[\left|\int_{u\left(x_-\right)}^0e^{iu}\frac{(\omega(u)-\omega_c)^n}{(x(u))^n}u^{\frac{n-1}{2}}\, du\right|+\left|\int_0^{u\left(x_+\right)}e^{iu}\frac{(\omega(u)-\omega_c)^n}{(x(u))^n}u^{\frac{n-1}{2}}\, du\right|\right]
				\\ \nonumber &\lesssim_n \tilde{t}^{\frac{1}{6}\left(-n + n-1\right) - 5/6} \\ \nonumber &\lesssim \tilde{t}^{-1}.
			\end{align}
			To obtain the second  inequality we have integrated by parts (repeatedly) with the exponential, used that $\left|u(x_{\pm})\right|^{\frac{n-1}{2}} \lesssim \tilde{t}^{\frac{n-1}{6}}$, used that $x \sim \left(\omega-\omega_+\right)$ within $I$, used that $\left|x^{-1}\partial_ux\right| \lesssim u^{-1}$, and also~\eqref{derboundderbound}. 
			
			Next, we note that we have
			\[\left|x\left(\omega_+\right) - x\left(\omega_-\right)\right| \lesssim \tilde{t}^{-2/3-\delta}.\]
			In view of bound for $J_n$ and~\eqref{taylortildeH}, we have 
			\begin{align}
				\left|I_2\right| &\lesssim_n \tilde{t}^{-1} + \tilde{t}^{\frac{2(N-1)}{3}}\int_{x_-}^{x_+}\left|\omega-\omega_c\right|^{N-1}\, dx
				\\ \nonumber &\lesssim \tilde{t}^{-1} + \tilde{t}^{\frac{2(N-1)}{3}}\tilde{t}^{-2/3-\delta}\tilde{t}^{-\left(2/3+\delta\right)(N-1)} \\ \nonumber &\lesssim \tilde{t}^{-1} + \tilde{t}^{-2/3 -\delta N} 
				\\ \nonumber &\lesssim \tilde{t}^{-1}.
			\end{align}
			Summing up all of the bounds then completes the proof.

		\end{proof}

		\subsection{Decay for  inhomogeneous Klein--Gordon equation with zero  data}
		
		In this subsection, we will establish the decay result, combining the frequency-analysis of Section~\ref{yakov.section} and Section~\ref{maxime.section}  with the stationary phases results of Section~\ref{stat.section}. We will split the three regimes: $\omega \in  \left((1-\theta)m, m \right)$ (by far the most delicate), $\omega \in  \left(m, (1+\theta)m \right)$ and  $\omega > (1+\theta) m $ (the easiest regime). 
		
		Before we turn to the decay result, we record a useful difference estimate between  $\check{u}_H(\omega,r)$ and $\check{u}_H(m,r)$.  \begin{lemma}\label{uH.lemma} We introduce the notations $\uhs=  \check{u}_H(m,r)$ and $\delta\check{u}_H(\omega,r) = \check{u}_H(\omega,r)-\uhs$.  Let $\theta \in (0,1)$. For any $R > r_+$, there exists $D(R)>0$, such that for all $\omega \in ( (1-\theta)m , (1+\theta)m)$ and $r \in [r_+,R]$:
			\begin{align}
				& |\duh|(\omega,r),|\rd_{\omega} \duh|(\omega,r)  \ls  D(R) \cdot  k^{-2},\label{DUH} \end{align}
			
		\end{lemma}\begin{proof}  This follows  immediately from the analyticity of $\omega \rightarrow \check{u}_H(\omega,r)$  uniformly for $r \in [r_+,R]$ proven in Lemma~\ref{wherewemakeuhui}. Note indeed that $|\omega - m| \ls k^{-2}$.\end{proof}
		Now we address  the regime $\omega \in  \left((1-\theta)m, m \right)$.
		\begin{prop}\label{decay.prop}
			Assume that $\int_{-\infty}^{+\infty}| \left(1+s_-^2\right)\hat{H}|(\omega,S)  dS,\ \int_{-\infty}^{+\infty}| \rd_{\omega}\check{H}|(\omega,S)  dS < +\infty$, and moreover that   $\hat{H}(\omega,s(r))$, $\rd_{\omega} \hat{H}(\omega,s(r))$ are supported on $\{ r \in [r_+,R]\}$ for some $R>0$ independent of $\omega$. Here $s_- \doteq {\rm min}\left(0,s\right)$.

			Let $R_0>r_+$. Then there exists $D(R,R_0)>0$ such that the solution $u$ of \eqref{second.split} obeys the following estimate, for $\theta\in (0,1)$ sufficiently close to $1$ (recall $|\gamma_m|<1$): for all  $s<R_0$:  \begin{align}
				&\bigl| \int_{(1-\theta)m}^{m} e^{-i \omega t^* } \check{u}(\omega,r) d\omega - C(M,e,m^2)  \cdot \mathcal{L}[\hat{H}]  \cdot \uhs(r) \cdot (t^*)^{-\frac{5}{6}}\cdot [\sum_{q=1}^{+\infty} (-\gamma_m)^{q-1}q^{\frac{1}{3}}e^{it^*\Phi_q(t^*)+i [\phi_-(m)-\frac{\pi}{4}]} ] \bigr| \\ &\ls D(R,R_0) \cdot  \sup_{\omega \in ((1-\theta)m,m)} \left[\int_{-\infty}^{+\infty} \left((1+s_-^2)|  \hat{H}|(\omega,S) + |\rd_{\omega}\hat{H}|(\omega,S)\right) dS\right]\cdot    (t^*)^{-1+\delta},
				\\ \nonumber & \Phi_q(t^*)=\omega_q(t^*)-2\pi m^2M q \cdot (t^*)^{-1} \left(m^2-\omega_q^2(t)\right)^{-\frac{1}{2}}
				\\ \nonumber &\qquad = m  -\frac{3}{2} [2\pi M]^{\frac{2}{3}} m q^{\frac{2}{3}} (t^*)^{-\frac{2}{3}}  +O(t^{-\frac{4}{3}}),
				\\ \nonumber &\qquad  \omega_q(t^*):=m - \frac{(4\pi^2 m^2 M)^{\frac{1}{3}}}{2} q^{\frac{2}{3}} (t^{*})^{-\frac{2}{3}}, \\  & C(M,e,m^2):= 2^{-\frac{1}{6}} 3^{-\frac{1}{2}} \pi^{\frac{4}{3}}[1-|\gamma_m|^2]  M^{\frac{2}{3}} m^{\frac{1}{6}} >0 ,\\ & \mathcal{L}[\hat{H}]:= \int_{-\infty}^{+\infty}u_H\left(m,s\right) \hat{H}(m,S) dS,\end{align}
			where $\uhs(r)=\check{u}_H(m,r)$.
		\end{prop}
		
		\begin{proof}  In the whole proof, we will assume that $r \in [r_+, R_0]$, for some $R_0> r_+$ so that $|\check{u}_H|(\omega,r),\ |\rd_{\omega}\check{u}_H|(\omega,r)   \leq \tilde{D}(R_0)$ for all $\omega \in \mathbb{R}$, $r \in [r_+,R_0]$. We write \eqref{third.split} as such, combining with 	\eqref{W.uI.uH.ratio} of Proposition~\ref{ODE.interior.prop}  	\begin{equation}\begin{split}
					& 	\check{u}(\omega,r) =  \tilde{G}(r)  e^{i\phi_-(m)} [1-|\gamma_m|^2]\sum_{q=1}^{+\infty}  (-\gamma_m)^{q-1}e^{2i\pi k q } + \sum_{N=0}^{+\infty}\tilde{F}_N(\omega)e^{2i\pi k N }+ h(\omega,r),\\ &  G(\omega,r)= - \frac{1}{2i\omega}  \check{u}(\omega,r) \int_{-\infty}^{+\infty} u_H(\omega,s)\cdot \hat{H}(\omega,s) ds,
					\\ \nonumber &\tilde{G}\left(r\right) = G\left(m,r\right),
					\\ & \tilde{F}_N(\omega,r)= - \frac{F_N(\omega)}{2i\omega} \check{u}_H(\omega,r)\int_{-\infty}^{+\infty} u_H(\omega,s) \cdot \hat{H}(\omega,s) ds \text{ for } N \geq 1, 
					\\ & \tilde{F}_0(\omega,r)=  \frac{\gamma(\omega) e^{i\phi_+(\omega)}- \gamma_m e^{i\phi_+(m)}}{2i\omega} \check{u}_H(\omega,r)\int_{-\infty}^{+\infty} u_H(\omega,s) \cdot \hat{H}(\omega,s) ds + G\left(\omega,r\right) - G\left(m,r\right),  \\  &   h(\omega)=\check{u}_H(\omega,r) \int_{s(r)}^{+\infty} \bar{u}_H(\omega,s)  \hat{H}(\omega,S) dS
					\\ \nonumber &\qquad +  e^{i\omega p(r)}\bar{u}_H(\omega,s) \int^{s(r)}_{-\infty} u_H(\omega,s)  \hat{H}(\omega,S) dS- \frac{\gamma_m\cdot  e^{i\phi_+(m)}}{2i\omega} \check{u}_H(\omega,r)\int_{-\infty}^{+\infty} u_H(\omega,s) \cdot \hat{H}(\omega,s) ds.\end{split}
			\end{equation}
			For $h$, we use Lemma~\ref{stat.phase.basic} with $a= (1-\theta)m$ and $b= m$   combined with the boundedness  of $\omega \rightarrow \check{u}_H(\omega,r)$ and $\omega \rightarrow \rd_{\omega} \check{u}_H(\omega,r)$ on $[r_+,R]$   (see Section~\ref{yakov.section}), taking advantage of the compactly support of $\hat{H}$   to obtain \begin{align}\label{eq1}\begin{split}
					&\bigl| \int_{(1-\theta)m}^{m} e^{-i \omega t } h(\omega,s) d\omega \bigr| \ls 
					\\  &\qquad D_1(R) \cdot  (t^*)^{-1} \cdot  \sup_{\omega \in ((1-\theta)m,m)} \left[\int_{-\infty}^{+\infty} \left((1+s_-^2)|  \hat{H}|(\omega,S) + |\rd_{\omega}\hat{H}|(\omega,S)\right) dS\right].\end{split}
			\end{align}
			Next, we apply	 Lemma~\ref{stat.phase.hard} and the bounds of Proposition~\ref{ODE.interior.prop} (including the estimate on $(\gamma(\omega) e^{i\phi_+(\omega)}- \gamma_m e^{i\phi_+(m)})$) to obtain similarly 
			\begin{align}\label{eq2}\begin{split}
					&\bigl| \sum_{N=0}^{+\infty}\int_{(1-\theta)m}^{m} e^{-i \omega t^* } \tilde{F}_N(\omega,r) d\omega \bigr| \ls 
					\\  &\qquad D_2(R) \cdot  (t^*)^{-1+\delta} \cdot  \sup_{\omega \in ((1-\theta)m,m)} \left[\int_{-\infty}^{+\infty} \left((1+s_-^2)|  \hat{H}|(\omega,S) + |\rd_{\omega}\hat{H}|(\omega,S)\right) dS\right].\end{split}
			\end{align}
			Finally, we apply Lemma~\ref{stat.phase.medium} (with $\theta_+=0$) to $\tilde{G}(r) e^{i \pi k q}$ with $A_q=   m^2 M q $. 
			We end up with \begin{equation}\label{eq3}\begin{split}
					&\bigl| \int_{(1-\theta)m}^{m} e^{-i \omega t^* } \tilde{G}(r) e^{i k \pi q } d\omega - \tilde{G}(r) \cdot(\frac{2\pi i }{3m^{\frac{1}{3}} A_1})^{\frac{1}{2}} \cdot (\frac{t^{*}}{A_1})^{-\frac{5}{6}}\cdot q^{\frac{1}{3}}e^{it^*\Phi_q(t^*)+i [\phi_-(m)-\frac{\pi}{4}]}  \bigr|\\ &\ls  D(R,R_0) \cdot (\frac{t^*}{q})^{-1+\delta} \cdot \int_{-\infty}^{+\infty} |\hat{H}|(m,S) dS.\end{split}
			\end{equation}
			Combining \eqref{eq1}, \eqref{eq2}, \eqref{eq3} concludes the proof of the proposition.
		\end{proof}
		
		Now we turn to the regime $\omega  \in (m, (1+\theta)m)$.
		\begin{prop}\label{decay.easy.prop}
			Under the assumptions of Proposition~\ref{decay.prop}, let $u$ be the solution of \eqref{first.split}. Then, $u$ obeys the following estimate: for all $r \in [r_+,\infty)$ \begin{align}
				&\bigl| \int_{m}^{(1+\theta)m} e^{-i \omega t^* } \check{u}(\omega,r) d\omega \bigr| \ls
				\\ &\qquad (t^*)^{-1} \cdot  \sup_{\omega \in (m,(1+\theta)m)}\left[  \int_{-\infty}^{+\infty}\left( \left(1+s_-^2\right)|  \hat{H}|(\omega,S) + |\rd_{\omega}\hat{H}|(\omega,S)\right) dS \right].
			\end{align}
		\end{prop}
		\begin{proof}
			For $\omega>m$, we define $\check{u}_I(\omega,s)= (\omega^2-m)^{-\frac{1}{4}} u_I(\omega,s)$ and $\check{W}(\omega) =  (\omega^2-m)^{-\frac{1}{4}} W(\omega)$ (in the notations of Proposition~\ref{repuform}), so that we have, equivalently to \eqref{first.split} multiplied by $e^{i\omega p(r)}$:  \begin{equation}
				\check{u}(\omega,s) = \check{W}^{-1}(\omega)\left[e^{i\omega p(r)}\check{u}_I(\omega,s)\int_{-\infty}^su_H(\omega,S)\hat{H}(\omega,S)\, dS+ \check{u}_H(\omega,s)\int_s^{+\infty}\check{u}_I\left(\omega,S\right)\hat{H}\left(\omega,S\right)\, dS \right]. 
			\end{equation} Then we use Proposition~\ref{repuform} and Proposition~\ref{betterbounduI} to obtain for all $\omega \in (m,(1+\theta)m)$  $$ |\check{u}|(\omega,r)  \ls \int_{-\infty}^{+\infty}|\hat{H}|\left(\omega,S\right) dS,  $$  
			$$ |\rd_\omega \check{u}|(\omega,r)  \ls (\omega-m)^{-\frac{1}{2}- 2\delta} \int_{-\infty}^{+\infty}\left((1+s_-^2)|\hat{H}|\left(\omega,S\right)+|\rd_{\omega}\hat{H}|\left(\omega,S\right) \right)dS . $$ 
			Then, applying Lemma~\ref{stat.phase.basic} gives the proof of the proposition. \end{proof}
		
		Finally, we take care of the frequencies away from $\omega = \pm m$, where the solution is regular. \begin{prop}\label{decay.trivial.prop}
			Under the assumptions of Proposition~\ref{decay.prop}, let $u$ be the solution of \eqref{first.split}. Then, $\check{u}$ obeys the following estimate: \begin{align}
				&\bigl| \int_{0}^{(1-\theta)m} e^{-i \omega t^* } \check{u}(\omega,r) d\omega \bigr|+ \bigl| \int^{+\infty}_{(1+\theta)m} e^{-i \omega t^* } \check{u}(\omega,r) d\omega \bigr| \ls
				\\ \nonumber &\qquad (t^*)^{-1} \cdot \sup_{\omega \in [0,\infty)]}\left[  \int_{-\infty}^{+\infty}\left( |  \hat{H}|(\omega,S) + |\rd_{\omega}\hat{H}|(\omega,S) \right)dS + \sup_{s \leq 0}s^2\int_{-\infty}^s\left|\hat{H}\right|(\omega,S)\, dS \right].
			\end{align}
		\end{prop}\begin{proof}
			First, we discuss $\int_{0}^{(1-\theta)m} e^{-i \omega t^{*} } u(\omega,s) d\omega$.  Using \eqref{first.split} multiplied by $e^{i\omega p(r)}$ and the fact that $\hat{H}$ is supported on $\{ s<R\}$, we get (using also that  $\frac{\tuI(\omega,s)}{W(\tuI,u_H)(\omega)}=  \frac{u_I(\omega,s)}{W(u_I,u_H)(\omega)}$) 	$$\check{u}(\omega,r) =  \frac{\check{u}_H(\omega,s) 1_{(-\infty,R)}(s) \int_{s}^{+\infty} u_I(\omega,s) \cdot \hat{H}(\omega,s) ds +  e^{i\omega p(r)} u_I(\omega,s) \int^{s}_{-\infty} u_H(\omega,s) \cdot\hat{H}(\omega,s) ds}{W(u_I,u_H)(\omega)}.$$
			
			We integrate  $\int_{0}^{(1-\theta)m} e^{-i \omega t^* } \check{u}(\omega,r) d\omega$ using the above formula and use Lemma~\ref{stat.phase.basic}. The relevant bounds are straightforward to obtain using  Lemma~\ref{makeuh} to control $\check{u}_H$ and $\rd_{\omega} \check{u}_H$ (note that we only need to do so in the region $\{ r \in [r_+,R]\}$) and Proposition~\ref{repuform} to control $|W|^{-1}(u_I,u_H)(\omega),\ |\frac{dW(u_I,u_H)(\omega)}{d\omega}|\ls 1$  (where the implicit constant depends on $\theta$). As a result, we obtain \begin{align}\label{trivial1}\begin{split}
					&\bigl| \int_{0}^{(1-\theta)m} e^{-i \omega t^* } \check{u}(\omega,r) d\omega \bigr| \ls
					\\  &\qquad (t^*)^{-1} \cdot \sup_{\omega \in (0,(1-\theta)m)}\left[  \int_{-\infty}^{+\infty} \left((1+s_-^2)|  \hat{H}|(\omega,S) + |\rd_{\omega}\hat{H}|(\omega,S)\right) dS  \right].\end{split}\end{align} For $\int^{+\infty}_{(1+\theta)m} e^{-i \omega t } \check{u}(\omega,r) d\omega$, we can use still use  Lemma~\ref{makeuh} and Proposition~\ref{repuform}, but this time  we have the improved estimate $|W|^{-1}(u_I,u_H)(\omega) \ls (1+\omega)^{-1},\ |\frac{dW(u_I,u_H)(\omega)}{d\omega}|\ls 1$ from which we get $|\frac{d}{d\omega}( \frac{1}{W(u_I,u_H)})|\ls (1+\omega)^{-2}$ is integrable: hence similarly an application of Lemma~\ref{stat.phase.basic} gives  \begin{align}\label{trivial2}\begin{split}
					&\bigl| \int^{+\infty}_{(1+\theta)m} e^{-i \omega t^* } \check{u}(\omega,r) d\omega \bigr| \ls
					\\&\qquad  (t^*)^{-1} \cdot\sup_{\omega \in ((1+\theta)m,+\infty)}\left[  \int_{-\infty}^{+\infty}\left( (1+s_-^2)|  \hat{H}|(\omega,S) + |\rd_{\omega}\hat{H}|(\omega,S)\right) dS  \right].\end{split}\end{align} Combining \eqref{trivial1} and \eqref{trivial2} concludes the proof of the lemma.
		\end{proof}                                                                                                                                                                                                             
		\subsection{The cut-off argument}
		In this section we will use a cutoff argument finish the proof of our main result. 
		
		\begin{prop}\label{final.decay.prop}
			Let $\phi$ be a smooth solution  to 
			\begin{align}
				& \left(\Box_g-m^2\right)\phi = 0, \label{main.eq}
				\\ &  (\phi,n_{\{t^*=0\}} \phi)|_{|(t^*=0,r)}= (\phi_0(r),\phi_1(r)) \label{main.eq2},
			\end{align} 
			with $\phi_0$ and $\phi_1$ vanishing for sufficiently large $r$, and $n_{\{t^*=0\}}$ denotes the normal to $\{t^* = 0\}$. Then, using the notations of Proposition~\ref{decay.prop} with $\tilde{C}(M,e,m^2)=\sqrt{\frac{2}{\pi}} \cdot C(M,e,m^2)>0$, $|\gamma_m|<1$, and $\uhs(s) =|\uhs|(s)   e^{i\check{\varphi}_H(s)}$, for all $t^*\geq 1$ and $r \in [r_+,R_0]$, we have \begin{align} \begin{split}
					&\psi(t^*,r)= 
					\\ \nonumber &\tilde{C}(M,e,m^2)  \cdot \bigl| \mathcal{L}[\phi_0,\phi_1] \bigr| \cdot |\uhs|(s) \cdot (t^*)^{-\frac{5}{6}}\cdot \left[\sum_{q=1}^{+\infty} (-\gamma_m)^{q-1}q^{\frac{1}{3}}\cos(\Phi_q(t^*)+ \phi_-(m)-\frac{\pi}{4}+ \check{\varphi}_H(s)+\theta[\phi_0,\phi_1]) \right] +\mathcal{E}(t^*,s),\\ &    \mathcal{L}[\phi_0,\phi_1] = | \mathcal{L}[\phi_0,\phi_1] | e^{i \theta[\phi_0,\phi_1]} =  2\int_{r_+}^{+\infty} e^{-ip(r) m}u_H(m,s(r))\left(-imA_1(r) \phi +2A_2(r)\phi + A_1(r) \partial_{t^*}\phi + 2A_3(r)\partial_r\phi\right)|_{t^*=0}\ r\, dr, 
					\\ &  |\mathcal{E}|(t^*,s) \ls D(R_0,R, \left(\phi_0,\phi_1\right)) \cdot (t^*)^{-1+\delta},
				\end{split}
			\end{align}
			where the functions $A_1(r)$, $A_2(r)$, and $A_3(r)$ are given by~\eqref{commutatorboxftstar}:
		\end{prop}
		\begin{proof} 
			Let $\Upsilon\left(t^*\right)$ be a smooth function of $t^*$ which vanishes for $t^* \leq 0$ and is identically $1$ for $t^* \leq 1$. For all sufficiently small $\epsilon > 0$ we define $\Upsilon_{\epsilon} \doteq \Upsilon\left(\frac{t^*}{\epsilon}\right)$ and $\phi_{\ep}(t^*,r) = \phi(t^*,r) \ce(t^*)$. In view of~\eqref{commutatorboxftstar}, we have 
			\begin{equation}\label{cutoffequationpsiep}
				\left(\Box_g-m^2\right)\phi_{\epsilon} = F_{\epsilon},
			\end{equation}
			where 
			\[F_{\epsilon} = \epsilon^{-2} A_1\left(r\right) \ce''(t^*)\phi + \epsilon^{-1}2A_2(r) \ce'(t^*)\phi + \epsilon^{-1}2A_1\left(r\right)\ce'(t^*)\partial_{t^*}\phi + \epsilon^{-1}2A_3\left(r\right)\ce'(t^*)\partial_r\phi,\]
			and we recall that $A_1$, $A_2$, and $A_3$ are given by~\eqref{commutatorboxftstar}. In turn, $\psi_{\ep} \doteq r\phi_{\ep}$ satisfies \eqref{eq:main} with 
			\[H \doteq H_{\epsilon} \doteq r\left(1-\frac{2M}{r}+\frac{\mathcal{D}}{r^2}\right)F_{\epsilon}.\]
			
			For any $\epsilon > 0$,  it is straightforward to see that there exists $D\left(\epsilon,\left(\phi_0,\phi_1\right)\right) < \infty$ such that
			\begin{equation}\label{hepisokwithme}
				\sup_{\omega \in (0,+\infty)}\left[  \int_{-\infty}^{+\infty} \left(|  (1+s_{-}^2)\hat{H}_{\epsilon}|(\omega,S) + |\rd_{\omega}\hat{H}_{\epsilon}|(\omega,S) \right)dS  \right] \leq D\left(\epsilon,\phi_0,\phi_1\right).
			\end{equation}
			
			We fix $R_0>0$ for the rest of the proof. In view of~\eqref{hepisokwithme}, applying Proposition~\ref{decay.prop}, we thus get for all $r \in [r_+,R_0]$: \begin{align}\label{decay1}\begin{split}
					&\bigl| \int_{(1-\theta)m}^{m} e^{-i \omega t^* } \check{u}_{\epsilon}(\omega,r) d\omega - C(M,e,m^2)  \cdot \mathcal{L}[\hat{H}_{\ep}]  \cdot \uhs(r) \cdot (t^*)^{-\frac{5}{6}}\cdot [\sum_{q=1}^{+\infty} (-\gamma_m)^{q-1}q^{\frac{1}{3}}e^{it^*\Phi_q(t^*)+i [\phi_-(m)-\frac{\pi}{4}]} ] \bigr| 
					\\ &\qquad \ls D\left(\epsilon,\phi_0,\phi_1\right) \cdot  (t^*)^{-1+\delta} . \end{split}\end{align} 
			Similarly, combining~\eqref{decay1} with Proposition~\ref{decay.easy.prop} and Proposition~\ref{decay.trivial.prop} we obtain, \begin{align}\label{decay2}\begin{split}
					&\bigl| \int_{0}^{+\infty} e^{-i \omega t^* } \check{u}_{\epsilon}(\omega,r) d\omega - C(M,e,m^2)  \cdot \mathcal{L}[\hat{H}_{\epsilon}]  \cdot \uhs(r) \cdot (t^*)^{-\frac{5}{6}}\cdot [\sum_{q=1}^{+\infty} (-\gamma_m)^{q-1}q^{\frac{1}{3}}e^{it^*\Phi_q(t^*)+i [\phi_-(m)-\frac{\pi}{4}]} ] \bigr|
					\\ &\qquad \ls D\left(\epsilon,\phi_0,\phi_1\right) \cdot  (t^{*})^{-1+\delta} . \end{split}\end{align} 
			
			Now, since $\psi_{\ep}$ is real-valued, it means that $\check{u}(-\omega,r)= \overline{\check{u}_{\epsilon}}(\omega,r)$. Hence $$ \int_{-\infty}^{0} e^{-i \omega t^* } \check{u}_{\epsilon}(\omega,r) d\omega  =  \int^{+\infty}_{0} e^{i \omega t^* } \check{u}_{\epsilon}(-\omega,r) d\omega= \overline{ \int_{0}^{+\infty} e^{-i \omega t^* } \check{u}_{\epsilon}(\omega,r) d\omega}.$$ 
			
			Hence $$ \psi_{\ep}(t^*,r) = \frac{1}{2\pi}\int_{-\infty}^{+\infty}  e^{-i \omega t^* } \check{u}_{\epsilon}(\omega,r) d\omega = \frac{1}{\pi}\Re \left(  \int_{0}^{+\infty}  e^{-i \omega t^* } \check{u}_{\epsilon}(\omega,r) d\omega \right) . $$ Combining this equality to \eqref{decay2}, recalling we defined $\tilde{C}(M,e,m^2)=\frac{1}{\pi} C(M,e,m^2)$ gives, writing the constant $\mathcal{L}[\hat{H}_{\ep}] =|\mathcal{L}[\hat{H}_{\ep}]|  e^{i \theta[\hat{H}_{\ep}] }$: 
			\begin{align}\label{profile.before.limit}\begin{split}
					&	\bigl|  \psi_{\ep}(t,s) - \tilde{C}(M,e,m^2)  \cdot |\mathcal{L}[\hat{H}_{\ep}]|  \cdot |\uhs|(s) \cdot (t^{*})^{-\frac{5}{6}}\cdot \left[\sum_{q=1}^{+\infty} (-\gamma_m)^{q-1}q^{\frac{1}{3}} \cos(t^{*}\Phi_q(t^{*})+ \phi_-(m)-\frac{\pi}{4}+\check{\varphi}_H(s)+ \theta[\hat{H}_{\ep}] )\right] \bigr| 
					\\ &\qquad \ls D\left(\epsilon,\phi_0,\phi_1\right) \cdot  (t^*)^{-1+\delta}.
			\end{split}\end{align} 
			
			For all $\epsilon > 0$ sufficiently small, $\psi_{\ep} = \psi$ when $t^* \geq 1$. Thus, $\mathcal{L}[\hat{H}_{\ep}]$ must in fact be independent of $\epsilon$ when $\epsilon$ is sufficiently small.  We may thus compute $\mathcal{L}[\hat{H}_{\ep}]$ in the limit $\ep\rightarrow 0$, while taking, $\epsilon$ to be $1/2$ in \eqref{profile.before.limit}. A short computation shows that $$ \lim_{\ep \rightarrow 0} \hat{H}_{\epsilon}(m,s(r)) = e^{-im p(r)}r\left(1-\frac{2M}{r}+\frac{\mathcal{D}}{r^2}\right)\left(-imA_1(r) \phi +2A_2(r)\phi + A_1(r) \partial_{t^*}\phi + 2A_3(r)\partial_r\phi\right)|_{t^*=0}.$$
			
			In conclusion, we may rewrite the integration with respect to the regular $r$-coordinate and use the dominated convergence theorem to obtain   \begin{equation}
				\mathcal{L}[\hat{H}]  =  \int_{r_+}^{+\infty} e^{-imp(r) }u_H(m,s(r))\left(-imA_1(r) \phi +2A_2(r)\phi + A_1(r) \partial_{t^*}\phi + 2A_3(r)\partial_r\phi\right)|_{t^*=0}\ r\, dr
			\end{equation} As per the above discussion, combining this with \eqref{profile.before.limit} finishes the proof.

		\end{proof}
		\appendix
		
		\section{Special functions}

		In this section, we collect definitions and facts about special functions that we have used in the proofs. We will use the notation $f = O(g)$ to mean that there exists a constant $C$ such that $|f(x)| \leq C |g(x)|$ for $x$ in a given prescribed range.
		
		\subsection{Bessel functions}\label{bessel.appendix}
		
		The Bessel functions $J_1: (1,\infty) \to \RR$ and $Y_1: (1,\infty) \to \RR$ are solutions to the differential equation:
		\begin{equation}\label{eq:besseldef}
			x^2 \frac{d^2w}{dx^2} + x \frac{dw}{dx} + (x^2-1) w = 0.
		\end{equation}
		
		We have the follow asymptotic statements for $J_1$ and $Y_1$ as $x\to +\infty$:
		\begin{align}
			\left|J_1(x)-\sqrt{\frac{2}{\pi x}}\cos\left(x-\frac{3\pi}{4}\right)\right| = O\left(x^{-3/2}\right),\qquad \left|Y_1(x) - \sqrt{\frac{2}{\pi x}}\sin\left(x-\frac{3\pi}{4}\right)\right| = O\left(x^{-3/2}\right),
		\end{align}
		\begin{align}
			\left|J_1'(x)+\sqrt{\frac{2}{\pi x}}\sin\left(x-\frac{3\pi}{4}\right)\right| = O\left(x^{-3/2}\right),\qquad \left|Y'_1(x) - \sqrt{\frac{2}{\pi x}}\cos\left(x-\frac{3\pi}{4}\right)\right| = O\left(x^{-3/2}\right).
		\end{align}
		\subsection{Airy functions $\Ai$ and $\Bi$ and the weight functions $E$, $M$, and $N$}
		\label{app:airy}
		The functions $\Ai: \RR \to \RR$ and $\Bi: \RR\to \RR$ are solutions to the differential equation
		\begin{equation}\label{eq:airydef}
			\frac{d^2w}{dx^2} = x w. 
		\end{equation}
		satisfying the following asymptotics as $x \to -\infty$:
		\begin{equation*}
			\Big|\Ai(x) - \pi^{-\frac 12} |x|^{-\frac 14} \sin \Big( \frac 23 |x|^{\frac 32}+ \frac \pi 4\Big)\Big| = O(|x|^{- \frac 74}),  \qquad \Big|\Bi(x) - \pi^{-\frac 12} |x|^{-\frac 14} \cos\Big( \frac 23 |x|^{\frac 32}+ \frac  \pi 4 \Big)\Big| = O(|x|^{-\frac 74}),
		\end{equation*}
		
		\begin{equation*}
			\Big|\Ai'(x)  + \pi^{-\frac 12} |x|^{\frac 14} \cos \Big( \frac 23 |x|^{\frac 32}+ \frac \pi 4\Big)\Big| =O(|x|^{-\frac 54}),  \qquad \Big|\Bi'(x)  - \pi^{-\frac 12} |x|^{\frac 14} \sin\Big( \frac 23 |x|^{\frac 32}+ \frac  \pi 4 \Big)\Big| = O(|x|^{-\frac 54}),
		\end{equation*}
		
		We also have the following asymptotics as $x \to+ \infty$:
		\begin{equation*}
			\left|\Ai(x) - \frac{1}{2\pi^{1/2}x^{1/4}}\exp\left(-\frac{2}{3}x^{3/2}\right)\right| = O\left(x^{-7/4}\exp\left(-\frac{2}{3}x^{3/2}\right)\right), 
		\end{equation*}
		\begin{equation*}
			\left|\Bi(x) - \frac{1}{\pi^{1/2}x^{1/4}}\exp\left(\frac{2}{3}x^{3/2}\right)\right| = O\left(x^{-7/4}\exp\left(\frac{2}{3}x^{3/2}\right)\right),
		\end{equation*}
		
		\begin{equation*}
			\left|\Ai'(x) +   \frac{x^{1/4}}{2\pi^{1/2}}\exp\left(-\frac{2}{3}x^{3/2}\right)\right| = O\left(x^{-5/4}\exp\left(-\frac{2}{3}x^{3/2}\right)\right),
		\end{equation*}
		\begin{equation*}
			\left|\Bi'(x) - \frac{x^{1/4}}{\pi^{1/2}}\exp\left(\frac{2}{3}x^{3/2}\right)\right| = O\left(x^{-5/4}\exp\left(\frac{2}{3}x^{3/2}\right)\right).
		\end{equation*} 
		
		We also recall the definition of the functions $E(x)$, $M(x)$, $N(x)$ which are  introduced in Section 2 of Chapter 11 of~\cite{olver} and will be used throughout the paper. We first define the function $E(x)$ as follows:
		\begin{equation}
			E(x) = \left\{\begin{array}{ll}
				\Big(\frac{\Bi(x)}{\Ai(x)}\Big)^{\frac 12} & \text{for } x \geq c,\\
				1 & \text{for } x \leq c.
			\end{array} \right.
		\end{equation}
		Here, $c \approx -0.37$ is the negative root  of the smallest absolute value of the equation $\Ai(x) = \Bi(x)$. We then pick $M(x)$ and $N(x)$ such that
		\begin{equation}
			M(x)= \left\{\begin{array}{ll}
				\big(\Ai(x)\Bi(x)\big)^{\frac 12} & \text{for } x \geq c,\\
				\big(\Ai^2(x)+\Bi^2(x)\big)^{\frac 12} & \text{for } x \leq c,
			\end{array} \right. \qquad    N(x) =  \left\{\begin{array}{ll}\left(\frac{(Ai'(x))^2Bi^2(x) + (Bi'(x))^2Ai^2(x)}{Ai(x)Bi(x)}\right)^{1/2} & \text{for } x \geq c,\\
				\left((Ai'(x))^2 + (Bi'(x))^2\right)^{1/2} & \text{for } x \leq c.
			\end{array} \right.
		\end{equation}
		
		The functions $E$, $M$, and $N$ are nowhere vanishing and we have the following asymptotics:
		\begin{equation}\label{EMN.asymp}
			\begin{aligned}
				&E(x) \sim 2^{\frac 12} \exp\big(\frac 23 x^{3/2} \big), \qquad M(x) \sim \pi^{-\frac 12} x^{- \frac 14},\qquad N(x) \sim \pi^{-\frac{1}{2}}x^{\frac{1}{4}} \qquad \text{as} \qquad x \to+ \infty,\\
				&E(x) = 1,\qquad M(x) \sim \pi^{- \frac 12} (-x)^{- \frac 14},\qquad N(x) \sim  \pi^{-\frac{1}{2}}(-x)^{\frac{1}{4}}\qquad \text{as} \qquad x \to -\infty.
			\end{aligned}
		\end{equation}
		
		In practice, the asymptotics \eqref{EMN.asymp} will be used jointly with the following calculation rules (see \cite{olver}) that we will apply consistently \begin{equation}\begin{split}
				&|Ai|(x) \leq M(x) \cdot E^{-1}(x),\\&|Bi|(x) \leq M(x) \cdot E(x),\\
				&|Ai'|(x) \leq N(x) \cdot E^{-1}(x),\\
				&|Bi'|(x) \leq N(x) \cdot E(x).\\
			\end{split}
		\end{equation}
		
		\section{Volterra equations}
		The following theorem is a mild generalization of Theorem~10.1 of Chapter 6 in~\cite{olver} to allow for a general inhomogeneous term on the right hand side.
		\begin{thm}\label{thm:volterra}
			Let $\beta \in \RR \cup\{ +\infty\}$ and $\alpha \in \RR \cup \{-\infty\}$ with $\alpha < \beta$. 
			We consider the following Volterra integral equation for an unkown function $h(\xi)$:
			\begin{equation}\label{volt}
				h(\xi) = \int_{\xi}^{\beta} K(\xi, v) \{\psi_0(v) h(v) + \psi_1(v) h'(v) \}dv + R(\xi),
			\end{equation}
			Here, the complex-valued functions $\psi_0(v)$, $\psi_1(v)$ are continuous on  $(\alpha, \beta)$, the complex-valued function $R(\xi)$  and kernel $K(\xi, v)$ and their derivatives $\rd_\xi K$, $\rd^2_\xi K$, $\rd_{\xi}R$, $\rd^2_{\xi}R$ are continuous on $(\alpha, \beta)$. We also suppose that
			\begin{enumerate}
				\item $K(\xi, \xi) = 0$, 
				\item  When $\xi \in (\alpha, \beta)$ and $v \in [\xi, \beta)$, we have, for $j \in \{0,1,2\}$:
				\begin{equation}
					|\rd^j_\xi K(\xi, v)| \leq P_j(\xi) Q(v), 
				\end{equation}
				Here, $P_j(\xi)$ and $Q(v)$ are continuous, positive real-valued functions.
				\item There exists a positive decreasing function $\Phi(\xi)$ so that for $j \in \{0,1\}$:
				\[\left|\partial_{\xi}^jR\left(\xi\right)\right| \leq P_j(\xi)\Phi(\xi).\]
				\item When $\xi \in (\alpha, \beta)$, the following expressions converge:
				\begin{equation}
					\Psi_0(\xi) = \int_{\xi}^{\beta} |\psi_0(v)| dv, \quad \Psi_1(\xi) = \int_{\xi}^{\beta} |\psi_1(v)| dv,
				\end{equation}
				and the following suprema are finite:
				\begin{equation}
					\kappa_0 =  \sup_{\xi \in (\alpha, \beta)} P_0(\xi) Q(\xi), \quad \kappa_1 = \sup_{\xi \in (\alpha, \beta)} P_1(\xi) Q(\xi).
				\end{equation}
				except in the case $\psi_1(v) = 0$, in which case $\kappa_1$ need not exist.
			\end{enumerate}
			Then, there exists a unique $C^2$ solution $h(\xi)$ to~\eqref{volt}, and we have, for $\xi \in (\alpha, \beta)$,
			\begin{equation}\label{finalvolterraestimate}
				\frac{\left|h(\xi)\right|}{ P_0(\xi)}, \frac{|\partial_{\xi}h(\xi)|}{P_1(\xi)} \leq  \Phi(\xi)\exp \Big(\kappa_0 \Psi_0(\xi) + \kappa_1 \Psi_1(\xi) \Big).
			\end{equation}
		\end{thm}
		\begin{rmk}\label{switchorder}It is straightforward to establish an alternative version of this result where all $\int_{\xi}^{\beta}$ are instead replaced by $\int_{\alpha}^{\xi}$.
		\end{rmk}
		\begin{rmk}\label{volterramodify}
			It is also useful to have a version of Theorem~\ref{thm:volterra} where we allow the functions $\psi_0$ and $\psi_1$ to be non-integrable in $\xi$ as long as this is balanced by the decay of $\Phi(\xi)$ and $R(\xi)$:
			
			We then re-define $\Psi_0$ and $\Psi_1$ by
			\[\Psi_0\left(\xi\right) = \Phi^{-1}(\xi)\int_{\xi}^{\beta}\left|\psi_0(v)\right|\Phi(v)\, dv,\qquad \Psi_1\left(\xi\right) = \Phi^{-1}(\xi)\int_{\xi}^{\beta}\left|\psi_1(v)\right|\Phi(v)\, dv,\]
			and assume that these integrals converge and that, moreover, $\sup_{\xi}\left[\kappa_0\Psi_0+\kappa_1\Psi_1\right]$ is sufficiently small. Then the conclusions of Theorem~\ref{thm:volterra} continue to hold, where we replace $\exp\left(\kappa_0\Psi_0+\kappa_1\Psi_1\right)$ by $2$.
		\end{rmk}
		
		In the next theorem, we consider the case when our  equation depends on an auxiliary parameter $\Par$.
		\begin{thm}\label{thm:volterraparam}We consider again the setting of Theorem~\ref{thm:volterra} and assume $P_0$, $P_1$, $\Psi_0$, $\Psi_1$, and $\Phi$ have already been fixed in that context. We then furthermore suppose that $K$, $R$, $\psi_0$, and $\psi_1$ depend additionally on a real parameter $\Par \in (\Par_0,\Par_1)$:
			\begin{equation}\label{withtheparamsvolt}
				h(\xi,\Par) = \int_{\xi}^{\beta} K(\xi, v,\Par) \{ \psi_0(v,\Par) h(v,\Par) + \psi_1(v,\Par) h'(v,\Par) \}dv + R(\xi,\Par).
			\end{equation}
			We now assume that the assumptions of Theorem~\ref{thm:volterra} hold uniformly for $\Par \in (\Par_0,\Par_1)$ and make the additional assumptions that $K$, $\psi_0$, $\psi_1$, and $R$ are all $C^1$ functions of $\Par$ and that there exists functions $\tilde{P}_j$ for $j \in \{0,1\}$ and $\tilde{Q}$ so that 
			\begin{enumerate}
				\item  When $\xi \in (\alpha, \beta)$ and $v \in [\xi, \beta)$, we have, for $j \in \{0,1\}$:
				\begin{equation}
					\sup_{\Par \in (\Par_0,\Par_1)}|\rd^j_\xi K(\xi, v,\Par)| \leq \tilde{P}_j(\xi) \tilde{Q}(v), 
				\end{equation}
				Here, $ \tilde{P}_j(\xi)$ and $ \tilde{Q}(v)$ are continuous, positive real-valued functions.
				\item 
				
				There exists a decreasing function $\tilde{\Phi}(\xi)$ so that for $j \in \{0,1\}$:
				\[\sup_{\Par \in (\Par_0,\Par_1)}\left[\left|\partial_\Par \partial^j_{\xi}R\right| + \left|\tilde{R}_j\left(\xi,\Par\right)\right|\right] \leq \tilde{P}_j(\xi)\tilde{\Phi}(\xi),\]
				where we define $\tilde{R}_j$ by
				\begin{align}\label{tiderj1}
					\tilde{R}_0 \doteq \exp\left(\kappa_0\Psi_0 + \kappa_1\Psi_1\right)\int_{\xi}^{\beta}\Phi\left[\left(\left|\partial_{\lambda}K\right|\left|\psi_0\right| + \left|K\right|\left|\partial_{\lambda}\psi_0\right|\right)P_0 +\left(\left|\partial_{\lambda}K\right|\left|\psi_1\right| + \left|K\right|\left|\partial_{\lambda}\psi_1\right|\right)P_1\right]\, dv,
				\end{align}
				\begin{align}\label{tilderj2}
					&\tilde{R}_1 \doteq \exp\left(\kappa_0\Psi_0 + \kappa_1\Psi_1\right)\int_{\xi}^{\beta}\Phi\left[\left(\left|\partial_{\lambda}\partial_{\xi}K\right|\left|\psi_0\right| + \left|\partial_{\xi}K\right|\left|\partial_{\lambda}\psi_0\right|\right)P_0 +\left(\left|\partial_{\lambda}\partial_{\xi}K\right|\left|\psi_1\right| + \left|\partial_{\xi}K\right|\left|\partial_{\lambda}\psi_1\right|\right)P_1\right]\, dv
					\\ \nonumber &\qquad + \exp\left(\kappa_0\Psi_0 + \kappa_1\Psi_1\right)\Phi\left(\left|\partial_{\lambda}K\right|\left(\left|\psi_0\right|P_0 + \left|\psi_1\right|P_1\right)\right).
				\end{align}
			\end{enumerate}
			Then $h$ is continuously differentiable with respect to $\Par$ and satisfies the following estimates:
			\begin{equation}\label{finalvolterraestimatewithaparam}
				\sup_{\Par\in (\Par_0,\Par_1)}  \frac{\left|\partial_\Par h(\xi,\Par)\right|}{ \tilde{P}_0(\xi)},  \sup_{\Par\in (\Par_0,\Par_1)}\frac{|\partial^2_{\Par \xi}h(\xi,\lambda)|}{\tilde{P}_1(\xi)} \leq  \tilde{\Phi}(\xi)\exp \Big(\kappa_0 \Psi_0(\xi) + \kappa_1 \Psi_1(\xi) \Big).
			\end{equation}
		\end{thm}
		\begin{proof}
			If we formally differentiate~\eqref{withtheparamsvolt} with respect to $\Par$, we end up with the following equation for $\nu\left(\xi,\Par\right) \doteq \partial_\Par h$
			\begin{equation}\label{eqnwithnu}
				\nu =  \int_{\xi}^{\beta} K(\xi, v,\Par) \{ \psi_0(v,\Par) \nu(v,\Par) + \psi_1(v,\Par) \partial_v\nu(v,\Par) \}dv + \tilde{S}(\xi,\Par),
			\end{equation}
			where
			\[\tilde{S} = \partial_\Par R + \int_{\xi}^{\beta}\left[\partial_\Par K\left(\psi_0h+\psi_1\partial_vh\right) + K\left(\partial_\Par \psi_0 h + \partial_\Par \psi_1\partial_vh\right)\right]\, dv.\]
			
			The assumptions of the theorem allow us to apply Theorem~\ref{thm:volterra} to solve for $\nu$ so that moreoever~\eqref{finalvolterraestimatewithaparam} holds with $\partial_\Par h$ replaced by $\nu$. By a mild adaption of this argument applied to the difference quotients $h^{(k)} \doteq k^{-1}\left(h(\xi,\Par+k)-h(\xi,\Par)\right)$, it is straightforward to show that $\nu = \partial_\Par h$. 
		\end{proof}
		\begin{rmk}\label{switchorderparam}We have the analogue of Remark~\ref{switchorder} in this context. More specifically, we have an alternative version of the theorem where all integrals $\int_{\xi}^{\beta}$ are by $\int_{\alpha}^{\xi}$. 
		\end{rmk}
		\begin{rmk}\label{volterramodifyparam}
			We also have the analogue of Remarks~\ref{volterramodify} in this context. More specifically, in the event that $\psi_0$ and $\psi_1$ are not integrable along $(\xi,\beta)$, we may redefine $\Psi_0$ and $\Psi_1$ by 
			\[\Psi_0(\xi) = \tilde{\Phi}^{-1}(\xi)\int_{\xi}^{\beta}\left|\psi_0(v)\right|\tilde{\Phi}(v)\, dv,\qquad \Psi_1(\xi) = \tilde{\Phi}^{-1}(\xi)\int_{\xi}^{\beta}\left|\psi_1(v)\right|\tilde{\Phi}(v)\, dv.\]
			Assuming then that $\sup_{\xi}\left[\kappa_0\Psi_0+\kappa_1\Psi_1\right]$ is sufficiently small, then in the final estimate~\eqref{finalvolterraestimatewithaparam} we may drop the $\exp\left(\kappa_0\Psi_0 + \kappa_1\Psi_1\right)$ and replace it with with $2$.  
		\end{rmk}
		\begin{rmk}\label{volterraholomorphic}Finally, a straightforward adaption of the proof allows us to consider the case when $\Par$ lies in a compact subset $\Omega$ of the complex plane, and the dependence of $K$, $\psi_0$, $\psi_1$, and $R$ on $\Par$ is holomorphic in $\Omega$. In this case, we will also obtain that $h$ is holomorphic in $\Omega$.
		\end{rmk}

		\bibliographystyle{plain}
		\bibliography{bibliography.bib}
		
	\end{document}